\documentclass[11pt,a4paper]{article}
\usepackage[T1]{fontenc}
\usepackage{lmodern}
\usepackage[margin=1in]{geometry}
\usepackage{amsmath,amssymb,amsthm,mathtools,booktabs,array,longtable,microtype,needspace,placeins}
\usepackage{pgfplots}
\pgfplotsset{compat=1.18}
\usepackage{tikz,pgfplots}
\usetikzlibrary{arrows.meta,calc,positioning,patterns,decorations.pathreplacing}
\pgfplotsset{compat=1.18}
\usepackage[hidelinks]{hyperref}
\definecolor{inkblue}{HTML}{254E70}
\definecolor{inkteal}{HTML}{26867B}
\definecolor{inkgold}{HTML}{BA7625}
\newcommand{\E}{\mathbb E}
\newcommand{\PP}{\mathbb P}
\newcommand{\ALG}{\mathrm{ALG}}
\newcommand{\ind}{\mathbf{1}}

\newcommand{\bet}{\beta^*}

\newtheorem{theorem}{Theorem}[section]
\newtheorem{lemma}[theorem]{Lemma}

\theoremstyle{definition}

\theoremstyle{remark}
\numberwithin{equation}{section}
\hypersetup{pdftitle={Prophet Inequalities with Only the Expected Maximum},pdfauthor={}}

\title{Mean-only Prophet Inequalities}
\author{Yilong Feng \and Yaonan Jin \and Xiaowei Wu \and Xinyue Xiang}
\date{\today}

\begin{document}

\maketitle

\begin{abstract}
We study mean-based algorithms for single-choice prophet inequalities: the algorithm knows the number of variables and the mean of their maximum, but not their distributions. For arbitrary independent variables in a fixed order, we can already achieve the optimal $1/2$ competitive ratio using only this mean. For i.i.d. variables and prophet secretary, however, we previously knew only a $1-1/e$ guarantee, which is far away from the ratios with full distribution information. We improve the guarantees in both settings using a simple and natural algorithm: we fix a continuous decreasing threshold function that depends only on the mean and accept the first variable exceeding its threshold. For i.i.d. variables, we prove a $0.7419$ competitive ratio. For the prophet secretary, we prove a $0.6902$ competitive ratio.
\end{abstract}

\setcounter{tocdepth}{2}
\renewcommand{\contentsname}{Contents}
\tableofcontents

\newpage
\section{Introduction}\label{sec:introduction}

A decision maker observes a sequence of values $X_1, X_2,\ldots, X_n$ and can accept at most one of them, where the distribution of each value $X_i$ is known, but not the realizations.
Accepting a value ends the process; rejecting it removes that option forever. How much can an online algorithm obtain compared with a decision maker who sees the entire sequence? Prophet inequalities answer this question when the values are independent random variables. The benchmark is the expected maximum $\E\!\left[X^*\right]$, where $X^* = \max\!\left\{X_1,\ldots,X_n\right\}$.
As shown by Krengel and Sucheston~\cite{journals/bams/KrengelS77}, for arbitrary independent nonnegative variables in a fixed order, we can achieve half of the expected maximum, and no algorithm can guarantee more. We can attain this optimal guarantee using only the expected maximum: we accept the first value that exceeds $\frac{1}{2}\cdot \E[X^*]$~\cite{DBLP:journals/geb/KleinbergW19}. Samuel-Cahn~\cite{journals/aop/SamuelCahn84} also proved the guarantee using a single threshold. Better guarantees become possible when the values are identically distributed (aka. i.i.d. prophet) or arrive in uniformly random order (aka. prophet secretary). Both settings are central models for online selection and pricing, and in both we can guarantee a ratio of $1-1/e$ using only the expected maximum~\cite{journals/siamdm/EsfandiariHLM17,conf/soda/EhsaniHKS18}.

With full distribution information, the optimal competitive ratio for i.i.d. prophet inequalities is $\beta^* \approx 0.7454$~\cite{conf/sigecom/CorreaFHOV17}, and we can obtain at least $0.688$ for prophet secretary~\cite{conf/soda/ChenHLT25}. We therefore see a substantial gap from the previously known $1-1/e$ guarantee using only $\E[X^*]$. A natural question to ask is how much we can close this gap without additional distribution information. From the expected maximum, we learn the scale of a good outcome. We cannot infer a quantile, the probability of seeing a particular value, or the expected benefit of continuing after a rejection. 
Existing work~\cite{conf/soda/EhsaniHKS18} on mean-only prophet secretary uses a natural non-adaptive (i.e., fixed before any values are observed) threshold-based algorithm as follows.
Assume that $\E\!\left[X^*\right]=1$ and each variable has an independent arrival time chosen uniformly at random in $\left[0,1\right]$.
They set a decreasing threshold function $\theta(t) = 1-e^{t-1}$ for $t\in \left[0,1\right]$.
The function is concave and decreases from $\theta(0)=1-1/e$ to $\theta\left(1\right)=0$. 
The algorithm then accepts the first variable exceeding the threshold at its arrival time.
The choice of the threshold function allows an elegant analysis that leads to a competitive ratio of $1-1/e$.

It is thus natural to ask whether the choice of $1-e^{t-1}$ is optimal.
If not, how much better can we do with other choices of threshold functions?

\subsection{Our Results}

In this work, we show that the choice of $1-e^{t-1}$ is not optimal.
Perhaps surprisingly, we show that substantially stronger competitive ratios are obtainable.

We first show that for the i.i.d. setting, using only $\E\!\left[X^*\right]$, it is possible to achieve a competitive ratio that is close to the optimal ratio $\beta^*$ using full distribution information.

\begin{theorem}\label{thm:iid}
There is a non-adaptive algorithm with a continuous concave decreasing threshold function for i.i.d. variables such that, for every $n$,
\begin{equation*}
 \E\!\left[\ALG\right]\ge 0.7419\cdot \E\!\left[X^*\right].
\end{equation*}
\end{theorem}

We also show that with a simple linear threshold of $\theta(t) = 1-t$, we can already achieve a ratio of at least $0.6906$, and the proof is only two pages long.
We accompany our positive results with two upper bounds.
We first show that every non-adaptive algorithm with a pre-committed $\theta$ cannot achieve a ratio strictly higher than $0.74197$, almost matching our competitive ratio.
We further establish a strict upper bound $\beta^*-\varepsilon$, for some constant $\varepsilon>0$, for every randomized, history-dependent (i.e., adaptive) mean-only algorithm. 

\smallskip

Then we consider the prophet secretary setting.
We show that using only the expected maximum, we can achieve a competitive ratio of at least $0.6902$, which beats the state-of-the-art ratio of $0.688$ using full distribution information by Chen, Huang, Li, and Tang~\cite{conf/soda/ChenHLT25}.

\begin{theorem}\label{thm:secretary}
There is a non-adaptive algorithm with a continuous concave decreasing threshold function for prophet secretary such that, for every $n$,
\begin{equation*}
 \E\!\left[\ALG\right]\ge 0.6902\cdot \E\!\left[X^*\right].
\end{equation*}
\end{theorem}

We also provide a simple (one-page) analysis showing that a two-segment piecewise-linear theta gives a ratio of at least $0.6385$, beating the classic $1-1/e$ ratio by~\cite{conf/soda/EhsaniHKS18}.
We then show that no deterministic precommitted threshold function guarantees an asymptotic ratio above $0.694451$.

\smallskip

We plot the threshold functions in Figure~\ref{fig:mean-only-thresholds}.
We also summarize and compare other guarantees under different information models in Table~\ref{tab:information-ratios}.

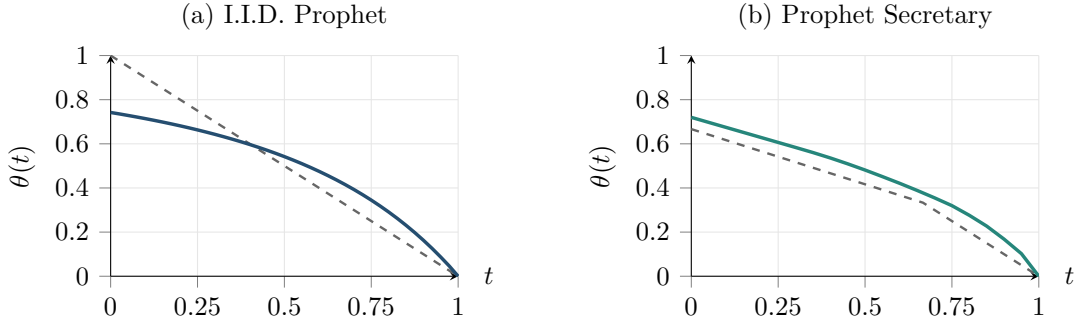
\begin{figure}[!htbp]
\centering
\begin{minipage}[t]{0.4\textwidth}
\centering
\begin{tikzpicture}
\begin{axis}[
    width=0.97\linewidth, height=4.5cm,
    xmin=0, xmax=1, ymin=0, ymax=1,
    xlabel={$t$}, ylabel={$\theta(t)$},
    xlabel style={at={(axis description cs:1,0)}, anchor=west, xshift=6pt},
    title={(a) I.I.D. Prophet},
    title style={font=\small},
    label style={font=\small}, tick label style={font=\small},
    xtick={0,0.25,0.5,0.75,1}, ytick={0,0.2,0.4,0.6,0.8,1},
    axis lines=left, tick align=outside,
    grid=major, grid style={black!10},
    scaled ticks=false,
]
\addplot[black!60, dashed, line width=0.9pt] coordinates {
    (0,1) (1,0)
};
\addplot[inkblue, line width=1.3pt] coordinates {
    (0,0.7419615683922162)
    (0.025,0.7355106076020216)
    (0.05,0.7285923252134327)
    (0.075,0.7214991111870098)
    (0.1,0.7141229394025917)
    (0.125,0.7064524097941529)
    (0.15,0.6984679273206473)
    (0.175,0.6901566241840774)
    (0.2,0.6814964681135451)
    (0.225,0.6724722944983229)
    (0.25,0.663059893119096)
    (0.275,0.6532408845494443)
    (0.3,0.6429886465201203)
    (0.325,0.6322814957447677)
    (0.35,0.6210891297610388)
    (0.375,0.6093871106314772)
    (0.4,0.5971399876309874)
    (0.425,0.5843200948203213)
    (0.45,0.5708868082460946)
    (0.475,0.5568075075455743)
    (0.5,0.5420367564877159)
    (0.525,0.5265350354412825)
    (0.55,0.5102515358182379)
    (0.575,0.4931393341365643)
    (0.6,0.4751397160306222)
    (0.625,0.4561992466825068)
    (0.65,0.4362479300257943)
    (0.675,0.4152248016224289)
    (0.7,0.3930481712746234)
    (0.725,0.3696458018462857)
    (0.75,0.3449246462975888)
    (0.775,0.3187968757277445)
    (0.8,0.2911566942060563)
    (0.825,0.2618988586517287)
    (0.85,0.2308995719332495)
    (0.875,0.1980369876457899)
    (0.9,0.1631620635352875)
    (0.925,0.1261339503672213)
    (0.95,0.0867759300198585)
    (0.975,0.0453829804826449)
    (1,0)
};
\end{axis}
\end{tikzpicture}
\end{minipage}
\hspace{30pt}
\begin{minipage}[t]{0.4\textwidth}
\centering
\begin{tikzpicture}
\begin{axis}[
    width=0.97\linewidth, height=4.5cm,
    xmin=0, xmax=1, ymin=0, ymax=1,
    xlabel={$t$}, ylabel={$\theta(t)$},
    xlabel style={at={(axis description cs:1,0)}, anchor=west, xshift=6pt},
    title={(b) Prophet Secretary},
    title style={font=\small},
    label style={font=\small}, tick label style={font=\small},
    xtick={0,0.25,0.5,0.75,1}, ytick={0,0.2,0.4,0.6,0.8,1},
    axis lines=left, tick align=outside,
    grid=major, grid style={black!10},
    scaled ticks=false,
]
\addplot[black!60, dashed, line width=0.9pt] coordinates {
    (0,{2/3}) ({2/3},{1/3}) (1,0)
};
\addplot[inkteal, line width=1.3pt] coordinates {
    (0,0.71987853048)
    (0.05,0.697157345745)
    (0.1,0.67443616101)
    (0.15,0.651714976275)
    (0.2,0.62899379154)
    (0.25,0.606272606805)
    (0.3,0.58355142207)
    (0.35,0.5605213421)
    (0.4,0.535934859795)
    (0.45,0.50937548407)
    (0.5,0.48123512577)
    (0.55,0.451412318035)
    (0.6,0.420814514945)
    (0.65,0.388682900985)
    (0.7,0.35518907771)
    (0.75,0.319841642715)
    (0.8,0.276142125695)
    (0.85,0.227177521185)
    (0.9,0.168117391995)
    (0.95,0.102213400975)
    (1,0)
};
\end{axis}
\end{tikzpicture}
\end{minipage}
\caption{The threshold functions for the i.i.d. and prophet secretary settings, normalized by $\E\!\left[X^*\right]=1$. We plot the thresholds for the main results in solid colors and the simple thresholds in dashed gray. We remark that the choices of one- and two-segment linear functions are not optimized. We choose the particular functions mainly for convenience of analysis and exposition.}
\label{fig:mean-only-thresholds}
\end{figure}

% \vspace{-40pt}

% Information models and selected guarantees; all ratios concern expected reward.
\begin{table}[htp]
\centering
\caption{Selected guarantees with limited distribution information.}
\label{tab:information-ratios}
\begingroup
\footnotesize
\setlength{\tabcolsep}{4pt}
\renewcommand{\arraystretch}{1.6}
\begin{tabular}{@{}>{\raggedright\arraybackslash}p{0.27\textwidth}
                    >{\centering\arraybackslash}p{0.17\textwidth}
                    >{\centering\arraybackslash}p{0.18\textwidth}
                    >{\raggedright\arraybackslash}p{0.34\textwidth}@{}}
\toprule
Information & I.I.D. Prophet & Prophet Secretary & Remark \\
\midrule
$o(n)$ samples & $1/e$ & $1/e$
 & Adaptive~\cite{conf/sigecom/CorreaDFS19}. \\
$\rho \cdot n$ samples, $\rho\le 1/\left(e-1\right)$ & $\left(1+\rho\right)/e$ & ---
 & Adaptive~\cite{conf/innovations/CorreaDFSZ21}. \\
One sample per distribution & $0.671$ & $0.63518$
 & Adaptive~\cite{journals/mor/CorreaCES24,conf/soda/CorreaCES20}. \\
$O\!\left(n/\varepsilon\right)$ samples in total & $\beta^*-\varepsilon$ & ---
 & Adaptive; $n\to\infty$~\cite{journals/mor/CorreaCES24}. \\
$O_\varepsilon\!\left(1\right)$ samples per distribution
 & $\beta^*-\varepsilon$ & $\beta^*_{\mathrm{PS}}-\varepsilon$
 & I.I.D.:~\cite{conf/innovations/RubinsteinWW20}; secretary:~\cite{conf/stoc/CristiZ24}. \\
\midrule
One $1/e$-quantile of $X^*$ & $1-1/e$ & $1-1/e$
 & Single threshold~\cite{conf/soda/EhsaniHKS18}. \\
The distribution of $X^*$ & $\beta^*$ & $0.6724$
 & Non-adaptive~\cite{conf/sigecom/CorreaFHOV17,conf/soda/Harb25}. \\
One quantile query
 & $0.6862$ & ---
 & Adaptive; $n\to\infty$~\cite{journals/ai/FengLLWW25}. \\
$2\sim 10$ quantile queries
 & $0.7080 \sim 0.7428$ & ---
 & Non-adaptive~\cite{journals/mor/PerezSalazarST26}. \\
\midrule
$\E\!\left[X^*\right]$ & $1-1/e$ & $1-1/e$
 & Non-adaptive~\cite{conf/soda/EhsaniHKS18}. \\
Priors with error at most $\delta$
 & $\beta^*-\frac{5\cdot n\cdot \delta}{\E\!\left[X^*\right]}$ & $0.6724-\frac{5\cdot n\cdot \delta}{\E\!\left[X^*\right]}$
 & Non-adaptive; values in $\left[0,1\right]$~\cite{conf/sigecom/DuttingK19}. \\
\midrule
Full distributions & $\beta^*\approx0.7454$ & $0.688$
 & Non-adaptive~\cite{conf/sigecom/CorreaFHOV17,conf/soda/ChenHLT25}. \\
\textbf{$\E\!\left[X^*\right]$}
 & $\boldsymbol{0.7419}$ & $\boldsymbol{0.6902}$
 & \textbf{Non-adaptive. This work.} \\
\bottomrule
\end{tabular}
\par\smallskip
\begin{minipage}{\textwidth}
\footnotesize
We write $\beta^*_{\mathrm{PS}}$ for the optimal full-information prophet secretary ratio, whose exact value is unknown.
Here, adaptive/non-adaptive refers to whether the acceptance rule depends on previously observed online values.
\end{minipage}
\endgroup
\end{table}

\subsection{Our Techniques}

The key to our algorithms, in both the i.i.d. and prophet secretary settings, is to design a threshold function that works well against all distributions with the same expected maximum (we assume that $\E\!\left[X^*\right]=1$). We explain our approach through the i.i.d. setting.

\paragraph{Reduction to a Max--min Program.}
We describe the threshold function $\theta$ by its drop rate $b(t)=-\theta'(t)$. Since we seek a decreasing concave threshold, we choose $b$ to be positive and nondecreasing. Intuitively, we lower the threshold slowly at the beginning, when we can afford to wait, and more quickly near the end, when fewer opportunities remain. The main difficulty is deciding how quickly to lower it without knowing the distributions. By analyzing the expected reward of the algorithm, we obtain a lower bound $L\!\left(b,q\right)$ involving two functions of continuous time. We control $b$, while $q$ encodes the unknown distribution. We then require our guarantee to hold for every admissible $q$, so that it applies to every distribution with this expected maximum. Thus, we choose $b$ to maximize the worst-case lower bound:
\begin{equation*}
 \max_b\left\{\inf_q\left\{L\!\left(b,q\right)\right\}\right\}.
\end{equation*}
For each fixed $q$, we obtain a linear constraint on $b$. We therefore face a program with infinitely many constraints, one for each possible $q$. Our main analysis centers on how to represent these constraints concisely and obtain a good feasible solution. To illustrate the idea, consider the linear threshold $\theta(t)=1-t$, for which $b(t)=1$. Although we still need to account for infinitely many $q$, we can bound their effects using one continuous function of a single variable. We then minimize this one-dimensional function using elementary calculations and numerical evaluation, establishing a competitive ratio of $0.6906$.

\paragraph{Characterization of Infinitely Many Distributions via Factor-revealing LP.}
For a general drop rate $b$, we can no longer summarize the effect of $q$ using only one variable. We address this difficulty through discretization. We first restrict $b$ to a step function, which corresponds to a piecewise linear $\theta$. We partition the time interval $\left[0,1\right]$ into finitely many intervals and cover the possible values of $q$ by finitely many intervals, including one extending to infinity. For a fixed $b$, we introduce functions $V_t\!\left(y,x\right)$ that lower-bound the remaining contribution to a reward bound of the algorithm from time $t$. Informally, we use $x$ and $y$ to summarize what has already accumulated and what remains. 
In particular, the functions $V_t$ account for all admissible $q$, so they do not depend on a particular unknown distribution. We obtain
\begin{equation*}
 L\!\left(b,q\right)\ge\theta(0)+V_0\!\left(1,0\right)
 \qquad\text{for every admissible }q.
\end{equation*}
We further discretize time $t$ and the two state coordinates $y$ and $x$. Since $V_t$ is concave in its state coordinates, we can lower-bound its value between grid points by weighted averages of the values at nearby vertices. We introduce LP variables for lower bounds at these vertices and connect consecutive times using our bounds for each interval of acceptance rates. With a sufficiently fine discretization, we certify the ratio $0.7419$. It remains to find good drop rates $b$. We use a heuristic search by constraint enumeration: we start with candidate rates $q$, optimize $b$ against them (this is possible because the max--min program introduced above is linear in $b$), and then search for additional rates $q$ that give a smaller reward bound. After choosing $b$, we use the factor-revealing LP to certify its guarantee against every admissible $q$.

\subsection{Related Work}

\paragraph{I.i.d. selection and random arrival.}
Hill and Kertz~\cite{journals/aop/HillK82} characterized the sharp finite-horizon competitive ratios for i.i.d. prophet inequalities, and Kertz~\cite{journals/jmva/Kertz86} determined the corresponding asymptotic constant $\beta^*$. \cite{conf/stoc/AbolhassaniEEHKL17} improved the classical $1-1/e$ guarantee to $0.738$ for sufficiently large $n$. Independently, \cite{conf/sigecom/CorreaFHOV17} achieved the optimal ratio $\beta^*$ for every $n$ using distribution-dependent quantile thresholds.
For prophet secretary, \cite{journals/siamdm/EsfandiariHLM17} established an asymptotic $1-1/e$ guarantee using only $\E\!\left[X^*\right]$. With known distributions, Azar, Chiplunkar, and Kaplan~\cite{conf/sigecom/AzarCK18} were the first to surpass the $1-1/e$ barrier. The guarantee was subsequently improved to $0.669$~\cite{conf/soda/CorreaSZ19}, $0.6724$~\cite{conf/soda/Harb25}, and $0.688$~\cite{conf/soda/ChenHLT25}. %Our results improve the competitive ratios achievable with knowledge of $\E[X^*]$ alone.

\paragraph{Arrival-order models and information.}
Arsenis, Drosis, and Kleinberg~\cite{conf/soda/ArsenisDK21} studied constrained-order prophet inequalities, where the algorithm chooses an arrival order from a prescribed set of permutations. For unrestricted order selection, Peng and Tang~\cite{conf/focs/PengT22} obtained a ratio of approximately $0.7251$. Bubna and Chiplunkar~\cite{conf/sigecom/BubnaC23} improved this to $0.7258$ and established a strict separation from random arrivals. Giambartolomei, Mallmann-Trenn, and Saona~\cite{journals/corr/GiambartolomeiMS23} proved a $0.7235$ upper bound for order-unaware random arrivals. Separately, \cite{conf/soda/EzraFGT23} studied algorithms without advance order information against an optimal online algorithm with advance knowledge of that order. 

\paragraph{Limited distributional information.}
For unknown i.i.d. distributions, \cite{conf/sigecom/CorreaDFS19} showed that the asymptotic competitive ratio cannot exceed $1/e$ with $o(n)$ total offline samples. Rubinstein, Wang, and Weinberg~\cite{conf/innovations/RubinsteinWW20} showed that $O_\varepsilon(n)$ total samples suffice to approach $\beta^*$ within any fixed accuracy $\varepsilon>0$. \cite{conf/innovations/CorreaDFSZ21} improved guarantees for linear sample budgets and developed streaming algorithms. Related work connected sample-based selection to the two-sided game of Googol~\cite{conf/soda/CorreaCES20} and studied rank-based optimal stopping with independent sampling~\cite{journals/mor/CorreaCES24}. Cristi and Ziliotto~\cite{conf/stoc/CristiZ24} showed that $O_\varepsilon(1)$ samples per distribution suffice to approach the optimal full-information guarantees within $\varepsilon$ in the i.i.d., prophet secretary, and free-order settings.
\cite{journals/ai/FengLLWW25} obtained an asymptotic competitive ratio of $0.6862$ given a single suitable value--quantile pair for the common distribution. 
\cite{journals/mor/PerezSalazarST26} showed that strictly higher ratios are possible with more (but still constant) quantile queries.

\paragraph{Pricing and sufficient conditions.}
Revenue--utility decomposition is a standard tool in pricing-based analyses of prophet inequalities, including balanced-threshold arguments~\cite{conf/stoc/KleinbergW12} and posted-pricing techniques~\cite{conf/soda/FeldmanGL15}. 
\cite{conf/focs/DuttingFKL17} developed a general balanced-pricing framework that derives prophet guarantees from sufficient conditions on prices. Further work extends prophet inequalities to subadditive combinatorial auctions~\cite{conf/focs/DuttingKL20,conf/stoc/CorreaC23}, general downward-closed constraints~\cite{conf/stoc/Rubinstein16}, and combinatorial valuation functions~\cite{conf/soda/RubinsteinS17}. Online contention resolution schemes provide another framework for combinatorial feasibility constraints~\cite{conf/soda/FeldmanSZ16}.

\section{Preliminaries}\label{sec:preliminaries}

Let $X_1,\ldots,X_n$ be independent nonnegative random variables. Write $X^*=\max_{1\le i\le n}\left\{X_i\right\}$.
By scaling all variables, we can assume without loss of generality that $\E\!\left[X^*\right]=1$.
The prophet observes all realizations and receives $X^*$. The online algorithm may accept at most one value, and must decide whether to accept a value when it arrives. Its reward is $\ALG$, with reward zero if it accepts nothing. An algorithm is $\gamma$-competitive if $\E\!\left[\ALG\right]\ge\gamma$ on every admissible input.
The algorithm knows only $n$ and $\E\!\left[X^*\right]$. In the \emph{i.i.d. setting}, the variables have the same unknown distribution. In the \emph{prophet secretary setting}, their unknown distributions may differ, and their arrival order is a uniformly random permutation independent of the values. 
For the analysis, we represent the random order of variables by independent arrival times $t_i$ drawn uniformly from $\left[0,1\right]$ for variable $X_i$, independent of other variables. 

A precommitted threshold algorithm chooses a decreasing function $\theta$ before seeing any variables and accepts the first value satisfying $X_i > \theta\left(t_i\right)$.
Throughout this paper (for both positive and negative results), we restrict our attention to natural threshold functions, which are continuous and strictly decrease to zero. Write
\begin{equation*}
 b(t)=-\theta'(t)\ge0,\qquad
 \theta(0)=\int_0^1 b(t)\,dt \leq 1,
 \qquad \theta(t)=\int_t^1 b(s)\,ds.
\end{equation*}
The \emph{drop rate} $b$ describes how the algorithm relaxes its acceptance requirement. A nondecreasing $b$ makes $\theta$ concave. The threshold functions in Theorems~\ref{thm:iid} and~\ref{thm:secretary} use a positive nondecreasing step function $b$, so they are continuous, piecewise linear, and concave. 

Figure~\ref{fig:threshold-algorithm} illustrates the threshold rule;
$\tau$ denotes the time at which a value is accepted.

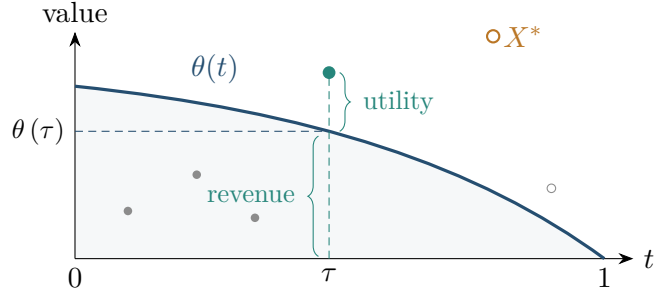
\begin{figure}[!htbp]
\centering
\begin{tikzpicture}[x=7cm,y=3cm,>={Stealth[length=2mm]}]
\fill[inkblue!4] (0,0)--(0,.76) .. controls (.45,.65) and (.76,.41) .. (1,0)--cycle;
\draw[->] (0,0)--(1.055,0) node[right] {$t$};
\draw[->] (0,0)--(0,1) node[above] {value};
\draw[very thick,inkblue] (0,.76) .. controls (.45,.65) and (.76,.41) .. (1,0);
\node[inkblue,anchor=south west] at (.20,.73) {$\theta(t)$};
\foreach \x/\y in {.10/.21,.23/.37,.34/.18} {\fill[black!45] (\x,\y) circle[radius=1.6pt];}
\draw[densely dashed,inkteal] (.48,0)--(.48,.82);
\fill[inkteal] (.48,.82) circle[radius=2.4pt];
\node[below] at (.48,0) {$\tau$};
% \draw[thick,inkteal] (.495,0.56131886)--(.495,.82);
\draw[decorate,decoration={brace,mirror,amplitude=4pt},inkteal] (.5,0.56131886)--(.5,.82) node[midway,right=5pt,font=\small] {utility};
\draw[decorate,decoration={brace,amplitude=4pt},inkteal] (.46,0.01)--(.46,.54) node[midway,left=5pt,font=\small] {revenue};
\draw[densely dashed,inkblue] (0,0.56131886)--(.48,0.56131886);
\node[left,font=\small] at (0,0.56131886) {$\theta\left(\tau\right)$};
\filldraw[draw=inkgold,fill=white,thick] (.79,0.98) circle[radius=2.3pt];
\node[right,inkgold] at (.79,0.98) {$X^*$};
\filldraw[draw=black!45,fill=white] (.90,.31) circle[radius=1.6pt];
\node[below] at (0,0) {$0$};\node[below] at (1,0) {$1$};
\end{tikzpicture}
\caption{An illustration: Filled gray variables are rejected, and a variable is accepted at $\tau$.}
\label{fig:threshold-algorithm}
\end{figure}

We write $(x)^+=\max\!\left\{x,0\right\}$ and $\mathbf{1}\left(E\right)$ for the indicator of an event $E$.

\section{I.I.D. Variables}
\label{sec:iid-analysis}

In the i.i.d. setting, all variables have the same probability of exceeding the threshold and the same conditional expected value.
This allows us to give a cleaner characterization of $\E\!\left[\ALG\right]$.
We establish a ratio of $0.7419$ by characterizing the expected reward.

\subsection{Characterizing the Expected Reward}
\label{sec:iid-reduction}

We first analyze the case when $n\to \infty$.
In this subsection, $X^*$ denotes their maximum, set to zero if there are no observations, and we normalize $\E\!\left[X^*\right]=1$.
Section~\ref{sec:iid-transfer} transfers the resulting lower bound to every fixed number of observations.

For i.i.d. variables and fixed threshold $\theta$, define
\begin{equation*}
    q(t) = n\cdot\PP\!\left(X_1>\theta(t)\right),
\end{equation*}
which denotes the expected number of values exceeding $\theta(t)$ over the whole time horizon.
The function $q$ is nondecreasing.
Recall that we use $\tau$ to denote the stopping time, and set $\tau=1$ if no variable is accepted.
Define
\begin{equation*}
 S(t)=\PP\!\left(\tau>t\right)=\exp\left(-\int_0^t q(s)\,ds\right)
 \qquad \forall t\in [0,1).
\end{equation*}
We further define $S\left(1\right)$ as the probability of no acceptance. Define
\begin{equation*}
 A(t)=\int_t^1S(s)\,ds=\E\!\left[\left(\tau-t\right)^+\right].
\end{equation*}
Note that $A(t)$ measures the area below the survival function after time $t$, and $A(0)=\E[\tau]$.

We characterize the reward $\E[\ALG]$.
A variable arriving at time $t$ is accepted when its value exceeds $\theta(t)$ and no earlier arrival was accepted.
Since the variables are independent, we have
\begin{align*}
 \E\!\left[\ALG\right]
 &= \int_0^1 S(t)\cdot n\cdot \E\!\left[X_1\cdot\mathbf{1}\left(X_1>\theta(t)\right)\right]dt\\
 &=\int_0^1S(t)\cdot\left(
      \theta(t)\cdot q(t)+n\cdot\int_{\theta(t)}^\infty
        \PP\!\left(X_1>x\right)\,dx\right)\,dt\\
 &=\int_0^1S(t)\cdot\left(
      \theta(t)\cdot q(t) + n\cdot\int_{\theta(t)}^{\theta(0)}
        \PP\!\left(X_1>x\right)\,dx + n\cdot\int_{\theta(0)}^\infty
        \PP\!\left(X_1>x\right)\,dx\right)\,dt\\
 &=\int_0^1 S(t)\cdot\left(
      \theta(t)\cdot q(t)+\int_0^t b(s)\cdot q(s)\,ds
      +n\cdot\E\!\left[ \left( X_1-\theta(0)\right)^+\right]\right)\,dt,
\end{align*}
where the last equality follows by substituting $x=\theta(s)$.
We evaluate the first two terms and lower-bound the third term.
For the first term:
\begin{align*}
 \int_0^1S(t)\cdot\theta(t)\cdot q(t)\,dt
 & =-\int_0^1\theta(t)\,dS(t) \\
 & =-\left[\theta(t)\cdot S(t)\right]_0^1 + \int_0^1S(t)\,d\theta(t)
 =\theta(0)-\int_0^1b(t)\cdot S(t)\,dt,
\end{align*}
where we use $S'=-q\cdot S$, $S(0)=1$, $\theta\left(1\right)=0$, and $\theta'=-b$.
For the second term:
\begin{align*}
 \int_0^1S(t)\cdot\int_0^t b(s)\cdot q(s)\,ds\,dt
 =\int_0^1b(s)\cdot q(s)\cdot\int_s^1S(t)\,dt\,ds
 =\int_0^1b(t)\cdot q(t)\cdot A(t)\,dt.
\end{align*}

Finally, using a union bound and $\E\!\left[X^*\right]=1$, we obtain
\begin{align*}
 n\cdot\E\!\left[\left(X_1-\theta(0)\right)^+\right]
 &=\int_{\theta(0)}^\infty n\cdot\PP\!\left(X_1>x\right)\,dx\\
 &\ge\int_{\theta(0)}^\infty\PP\!\left(X^*>x\right)\,dx\\
 &=1-\theta(0)
   +\int_0^{\theta(0)}\PP\!\left(X^*\le x\right)\,dx\\
 &=1-\theta(0)+\int_0^1b(t)\cdot e^{-q(t)}\,dt,
\end{align*}
where the last equality substitutes $x=\theta(t)$ and uses $\PP\!\left[ X^* \leq \theta(t) \right] = (1-q(t)/n)^n = e^{-q(t)}$.
Putting these together, we obtain the following lower bound on $\E\!\left[\ALG\right]$:
\begin{align*}
 & \theta(0)-\int_0^1 b(t)\cdot\left(S(t)-q(t)\cdot A(t)\right)\,dt
   +A(0)\cdot \left( 1-\theta(0)+\int_0^1b(t)\cdot e^{-q(t)}\,dt \right) \\
=\ & \theta(0) + \left(1-\theta(0)\right)\cdot A(0)
    -\int_0^1b(t)\cdot \left( S(t)-q(t)\cdot A(t)-A(0)\cdot e^{-q(t)} \right)\,dt.
\end{align*}

For ease of notation, we define
\begin{align*}
 k(t) & =S(t)-q(t)\cdot A(t)-A(0)\cdot e^{-q(t)}, \quad \text{ and } \\
 L & =\theta(0)+\left(1-\theta(0)\right)\cdot A(0)-\int_0^1b(t)\cdot k(t)\,dt.
\end{align*}

It is easy to verify that $0\le k(t)\le S(t)\le1$ (see Lemma~\ref{lem:k-bounds} for a proof).
We therefore obtain the following lemma immediately, by the definition of $k$ and $L$.

\begin{lemma}\label{lem:iid-function}
The expected reward satisfies $\E\!\left[\ALG\right]\ge L$.
\end{lemma}

We have reduced lower bounding the competitive ratio to the following max--min program.
The variables are $L$ and $b(t)$.
We choose the drop rate $b$ to maximize the guaranteed value of $L$ against every admissible $q$, which encodes the unknown input.
\begin{align*}
 \text{maximize}\quad &L\\
 \text{subject to}\quad
 &b(t)\ge0,\qquad\int_0^1b(t)\,dt\le1,\\
 &L\le\theta(0)+\left(1-\theta(0)\right)\cdot A(0)
       -\int_0^1b(t)\cdot k(t)\,dt
       \qquad\forall q.
\end{align*}
Here an admissible $q$ is any nonnegative nondecreasing rate.
Each such $q$ supplies a constraint, with $S$, $A$, and $k$ computed from that rate.
Every distribution with $\E\!\left[X^*\right]=1$ supplies an admissible $q$, so a lower bound valid for all these rates also holds for every such input.

To approximate a continuous drop rate $b$, we partition the time horizon into $N$ intervals of length $1/N$ and make $b(t)$ constant on each interval.
The resulting threshold curve is piecewise linear.
This discretization approximates the unrestricted max--min value.
Indeed, writing $L\!\left(b,q\right)$ for $L$ with fixed $b$ and $q$, we have the uniform Lipschitz estimate:
\begin{align*}
 \left|L\!\left(b,q\right)-L\!\left(\widehat b,q\right)\right|
 =\left|\int_0^1
     \left(1-A(0)-k(t)\right)\cdot\left(b(t)-\widehat b(t)\right)\,dt
   \right|
 \le\int_0^1\left|b(t)-\widehat b(t)\right|\,dt.
\end{align*}
The inequality uses $0\le A(0)\le1$ and $0\le k(t)\le1$.
Consequently, the optimal values of the discretized programs converge to the unrestricted value as $N\to\infty$.

\subsection{Linear Threshold}
\label{sec:iid-linear}

In this section, we consider the linear threshold $\theta(t)=1-t$, which corresponds to the
coarsest discretization of $N=1$: one interval with $b(t)=1$ and
$\theta(0)=1$.
This serves as a warm-up and establishes a competitive ratio $0.6906$.
By the definition of $L$, for $\theta(t)=1-t$, we have
\begin{align*}
 L
 &=1-\int_0^1k(t)\,dt\\
 &=1-A(0)+\int_0^1 q(t)\cdot A(t)\,dt+A(0)\cdot\int_0^1e^{-q(t)}\,dt\\
 &=1-A(0)+\int_0^1 q(t)\cdot \int_t^1 S(s)\,ds  dt+A(0)\cdot\int_0^1e^{-q(t)}\,dt\\
 &=1-A(0)+\int_0^1 \left( S(s)\cdot \int_0^s q(t) dt \right) ds +A(0)\cdot\int_0^1e^{-q(t)}\,dt.
\end{align*} 
We now focus on lower-bounding $L$. Its value is determined by the unknown distribution, characterized by the acceptance rate $q$. We show that, for every $q$, we can lower-bound $L$ by a function $\tilde{L}(a)$ of the single variable $a=A(0)=\E\!\left[\tau\right]$, and then minimize over $a$. The main step is to remove $q(t)$ from the integral. We do this by finding the pointwise worst possible rate.
% at each fixed survival coordinate $x(t) :=-\log S(t)=\int_0^t q(s)\,ds$. Since $d x(t)= q(t)\,dt$, we will minimize the contribution per unit increase in $x$, rather than the contribution per unit of time.

\begin{lemma}\label{lem:iid-linear}
There is a convex function $\tilde{L}:\left[0,1\right]\to\left[0,1\right]$ such that, for every nonnegative and nondecreasing acceptance rate $q$, we have $L\ge\tilde{L}\left(A(0)\right)$.
\end{lemma}
\begin{proof}
Fix any $q$. Let $a=A(0)=\int_0^1 S(t)\,dt\in\left(0,1\right]$ and $x(t)=-\log S(t)=\int_0^tq(s)\,ds$. We have
\begin{align*}
 L-\left(1-a+a^2\right)
 &=\int_0^1\left[S(t)\cdot\left( \int_0^t q(s)\,ds - a\right)+a\cdot e^{-q(t)}\right]dt\\
 &=\int_0^1\left(e^{-x(t)}\cdot\left(x(t)-a\right)+a\cdot e^{-q(t)}\right)\,dt.
\end{align*}
Since $x'(t)=q(t)$, it would be convenient to lower bound the above integrand by $c\cdot q(t)$ for some constant $c$.
We therefore seek a pointwise lower bound proportional to $q(t)$. For $x\ge0$, define
\begin{equation*}
 \phi_a(x)=\inf_{v>0}\left\{
   \frac{e^{-x}\cdot\left(x-a\right)+a\cdot e^{-v}}{v}
 \right\}.
\end{equation*}
The definition gives $\phi_a(x)\le0$, and $\phi_a(x)=0$ for $x\ge a$.
For $q(t)>0$, the definition yields
\begin{equation*}
 e^{-x}\cdot \left(x-a\right)+a\cdot e^{-q(t)}\ge q(t)\cdot \phi_a(x).
\end{equation*}
The same inequality follows when $q(t)=0$.
Using $dx(t)=q(t)\,dt$, we obtain
\begin{align*}
 L-\left(1-a+a^2\right)
 \ge\int_0^1q(t)\cdot\phi_a(x)\,dt
 =\int_0^{\int_0^1 q(s)\,ds}\phi_a(x)\,dx
 \ge\int_0^a\phi_a(x)\,dx.
\end{align*}
The last inequality follows because
$\phi_a(x)\le0$ and $\phi_a(x)=0$ for $x\ge a$. Thus, defining
\begin{equation*}
 \tilde{L}(a)=1-a+a^2+\int_0^a\phi_a(x)\,dx,
\end{equation*}
we conclude that $L\ge \tilde{L}(a)$.
It remains to show that $\tilde{L}$ is convex.

For $0<x<a$, the minimum defining $\phi_a(x)$ is attained at the unique
$v>0$ satisfying
\begin{equation*}
 \left(1+v\right)\cdot e^{-v}
 =\left(1-\frac{x}{a}\right)\cdot e^{-x}.
\end{equation*}
At this minimum, $\phi_a(x)=-a\cdot e^{-v}$. Differentiating the
first-order condition with respect to $a$ and substituting $v$ for $x$
gives
\begin{align*}
 -\frac{d^2}{da^2}\int_0^a\phi_a(x)\,dx
 =\int_0^\infty
   \frac{x^2\cdot e^{-x}}{a^2\cdot v^2\cdot\left(1+a-x\right)}\,dv
 \le\int_0^1\frac{v^2}{4}\,dv+\int_1^\infty\frac{1}{v^2}\,dv
 =\frac{13}{12}.
\end{align*}

Consequently, we have $\tilde{L}''(a)\geq 2-13/12 > 0$, which implies that $\tilde{L}$ is strictly convex.
\end{proof}

Since $\tilde{L}$ is strictly convex on $\left[0,1\right]$, we can use numerical evaluation to locate its unique minimum at $a\approx0.6022588$, with $\tilde{L}(a)\approx0.690619858$, as illustrated in
Figure~\ref{fig:iid-linear-scalar}.
Consequently,
\begin{equation*}
 L\ge\min_{0<a<1}\left\{\tilde{L}(a)\right\}>0.6906.
\end{equation*}

\begin{figure}[htb]
    \centering
    \begin{tikzpicture}
    \begin{axis}[
      width=9cm, height=4.5cm,
      xmin=0, xmax=1.05, ymin=0.68, ymax=1.03,
      xlabel={$a = A(0)$}, ylabel={$\tilde{L}(a)$},
      xlabel style={at={(axis description cs:1,0)}, anchor=west, xshift=6pt},
      xtick={0,0.2,0.4,0.6,0.8,1},
      ytick={0.7,0.8,0.9,1},
      ymajorgrids=true,
      grid style={gray!25},
      axis lines=left,
      tick align=outside,
      font=\small,
      clip=false,
    ]
    % Precomputed values of \tilde{L}(a).
    \addplot[blue!70!black, very thick] coordinates {
        (0.000000000,1.000000000) (0.020000000,0.980300955) (0.040000000,0.961207572)
        (0.060000000,0.942725340) (0.080000000,0.924859563) (0.100000000,0.907615367)
        (0.120000000,0.890997710) (0.140000000,0.875011390) (0.160000000,0.859661045)
        (0.180000000,0.844951168) (0.200000000,0.830886105) (0.220000000,0.817470065)
        (0.240000000,0.804707126) (0.260000000,0.792601236) (0.280000000,0.781156220)
        (0.300000000,0.770375787) (0.320000000,0.760263530) (0.340000000,0.750822933)
        (0.360000000,0.742057375) (0.380000000,0.733970132) (0.400000000,0.726564381)
        (0.420000000,0.719843207) (0.440000000,0.713809600) (0.460000000,0.708466464)
        (0.480000000,0.703816617) (0.500000000,0.699862795) (0.520000000,0.696607655)
        (0.540000000,0.694053774) (0.560000000,0.692203660) (0.580000000,0.691059745)
        (0.600000000,0.690624393) (0.602258796,0.690619858) (0.620000000,0.690899901)
        (0.640000000,0.691888501) (0.660000000,0.693592361) (0.680000000,0.696013592)
        (0.700000000,0.699154241) (0.720000000,0.703016302) (0.740000000,0.707601711)
        (0.760000000,0.712912353) (0.780000000,0.718950061) (0.800000000,0.725716615)
        (0.820000000,0.733213751) (0.840000000,0.741443153) (0.860000000,0.750406464)
        (0.880000000,0.760105279) (0.900000000,0.770541153) (0.920000000,0.781715597)
        (0.940000000,0.793630084) (0.960000000,0.806286045) (0.980000000,0.819684877)
        (1.000000000,0.833827936)
    };
    \addplot[only marks, mark=*, mark size=2pt, blue!70!black]
      coordinates {(0.6022587955,0.6906198580)};
    \node[align=left, anchor=south west] (minimum)
      at (axis cs:0.43,0.85)
      {$a\approx0.6022588$\\$\tilde{L}(a)\approx0.690619858$};
    \draw[gray, ->] (minimum.south)
      -- (axis cs:0.6022587955,0.6906198580);
    \end{axis}
    \end{tikzpicture}
    \caption{Plot for the function $\tilde{L}(a)$.}
    \label{fig:iid-linear-scalar}
\end{figure}
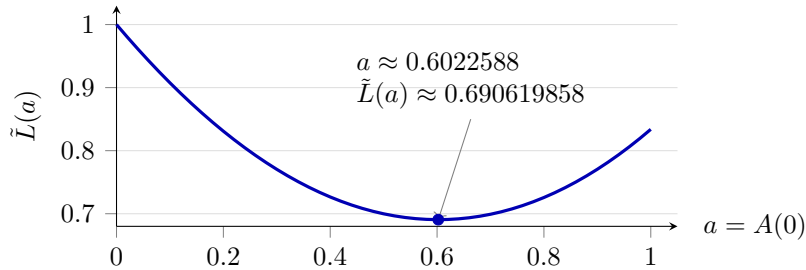

\subsection{Lower Bounding the Ratio via a Factor-revealing LP}
\label{sec:iid-comparison}

We now consider a piecewise linear threshold with $N$ equal-length intervals and positive nondecreasing drop rates $b_1,\ldots,b_N$. We fix these rates throughout the analysis. We describe how we choose $N$ and $b_1,\ldots,b_N$ later.
For every nonnegative nondecreasing acceptance rate $q$, we have an upper constraint on the guaranteed value in the max--min program of Section~\ref{sec:iid-reduction}. Since there are infinitely many such rates, as in the warm-up analysis above, the key is to concisely represent the upper bounds imposed by the infinitely many $q$ using a single function independent of $q$, e.g., by finding the pointwise worst case.

We introduce continuous functions $V_t\!\left(y,x\right)$ to characterize (lower bound) the contribution of the unknown $q$ in the time interval $\left[t,1\right]$. Informally, $x$ records the contribution already accumulated, and $y$ records the remaining contribution. At time zero, we start from $\left(y,x\right)=\left(1,0\right)$. We will use the following lemma to reduce lower-bounding $L$ to lower-bounding $V_0\!\left(1,0\right)$. The formal definition of $V_t\!\left(y,x\right)$ and the proofs of the lemma are deferred to Appendix~\ref{app:iid-continuation}.

\begin{lemma}\label{lem:iid-continuation}
There exists a family of continuous functions $V_t\!\left(y,x\right)$ on $0\le x\le y\le1$ such that for every nonnegative nondecreasing rate $q$, we have
\begin{equation*}
 L\ge\theta(0)+V_0\!\left(1,0\right).
\end{equation*}
For each $t$, $V_t\!\left(y,x\right)$ is a nonpositive concave function of $\left(y,x\right)$ that is nondecreasing in $x$ and nonincreasing in $y$, with $V_t\!\left(y,y\right)=0$ and $V_1\left(y,x\right)=0$.

Furthermore, fix any discretization of the time window $[0,1]$, any consecutive partition points $t,t+\Delta$, and a state $0\le x<y\le1$.
We cover $\left[0,\infty\right)$ by $N_q$ intervals for the mean acceptance rate on $\left[t,t+\Delta\right]$.
For each $j\in[N_q]$, we can compute constants $a_j\in \mathbb R$, $\phi_j\ge0$, and a state point $\left(y_j,x_j\right)$ with $0\le x_j\le y_j\le1$ that are independent of the unknown $q$, such that 
\begin{equation*}
 V_t\!\left(y,x\right)\ge
 \min_{j\in[N_q]}\left\{
 a_j+\phi_j\cdot V_{t+\Delta}\left(y_j,x_j\right)
 \right\}.
\end{equation*}
\end{lemma}

% We now describe how we choose $b_1,\ldots,b_N$ and derive the resulting ratio.
Recall the max--min program from Section~\ref{sec:iid-reduction}, and write $L\!\left(b,q\right)$ for the value of $L$ when $b$ and $q$ are fixed.
We aim to choose $b$ to maximize $\inf_q\left\{L\!\left(b,q\right)\right\}$.
For any fixed $b$, we have $L\!\left(b,q\right)\ge\theta(0)+V_0\!\left(1,0\right)$ for every admissible $q$, by Lemma~\ref{lem:iid-continuation}.
We first lower bound $V_0\!\left(1,0\right)$ by discretization and using a factor-revealing LP, then discuss how we select $b$.

\paragraph{Discretization.}
We divide the time window $\left[0,1\right]$ into $N_t$ equal-length intervals, where $N_t$ is divisible by $N$, and set $t_k=\frac{k}{N_t}$ for $k=0,\ldots,N_t$.
For the two state coordinates $y$ and $x$, we divide $\left[0,1\right]$ into $N_y$ intervals, whose lengths may vary:
\begin{equation*}
0=y_0=x_0<y_1=x_1<\cdots<y_{N_y}=x_{N_y}=1.
\end{equation*}
We use finer intervals near $0$ and $1$, in the numerical certificate below.
At each time $t_k$, we retain the state vertices $\left(y_i,x_j\right)$ with $0\le j\leq i\le N_y$.
We use $v_{k,i,j}$ to represent a certified lower bound on $V_{t_k}\left(y_i,x_j\right)$.
Here $k,i,j$ are integer indices for $t$, $y$, and $x$, respectively.
We triangulate the state domain using these vertices and use concavity to lower bound $V_t\!\left(y,x\right)$ by a weighted average of its values at the three vertices of a containing triangle.
We combine this bound with Lemma~\ref{lem:iid-continuation}, using a common cover of $\left[0,\infty\right)$ by $N_q=131{,}093$ intervals.
% We use the same index $h$ at every state because we consider the same rate $q$, whose average over the time interval is independent of the state; we cover the additional zero expression using the interval that extends to infinity.
% We retain the numerical guarantee because taking the worst case separately at each state before averaging cannot give a larger bound than averaging first and then taking the worst case, and rounding down can only decrease the bound further.
We obtain the following bound at each nonterminal grid vertex:
\begin{equation*}
 V_{t_k}\left(y_i,x_j\right)\ge
 \min_{h\in[N_q]}\left\{a_{k,i,j,h}
   +\sum_{x_{j'}\le y_{i'}}
       \eta_{k,i,j,h,i',j'}\cdot
       V_{t_{k+1}}\left(y_{i'},x_{j'}\right)\right\}.
\end{equation*}
Here $h$ indexes the $N_q$ intervals covering $\left[0,\infty\right)$.
We compute the constants $a_{k,i,j,h}$ and
nonnegative coefficients $\eta_{k,i,j,h,i',j'}$
in advance and keep them fixed in the LP.

\paragraph{The factor-revealing LP for fixed $b$.}
We fix the drop rates $b_1,\ldots,b_N$ throughout the LP; in particular, they are not LP variables.
We optimize only the candidate guarantee $\gamma$ and the grid values $v_{k,i,j}$, for $0\le k\le N_t$, $0\le j\leq i\le N_y$:
\begin{align*}
 \text{maximize}\quad &\gamma\\
 \text{subject to}\quad
 &\gamma\le\theta(0)+v_{0,N_y,0},\\
 &v_{N_t,i,j}=0,\\
 &v_{k,i,j}=0 &&\text{if } j=i,\\
 &v_{k,i,j}\le0,\\
 &v_{k,i,j}\le a_{k,i,j,h}
      +\sum_{x_{j'}\le y_{i'}}
         \eta_{k,i,j,h,i',j'}\cdot v_{k+1,i',j'}.
\end{align*}
We impose the constraints for all valid indices, with $k<N_t$, $x_j<y_i$, and $h\in[N_q]$ in the last line.
We obtain a feasible solution by evaluating the constraints backward from $k=N_t$ to $k=0$.
For every feasible solution, we prove $v_{k,i,j}\le V_{t_k}\!\left(y_i,x_j\right)$ by backward induction on $k$.
At $k=N_t$, we use $V_1=0$.
Therefore, for every nonnegative nondecreasing rate $q$, we obtain
\begin{equation*}
 L\!\left(b,q\right)\ge\theta(0)+V_0\!\left(1,0\right)
   \ge\theta(0)+v_{0,N_y,0}\ge\gamma.
\end{equation*}

% For the numerical certificate, we also apply the same interval bounds on enclosing rate intervals and take their strongest applicable consequence on each fine interval.
% At each time, we improve the certified state values using coordinate concavity and monotonicity of $V_t$ before continuing backward.
% Each step preserves a lower bound on $V_t$; once its interpolation weights are fixed, it is a linear inequality and can be checked by backward substitution.

\paragraph{Heuristic search for $b$.}
We choose $b_1,\ldots,b_N$ by heuristic searches using constraint enumeration.
Starting from a collection of candidate rates $q$, we solve the max--min program from Section~\ref{sec:iid-reduction} restricted to this collection.
We can optimize $b_1,\ldots,b_N$ for this restricted program because, for fixed $q$, we have a linear expression in $b$.
We then search for additional adverse rates against the resulting $b$, add them to the collection, and repeat.
We do not establish that the chosen $b$ is optimal for fixed $N$.
After fixing $b$, we certify its guarantee against every admissible $q$ using the factor-revealing LP.
We use $N=40$ positive nondecreasing drop rates obtained by this search.
With $N_t=640$, $N_y=1200$, and $N_q=131{,}093$, the following lower bound can be verified by backward evaluation of these finite inequalities with outward rounding:
\begin{equation*}
 \gamma=0.74191589\ldots,
 \quad \text{ and } \quad
 \theta(0)=0.74196156\ldots.
\end{equation*}

\subsection{Transferring the Guarantee to Every \texorpdfstring{$n$}{n}}
\label{sec:iid-transfer}

We now show that the proved guarantee for $n\to\infty$ also holds for every finite $n$ using the following shifted-rank rule and a transfer theorem.
Draw $U$ uniformly from $\left[0,1\right]$, independently of the values.
Let $X_i$ denote the $i$-th arrival.
We accept the first value satisfying $X_i>\E\!\left[X^*\right]\cdot\theta\left(\frac{i-U}{n}\right)$.
If the first $n-1$ variables are rejected, we accept $X_n$. The same $U$ is used at every rank, so the entire decreasing threshold sequence is fixed before the values are observed. The algorithm uses only $n$ and the expected maximum.
The proof appears in Appendix~\ref{app:iid-transfer}.

\begin{theorem}
\label{thm:iid-transfer}
Using the shifted-rank rule with the forty-slope threshold $\theta$, we have
\begin{equation*}
 \E\!\left[\ALG\right] > 0.7419\cdot \E\!\left[X^*\right],
 \qquad \text{ for every } n.
\end{equation*}
\end{theorem}

\section{Prophet Secretary}\label{sec:secretary}

We now consider the prophet secretary setting, where the variables may have different distributions.
We use the same notation for the threshold $\theta$ and its drop rate $b$.
As in the i.i.d. analysis, we first derive a lower bound on the expected reward that holds for every distribution, and then choose the threshold to maximize this guarantee.
Unlike the i.i.d. setting, the probability that $X^*$ exceeds the threshold does not determine the algorithm's survival probability $S(t)$.
We therefore retain the individual conditioning factors and derive lower bounds using $X^*$.
Note that we can assume without loss of generality that $n$ is sufficiently large, by inserting dummy variables with value $0$.

\subsection{Characterizing the Expected Reward}

As before, let $\tau\in\left[0,1\right]$ be the stopping time of the algorithm.
We set $\tau=1$ if no variable is accepted.
For each variable $i$, define
\begin{equation*}
 p_i(t)=\PP\!\left(X_i>\theta(t)\right),\qquad
 P_i(t)=\int_0^t p_i(s)\,ds.
\end{equation*}
The first term $p_i(t)$ is the probability of $X_i$ exceeding the threshold when it arrives at time $t$.
The second term $P_i(t)$ is the probability that variable $i$ arrives before $t$ and exceeds its threshold, regardless of whether another variable is selected.
By independence, the \emph{survival probability} is
\begin{equation*}
 S(t)=\PP\!\left(\tau>t\right)
 =\prod_i\left(1-P_i(t)\right),\qquad t\in\left[0,1\right).
\end{equation*}
We further let $S\left(1\right)$ be the probability that no value is accepted.
Define
\begin{equation*}
 p^*(t)=\PP\!\left(X^*>\theta(t)\right).
\end{equation*}
The function $p^*$ is nondecreasing as $\theta$ is nonincreasing.
For $0\le s\le t<1$, define
\begin{equation*}
 B_s(t)=\sum_i\frac{p_i(s)}{1-P_i(t)}-p^*(s)\ge0.
\end{equation*}
The nonnegativity follows because\footnote{The same lower bound is used in existing works~\cite{journals/siamdm/EsfandiariHLM17,conf/soda/EhsaniHKS18} to derive a lower bound of $1-1/e$. The crucial step of our analysis is to retrieve some extra gain when this lower bound is loose.}
\begin{equation*}
    \sum_i\frac{\PP\!\left(X_i>\theta(s)\right)}{1-P_i(t)} \ge \sum_i \PP\!\left(X_i>\theta(s)\right) \ge \PP\!\left(X^*>\theta(s)\right).
\end{equation*}
We use $S(t)\cdot B_s(t)$ to account for the extra gain beyond the bound obtained from $p^*$ alone.
To express the reward bound in terms of the stopping time, define, for $t\in\left[0,1\right]$,
\begin{equation*}
 F_{p^*}(t)=t+\theta(t)
 -\int_0^1b(s)\cdot\min\left\{s,t\right\}\cdot p^*(s)\,ds.
\end{equation*}

\begin{lemma}\label{lem:baseline}
The expected reward satisfies
\begin{equation*}
 \E\!\left[\ALG\right]
 \ge\E\!\left[F_{p^*}(\tau)\right]
 +\int_0^1b(s)\cdot\int_s^1S(t)\cdot B_s(t)\,dt\,ds.
\end{equation*}
\end{lemma}
\begin{proof}
We decompose an accepted value $X_i$ arriving at time $t_i$ into revenue $\theta(t_i)$ and utility $(X_i - \theta(t_i))^+$.
An observation $X_i$ arriving at time $t$ is accepted if $X_i>\theta(t)$ and no other variable was accepted earlier.
Independence and the uniform arrival times give
\begin{equation*}
 \E\!\left[\ALG\right]
 =\sum_i\int_0^1
 \E\!\left[X_i\cdot\mathbf{1}\left(X_i>\theta(t)\right)\right]
 \cdot\prod_{k\ne i}\left(1-P_k(t)\right)\,dt.
\end{equation*}
The revenue is $\E\!\left[\theta(\tau)\right]$, since $\theta\left(1\right)=0$ on failure.
For the utility, we use
$\E\!\left[\left(X_i-\theta(t)\right)^+\right]
=\int_{\theta(t)}^\infty\PP\!\left(X_i>x\right)\,dx$.
Using $S(t) = \prod_i\left(1-P_i(t)\right)$, the revenue--utility decomposition becomes
\begin{equation*}
 \E\!\left[\ALG\right]
 =\E\!\left[\theta(\tau)\right]
 +\int_0^1S(t)\cdot\int_{\theta(t)}^\infty
 \sum_i\frac{\PP\!\left(X_i>x\right)}{1-P_i(t)}\,dx\,dt.
\end{equation*}

Fix any $x\in\left[\theta(t),\theta\left(0\right)\right]$ and write $x=\theta(s)$.
Using the definition of $B_s(t)$, we obtain
\begin{align*}
 S(t)\cdot\int_{\theta(t)}^\infty
 \sum_i\frac{\PP\!\left(X_i>x\right)}{1-P_i(t)}\,dx
 \ge S(t)\cdot\int_{\theta(t)}^\infty
 \PP\!\left(X^*>x\right)\,dx
 +\int_0^t b(s)\cdot S(t)\cdot B_s(t)\,ds.
\end{align*}
Using $\E\!\left[X^*\right]=1$ and substituting $x=\theta(s)$ again gives
\begin{equation*}
 \int_{\theta(t)}^\infty\PP\!\left(X^*>x\right)\,dx
 =1-\int_0^{\theta(t)}\PP\!\left(X^*>x\right)\,dx
 =1-\int_t^1b(s)\cdot p^*(s)\,ds.
\end{equation*}

Finally, we use
$\int_0^1S(t)\,dt=\E\!\left[\tau\right]$
and
$\int_0^sS(t)\,dt=\E\!\left[\min\left\{s,\tau\right\}\right]$.
Integrating the utility bound, exchanging the order of integration, and adding the revenue, we obtain
\begin{align*}
 \E\!\left[\ALG\right]
 &\ge\E\!\left[\theta(\tau)\right]+\E\!\left[\tau\right]
 -\int_0^1b(s)\cdot p^*(s)\cdot
 \E\!\left[\min\left\{s,\tau\right\}\right]\,ds
 +\int_0^1b(s)\cdot\int_s^1S(t)\cdot B_s(t)\,dt\,ds\\
 &=\E\!\left[F_{p^*}(\tau)\right]
 +\int_0^1b(s)\cdot\int_s^1S(t)\cdot B_s(t)\,dt\,ds.\qedhere
\end{align*}
\end{proof}

\subsection{Threshold with Two Steps}\label{sec:secretary-warmup}

We begin with a warm-up showing that we can improve on $1-1/e$ using a threshold with only two linear pieces.
Choose
\begin{equation*}
    b(t) =
    \begin{cases}
        1/2, & 0\le t<2/3, \\
        1, & 2/3\le t\le1.
    \end{cases}
    \qquad
    \theta(t) = \int_t^1b(s)\,ds =
    \begin{cases}
        2/3-t/2, & 0\le t\le2/3, \\
        1-t,&2/3\le t\le1.
    \end{cases}
\end{equation*}

\begin{theorem}\label{prop:secretary-warmup}
    For the prophet secretary, the two-step drop rate above guarantees a competitive ratio of at least $0.6385$ for every $n$.
\end{theorem}
\begin{proof}
Recall that $S(t)=\PP\!\left(\tau>t\right)=\E\!\left[\mathbf{1}\left(\tau>t\right)\right]$, we have
\begin{align*}
    \int_s^1 S(t)\cdot B_s(t)\,dt &= \E\!\left[\int_s^1 \mathbf{1}\left(\tau>t\right)\cdot B_s(t)\,dt\right] \\
    &= \E\!\left[\int_0^{\left(\tau-s\right)^+} B_s\!\left(s+u\right)du\right].
\end{align*}
Notice that $p_i(t)$ is nondecreasing, we have $P_i\left(s+u\right) \ge P_i(s) + u\cdot p_i(s) \ge u\cdot p_i(s)$.
Consequently, for any $0\le d\le 1-s$, we have
\begin{align*}
    \int_0^d B_s\!\left(s+u\right)du &\ge \sum_i\int_0^d\frac{p_i(s)}{1-u\cdot p_i(s)}\,du-d\cdot p^*(s) \\
    &= -\log\left(\prod_i\left(1-d\cdot p_i(s)\right)\right)-d\cdot p^*(s) \\
    &\ge -\log\left(1-d\cdot \left(1-\prod_i\left(1-p_i(s)\right)\right)\right)-d\cdot p^*(s) \\
    &= -\log\left(1-d\cdot p^*(s)\right)-d\cdot p^*(s) \ge\frac{d^2\cdot p^*(s)^2}{2}.
\end{align*}
The second inequality follows by repeatedly applying
\begin{equation*}
    \left(1-d\cdot a\right)\cdot\left(1-d\cdot b\right) \le 1-d\cdot\left(a+b-a\cdot b\right),\qquad0\le a,b,d\le1,
\end{equation*}
and the last inequality uses $-\log\left(1-x\right)-x\ge x^2/2$.
Therefore,
\begin{equation*}
    \int_s^1 S(t)\cdot B_s(t)\,dt \ge \frac{p^*(s)^2}{2}\cdot
   \E\!\left[\left(\left(\tau-s\right)^+\right)^2\right].
\end{equation*}
Combining this with Lemma~\ref{lem:baseline}, we obtain
\begin{align*}
    \E[\ALG] &\ge \E\!\left[F_{p^*}(\tau)\right]+\frac12\cdot \int_0^1 b(s)\cdot p^*(s)^2\cdot \E\!\left[\left(\left(\tau-s\right)^+\right)^2\right]ds \\
    &= \E\!\left[F_{p^*}(\tau)+\frac12\cdot \int_0^\tau b(s)\cdot \left(\tau-s\right)^2\cdot p^*(s)^2\,ds
    \right].
\end{align*}
For any $t\in \left[0,1\right]$, let
\begin{equation*}
    G_{p^*}(t) = F_{p^*}(t)+\frac12\cdot \int_0^t b(s)\cdot \left(t-s\right)^2\cdot p^*(s)^2\,ds
\end{equation*}
be the quantity inside the above expectation.
In the following, we give a uniform lower bound on $G_{p^*}(t)$ for any time $t\in \left[0,1\right]$ and any $p^*$.
Let
\begin{equation*}
    F_1(t) = t+\theta(t) -\int_0^1b(s)\cdot\min\!\left\{s,t\right\} ds,
\end{equation*}
which is $F_{p^*}$ when $p^*(s)=1$ for all $s$.
By definition of $F_{p^*}$, we have
\begin{align*}
    F_{p^*}(t) &= F_1(t) + \int_0^1 b(s)\cdot \min\left\{s,t\right\} \cdot \left(1-p^*(s)\right) \,ds \\
    &= F_1(t) + \int_0^t b(s)\cdot s\cdot \left(1-p^*(s)\right)\,ds + \int_t^1 b(s)\cdot t\cdot \left(1-p^*(s)\right)\,ds.
\end{align*}
Notice that the last term on the RHS is nonnegative.
We have
\begin{align*}
    G_{p^*}(t) &\ge F_1(t) + \int_0^t b(s)\cdot s\cdot \left(1-p^*(s)\right)\,ds + \frac12\cdot \int_0^t b(s)\cdot \left(t-s\right)^2\cdot p^*(s)^2\,ds \\
    &= F_1(t) + \int_0^t b(s)\cdot \left(s\cdot \left(1-p^*(s)\right) + \frac{\left(t-s\right)^2}{2}\cdot p^*(s)^2\right)\,ds.
\end{align*}

For any $t\in \left[0,1\right]$ and $s\in \left[0,t\right)$, let $h = \left(t-s\right)^2$ and define
\begin{equation*}
    f(u) = s\cdot \left(1-u\right) + \frac{h}{2}\cdot u^2, \qquad \text{for } u\in \left[0,1\right].
\end{equation*}
Observe that $f$ is convex and achieves its minimum at $u^* = \min\left\{1,s/h\right\}$.
We have
\begin{equation*}
    \min_{u\in \left[0,1\right]}\left\{f(u)\right\} =
    \begin{cases}
        s-s^2/\left(2\cdot h\right), & s\le h, \\
        h/2, & s\ge h,
    \end{cases}
    \ \ge\frac12\cdot \min\left\{s,h\right\}.
\end{equation*}

Therefore, we establish a lower bound on $G_{p^*}$ that is independent of $p^*$:
\begin{equation*}
    G_{p^*}(t) \ge F_1(t) + \frac12\cdot \int_0^t b(s)\cdot\min\!\left\{s,\left(t-s\right)^2\right\} ds \geq F_1(t) + \frac{t^3}{32}.
\end{equation*}
We minimize the above bound over $t\in\left[0,1\right]$ and obtain $G_{p^*}(t)>0.6385$.
Taking expectations at $t=\tau$, we conclude that $\E[\ALG]\ge\E\!\left[G_{p^*}(\tau)\right]>0.6385$.
\end{proof}

\subsection{From Extra Gain to a Pointwise Guarantee}

Starting from Lemma~\ref{lem:baseline}, our goal is to construct a deterministic
function $\Gamma_{p^*}(\tau)$ such that, for the actual input, $\E\!\left[\ALG\right]\ge\E\!\left[\Gamma_{p^*}(\tau)\right]$.
A uniform bound $\Gamma_{p^*}(\tau)\ge\gamma$ for every nondecreasing
$p^*:\left[0,1\right]\to\left[0,1\right]$ and every $\tau\in\left[0,1\right]$
then proves a competitive ratio of $\gamma$.
We construct this function in three steps.
We first derive two lower bounds on the extra gain $\int_0^1b(s)\cdot\int_s^1S(t)\cdot B_s(t)\,dt\,ds$.
We then combine them with nonnegative weights whose total coefficient at each
pair $\left(s,t\right)$ is at most $b(s)$.
Finally, we use a lower bound on $S$ to add a correction with nonpositive expectation.

\paragraph{First lower bound.}
The warm-up uses $P_i\left(s+u\right)\ge u\cdot p_i(s)$.
We obtain a sharper lower bound by also using the probability accumulated before $s$:
\begin{equation*}
    P_i\left(s+u\right)\ge P_i(s)+u\cdot p_i(s).
\end{equation*}
To express the resulting lower bound using only $p^*$, for
$0\le s\le1$ and $0\le\Delta\le1-s$, let
\begin{equation*}
    \delta=\min\left\{s,\frac{1-\Delta}{2}\right\},\qquad
    y=\int_{s-\delta}^s p^*(v)\,dv,
\end{equation*}
and define
\begin{equation*}
    H_{p^*}\!\left(s,\Delta\right)
    =\log\left(1-y\right)-\log\left(1-y-\Delta\cdot p^*(s)\right)
     -\Delta\cdot p^*(s).
\end{equation*}
We use $\log0=-\infty$; in particular, $H_{p^*}\!\left(s,0\right)=0$.
The integral $y$ records probabilities on an earlier interval.
If we omit this contribution by setting $y=0$, the expression reduces to
$-\log\left(1-\Delta\cdot p^*(s)\right)-\Delta\cdot p^*(s)$,
whose quadratic lower bound was used in the warm-up.

\begin{lemma}\label{lem:conditioning}
For $0\le s\le c\le1$, we have
\begin{equation*}
    \int_s^c S(t)\cdot B_s(t)\,dt
    \ge\E\!\left[H_{p^*}\!\left(s,
    \left(\min\left\{\tau,c\right\}-s\right)^+\right)\right].
\end{equation*}
\end{lemma}

\begin{proof}
Fix $0\le s\le c\le1$. If $s=c$, both sides are zero, so assume $s<c$.
For any $\Delta\in\left[0,1-s\right]$, use $\delta$ and $y$ from the
definition of $H_{p^*}\!\left(s,\Delta\right)$ above.
Integrating the preceding monotonicity bound gives
\begin{equation*}
 \int_0^\Delta B_s\!\left(s+u\right)\,du
 \ge\sum_i-\log\left(1-\frac{\Delta\cdot p_i(s)}{1-P_i(s)}\right)
 -\Delta\cdot p^*(s).
\end{equation*}
Compared with the warm-up analysis, we keep $P_i(s)$, the contribution of
probabilities accumulated before $s$.
For $\delta>0$, define the average exceedance probability of variable $i$ on
$\left[s-\delta,s\right]$ by
\begin{equation*}
 q_i=\frac1\delta\cdot\int_{s-\delta}^s p_i(v)\,dv.
\end{equation*}
When $\delta=0$, set $q_i=0$.
Monotonicity gives $q_i\le p_i(s)$ and $P_i(s)\ge\delta\cdot q_i$, so
\begin{equation*}
 \int_0^\Delta B_s\!\left(s+u\right)\,du
 \ge\sum_i-\log\left(1-\frac{\Delta\cdot p_i(s)}{1-\delta\cdot q_i}\right)
 -\Delta\cdot p^*(s).
\end{equation*}

We next combine the individual factors.
For any pairs $\left(a_i,q_i\right)$ satisfying $0\le q_i\le a_i\le1$,
and nonnegative $\delta,\Delta$ with $2\cdot\delta+\Delta\le1$, we claim that
\begin{equation*}
 \prod_i\left(1-\frac{\Delta\cdot a_i}{1-\delta\cdot q_i}\right)
 \le1-\frac{\Delta\cdot\left(1-\prod_i\left(1-a_i\right)\right)}
 {1-\delta\cdot\left(1-\prod_i\left(1-q_i\right)\right)}.
\end{equation*}
To prove this, first merge two pairs and write $a=a_1+a_2-a_1\cdot a_2$ and $q=q_1+q_2-q_1\cdot q_2$.
We show that
\begin{equation*}
 \left(1-\frac{\Delta\cdot a_1}{1-\delta\cdot q_1}\right)
 \cdot\left(1-\frac{\Delta\cdot a_2}{1-\delta\cdot q_2}\right)
 \le1-\frac{\Delta\cdot a}{1-\delta\cdot q}.
\end{equation*}
After multiplying by the positive denominators, this is equivalent to
$D\!\left(a_1,a_2\right)\ge0$, where
\begin{align*}
 D\!\left(a_1,a_2\right)
 &=\left(1-\delta\cdot q_1\right)\cdot\left(1-\delta\cdot q_2\right)
   \cdot\left(1-\delta\cdot q-\Delta\cdot a\right)\\
 &\quad-\left(1-\delta\cdot q_1-\Delta\cdot a_1\right)
   \cdot\left(1-\delta\cdot q_2-\Delta\cdot a_2\right)
   \cdot\left(1-\delta\cdot q\right).
\end{align*}
Fix $q_1,q_2$. The expression $D$ is linear in either $a_i$ when the other is fixed.
It can be verified that the coefficients of both $a_1$ and $a_2$ are nonnegative.
Thus $D$ is minimized at $\left(a_1,a_2\right)=\left(q_1,q_2\right)$, where $D\!\left(q_1,q_2\right)
 =\Delta\cdot q_1\cdot q_2\cdot
 \left(1-2\cdot\delta-\Delta+\delta\cdot\left(\delta+\Delta\right)\cdot q\right)
 \ge0$.
This proves the two-pair inequality.
Repeating the merging proves the claimed product inequality.

Apply this inequality with $a_i=p_i(s)$ and take negative logarithms.
The choice of $\delta$ ensures $2\cdot\delta+\Delta\le1$, and we obtain
\begin{equation*}
 \int_0^\Delta B_s\!\left(s+u\right)\,du
 \ge-\log\left(1-\frac{\Delta\cdot p^*(s)}
 {1-\delta\cdot\left(1-\prod_i\left(1-q_i\right)\right)}\right)
 -\Delta\cdot p^*(s).
\end{equation*}
For $\delta>0$, the functions $1-p_i(v)$ are nonnegative and nonincreasing.
Applying Chebyshev's integral inequality repeatedly gives
\begin{align*}
 \prod_i\left(1-q_i\right)
 =\prod_i\left(\frac1\delta\cdot\int_{s-\delta}^s\left(1-p_i(v)\right)\,dv\right)
 \le\frac1\delta\cdot\int_{s-\delta}^s\prod_i\left(1-p_i(v)\right)\,dv
 =1-\frac y\delta.
\end{align*}
Hence $\delta\cdot\left(1-\prod_i\left(1-q_i\right)\right)\ge y$.
This also holds for $\delta=0$, when $q_i=y=0$.
Substituting into the logarithmic bound yields
\begin{equation*}
 \int_0^\Delta B_s\!\left(s+u\right)\,du
 \ge\log\left(1-y\right)-\log\left(1-y-\Delta\cdot p^*(s)\right)
 -\Delta\cdot p^*(s)
 =H_{p^*}\!\left(s,\Delta\right).
\end{equation*}

Finally, use $S(t)=\PP\!\left(\tau>t\right)$ for $t<1$ and apply this bound
with $\Delta=\left(\min\left\{\tau,c\right\}-s\right)^+$:
\begin{align*}
 \int_s^c S(t)\cdot B_s(t)\,dt
 &=\E\!\left[\int_s^c\mathbf{1}\!\left(\tau>t\right)\cdot B_s(t)\,dt\right]\\
 &=\E\!\left[\int_0^{\left(\min\left\{\tau,c\right\}-s\right)^+}
 B_s\!\left(s+u\right)\,du\right]\\
 &\ge\E\!\left[H_{p^*}\!\left(s,
 \left(\min\left\{\tau,c\right\}-s\right)^+\right)\right].\qedhere
\end{align*}
\end{proof}

\paragraph{Second lower bound.}
We also bound $S(t)\cdot B_s(t)$ at a fixed time $t$ by comparing the algorithm's stopping
probability with the arrival of $X^*$.
Let $P^*(t)=\int_0^t p^*(v)\,dv$.
For $z\le t$, the difference $P^*(t)-P^*(z)$ is the probability that
the maximum arrives during $\left[z,t\right]$ and exceeds its threshold.
The next lower bound relates this probability and $1-S(t)$ to the extra gain.

\begin{lemma}\label{lem:late}
For $0\le z\le s\le t<1$, we have
\begin{equation*}
    S(t)\cdot B_s(t)\ge\left(1-S(t)\right)\cdot p^*(s)
    +\frac{P^*(t)-P^*(z)-\left(1-S(t)\right)}{1-t+z}.
\end{equation*}
\end{lemma}
\begin{proof}
Fix $0\le z\le s\le t<1$.
Choose an index $j$ attaining $X^*$, breaking ties by a fixed order on indices.
We lower-bound $1-S(t)=\PP\!\left(\tau\le t\right)$ using two disjoint events:
\begin{itemize}
    \item $X^*$ arrives during $\left[z,t\right]$ and exceeds its threshold.
    \item Some variable other than $j$ arrives before $t$ and exceeds its threshold;
    $X^*>\theta(s)$; and $t_j\notin\left[z,t\right]$.
\end{itemize}
Both events imply $\tau\le t$.
The first has probability $P_1=P^*(t)-P^*(z)$.
Let $P_2$ be the probability of the second event.
If $X^*>\theta(s)$ and $t_j\notin\left[z,t\right]$, but the second event
does not occur, then $X_j>\theta(s)$ and no other variable exceeds its threshold
upon arrival by time $t$. Hence
\begin{align*}
    p^*(s)-\frac{P_2}{1-\left(t-z\right)}
    &\le\sum_i p_i(s)\cdot\prod_{k\ne i}\left(1-P_k(t)\right)\\
    &=S(t)\cdot\sum_i\frac{p_i(s)}{1-P_i(t)}
     =S(t)\cdot\left(B_s(t)+p^*(s)\right).
\end{align*}
Rearranging gives
\begin{equation*}
    P_2\ge\left(1-t+z\right)\cdot
    \left(\left(1-S(t)\right)\cdot p^*(s)-S(t)\cdot B_s(t)\right).
\end{equation*}
Substituting into $1-S(t)\ge P^*(t)-P^*(z)+P_2$ proves
\begin{equation*}
    S(t)\cdot B_s(t)\ge\left(1-S(t)\right)\cdot p^*(s)
    +\frac{P^*(t)-P^*(z)-\left(1-S(t)\right)}{1-t+z}.\qedhere
\end{equation*}
\end{proof}

\paragraph{Linear combination of the lower bounds.}
Lemmas~\ref{lem:conditioning} and~\ref{lem:late} give lower bounds on the same extra gain.
We therefore allocate the available coefficient $b(s)$ between these two lower bounds.
We choose the endpoint and weight functions without using the unknown distributions. The resulting lower bound must hold for every admissible $p^*$.
For Lemma~\ref{lem:conditioning}, choose endpoints
$c_j(s)\in\left[s,1\right]$ and weights $\lambda_j(s)\ge0$.
An application with endpoint $c_j(s)$ uses coefficient $\lambda_j(s)$
at every $t\in\left[s,c_j(s)\right]$.
For Lemma~\ref{lem:late}, choose starting points
$z_k\!\left(s,t\right)\in\left[0,s\right]$ and weights
$\mu_k\!\left(s,t\right)\ge0$ at each pair $\left(s,t\right)$.
To keep the total coefficient within $b(s)$, require
\begin{equation}\label{eq:secretary-capacity}
    \sum_{j:c_j(s)\ge t}\lambda_j(s)
    +\sum_k\mu_k\!\left(s,t\right)\le b(s),
    \qquad 0\le s\le t<1.
\end{equation}
Because $S(t)\cdot B_s(t)\ge0$, this constraint gives
\begin{align*}
    \int_0^1b(s)\cdot\int_s^1S(t)\cdot B_s(t)\,dt\,ds
    &\ge\int_0^1\sum_j\lambda_j(s)\cdot\int_s^{c_j(s)}S(t)\cdot B_s(t)\,dt\,ds\\
    &\quad+\int_0^1\int_s^1\sum_k\mu_k\!\left(s,t\right)\cdot S(t)\cdot B_s(t)\,dt\,ds.
\end{align*}
We apply Lemma~\ref{lem:conditioning} to the first term and
Lemma~\ref{lem:late} to the second.
The introduction of the linear combination allows us to derive a stronger lower bound by adjusting these parameters when running the factor-revealing LP.

\paragraph{How to define $\Gamma_{p^*}(\tau)$?}
Combining these contributions with $F_{p^*}(\tau)$ gives an expression
whose expectation lower-bounds $\E\!\left[\ALG\right]$.
Suppose we define this lower bound by $\Gamma_{p^*}(\tau)$.
Our goal is to prove $\E\!\left[\ALG\right] \ge\E\!\left[\Gamma_{p^*}(\tau)\right]\ge\gamma$ for every input.
The guarantee must therefore hold for every admissible $p^*$.
However, $p^*$ does not determine the distribution of $\tau$:
the survival probability $S(t)=\prod_i\left(1-P_i(t)\right)$ also depends
on the individual distributions that give rise to $p^*$.
To avoid characterizing all compatible stopping-time distributions,
we will establish $\Gamma_{p^*}(\tau)\ge\gamma$ for every admissible $p^*$
and every possible value $\tau\in\left[0,1\right]$ of the stopping time.
Taking a minimum over all possible stopping times allows a single small value to limit the guarantee, regardless of how unlikely that stopping time is. In particular, a small value at an early stopping time can limit the guarantee even when the algorithm is unlikely to stop so early.
We therefore introduce a decreasing correction term when defining $\Gamma_{p^*}(\tau)$ that shifts value towards earlier stopping times, at the expense of later ones.
We choose a nonnegative integrable weight $\kappa$ without using the
unknown input distributions, and define
\begin{equation*}
    \widehat S(s)=\left(1-p^*(s)\right)\cdot
    \exp\left(\int_0^s\log\left(1-p^*(v)\right)\,dv\right),
    \qquad \forall s\in [0,1].
\end{equation*}
We use these functions to form the correction term
\begin{equation*}
    \int_0^1\kappa(s)\cdot\widehat S(s)\,ds
    -\int_0^\tau\kappa(s)\cdot\left(1-p^*(s)\right)\,ds.
\end{equation*}
This correction term is nonincreasing in $\tau$.
Since $0\le\widehat S(s)\le1-p^*(s)$, the correction is nonnegative
at $\tau=0$ and nonpositive at $\tau=1$.
It can therefore shift value towards earlier stopping times.
The weight $\kappa$ controls the amount of this adjustment at each time.
The factor $1-p^*(s)$ makes the adjustment small where $p^*(s)$ is close
to $1$ and gives the expression $\widehat S$ used in the linear lower
bounds of Appendix~\ref{app:secretary-linear-bounds}.
Lemma~\ref{lem:survival-bounds} below shows that the first integral is
no larger than the expected value of the subtracted integral.
Thus, the correction has nonpositive expectation and preserves a lower bound on expected reward.
Define:
\begin{align*}
    \Gamma_{p^*}(\tau)
    &=F_{p^*}(\tau)
    +\int_0^1\sum_j\lambda_j(s)\cdot
    H_{p^*}\!\left(s,\left(\min\left\{\tau,c_j(s)\right\}-s\right)^+\right)\,ds\\
    &\quad+\int_0^1\int_s^1\sum_k\mu_k\!\left(s,t\right)\cdot
    \left(\mathbf{1}\!\left(\tau\le t\right)\cdot p^*(s)
    +\frac{P^*(t)-P^*\!\left(z_k\!\left(s,t\right)\right)
    -\mathbf{1}\!\left(\tau\le t\right)}{1-t+z_k\!\left(s,t\right)}\right)\,dt\,ds\\
    &\quad+\int_0^1\kappa(s)\cdot\widehat S(s)\,ds
    -\int_0^\tau\kappa(s)\cdot\left(1-p^*(s)\right)\,ds.
\end{align*}

\begin{lemma}\label{lem:survival-bounds}
For every nonnegative integrable weight $\kappa$, the actual stopping time
$\tau$ satisfies
\begin{equation*}
    \E\!\left[
        \int_0^1\kappa(s)\cdot\widehat S(s)\,ds
        -\int_0^\tau\kappa(s)\cdot\left(1-p^*(s)\right)\,ds
    \right]\le0.
\end{equation*}
\end{lemma}
\begin{proof}
We first lower-bound $S(s)$ for $0\le s<1$.
Since $S(s) = \prod_{i}\left( 1-P_i(s) \right)$, we have
\begin{align*}
    \log\left(S(s)\right)
    &=\sum_i\log\left(1-P_i(s)\right)\\
    &\ge\sum_i\int_0^s\log\left(1-p_i(v)\right)\,dv\\
    &=\int_0^s\log\left(\prod_i\left(1-p_i(v)\right)\right)\,dv
    =\int_0^s\log\left(1-p^*(v)\right)\,dv,
\end{align*}
where the inequality follows from Jensen's inequality for the concave logarithm, and $1-P_i(s)=\left(1-s\right)\cdot1
    +\int_0^s\left(1-p_i(v)\right)\,dv$.
Taking exponential and multiplying by $1-p^*(s)$, we obtain
\begin{equation*}
    \left(1-p^*(s)\right)\cdot S(s)\ge\widehat S(s).
\end{equation*}
Using $S(s)=\PP\!\left(\tau>s\right)$ for $s<1$ and exchanging the
expectation and integral, we conclude that
\begin{align*}
    &\E\!\left[
        \int_0^1\kappa(s)\cdot\widehat S(s)\,ds
        -\int_0^\tau\kappa(s)\cdot\left(1-p^*(s)\right)\,ds
    \right]\\
    &\qquad=\int_0^1\kappa(s)\cdot
    \left(\widehat S(s)-\left(1-p^*(s)\right)\cdot S(s)\right)\,ds
    \le0.\qedhere
\end{align*}
\end{proof}

By definition of $\Gamma_{p^*}(\tau)$ and the nonpositive expectation of the correction term, we have the following immediately.
In particular, if $\Gamma_{p^*}(\tau)\ge\gamma$ for every nondecreasing $p^*:\left[0,1\right]\to\left[0,1\right]$ and every possible value $\tau\in\left[0,1\right]$, then the algorithm is $\gamma$-competitive.

\begin{lemma}\label{lem:secretary-criterion}
For the actual input, we have
\begin{equation*}
    \E\!\left[\ALG\right]
    \ge\E\!\left[\Gamma_{p^*}(\tau)\right].
\end{equation*}
\end{lemma}

\subsection{Optimizing and Verifying the Pointwise Guarantee}

As in the i.i.d. setting, we restrict $b$ to $N$ equal-length intervals with
nonnegative nondecreasing drop rates $b_1,\ldots,b_N$, and require $\theta(0)\le1$.
We treat the endpoints $c_j(s)$ and starting points $z_k\!\left(s,t\right)$
as outer design parameters. For each fixed choice, we consider
\begin{equation*}
    \sup_{b,\lambda,\mu,\kappa}\left\{ \inf_{\tau,p^*}
    \left\{\Gamma_{p^*}(\tau)\right\} \right\},
\end{equation*}
where $\lambda$ and $\mu$ satisfy~\eqref{eq:secretary-capacity}, $\kappa$ is nonnegative and integrable, $\tau\in[0,1]$, and $p^*$ is nondecreasing.
The outer optimization chooses a common guarantee; the inner minimization
requires it to hold for every admissible $p^*$ and every possible stopping time.

The two nonlinear contributions in $\Gamma$ are $H_{p^*}$ and $\widehat S$.
We lower-bound them by expressions linear in $p^*$, using the bounds in
Appendix~\ref{app:secretary-linear-bounds}.
As in the i.i.d. case, we use interval averages of $p^*$ and bounds valid throughout each stopping-time interval. Finally, monotonicity of those averages reduces the
remaining minimization to finitely many linear constraints.
Appendix~\ref{app:ind-cells} gives the discretization and the resulting
factor-revealing LP.
Every feasible solution certifies $\Gamma_{p^*}(\tau)\ge\gamma$ uniformly;
optimizing the LP searches for the largest such certified value.
We use $N=20$ positive nondecreasing drop rates.
The resulting threshold has initial value $\theta(0)\approx 0.71987853$.

\begin{lemma}
There exists a continuous, concave, strictly decreasing threshold function $\theta$ with admissible proof parameters, such that
\begin{equation*}
    \Gamma_{p^*}(\tau)\ge0.690215522\ldots
\end{equation*}
for every nondecreasing function $p^*:[0,1]\to[0,1]$ and every $\tau\in [0,1]$.
\end{lemma}

Applying Lemma~\ref{lem:secretary-criterion} to the actual input now gives
\begin{equation*}
    \E\!\left[\ALG\right]\ge\E\!\left[\Gamma_{p^*}(\tau)\right]>0.6902.
\end{equation*}
This proves Theorem~\ref{thm:secretary} for every $n$.

\section*{Declaration for the Use of AI}
We used OpenAI's GPT-6 Astra and GPT-6.1 Sol to assist with mathematical analysis, language editing, notation consistency checks, and LaTeX formatting. All mathematical statements, proofs, calculations, and citations were reviewed and verified by the authors, who take full responsibility for the content of the paper.

\clearpage
\bibliographystyle{alpha}
\bibliography{ref}

\clearpage

\addtocontents{toc}{\protect\setcounter{tocdepth}{1}}
\appendix

\section{Supplementary I.I.D. Analysis}\label{app:primal}

\begin{lemma}\label{lem:k-bounds}
Every nonnegative nondecreasing $q$ satisfies $0\le k(t)\le S(t)\le1$.
\end{lemma}
\begin{proof}
Since $q$ is nondecreasing, we have
\begin{align*}
 S(t)-q(t)\cdot A(t)
 &=S(t)-q(t)\cdot\int_t^1S(s)\,ds\\
 &\ge S(t)-q(t)\cdot S(t)
       \cdot \int_t^1e^{-q(t)\cdot\left(s-t\right)}\,ds\\
 &=S(t)\cdot e^{-q(t)\cdot\left(1-t\right)}\\
 &=\exp\left(-\int_0^tq(u)\,du-q(t)\cdot\left(1-t\right)\right)\\
 &\ge\exp\left(-t\cdot q(t)-q(t)\cdot\left(1-t\right)\right)
 =e^{-q(t)},
\end{align*}
where the first inequality uses
$S(s)\le S(t)\cdot e^{-q(t)\cdot\left(s-t\right)}$ for $s\ge t$.
It follows that
$k(t)\ge\left(1-A(0)\right)\cdot e^{-q(t)}\ge0$.
Finally, $k(t)\le S(t)\le1$ follows since $q(t)\geq 0$ and $A(0)\geq 0$.
\end{proof}

\subsection{Proof of Lemma~\ref{lem:iid-continuation}}
\label{app:iid-continuation}

We first construct $V_t$ and establish its basic properties.
Define
\begin{equation*}
    c=\max_{0\le s\le1}\left\{b(s)+\theta(s)\right\}-1.
\end{equation*}
Note that $k(s) \le S(s) - A(s)\cdot \left(q(s)+e^{-q(s)}\right) \leq S(s) - A(s)$.
If $c\le 0$, we have
\begin{align*}
 L-\theta(0)
 &\ge\left(1-\theta(0)\right)\cdot A(0)
   -\int_0^1b(s)\cdot\left(S(s)-A(s)\right)\,ds\\
& =\int_0^1 S(s)\cdot
\left(1-\theta(0)-b(s)+ \int_0^s b(t)\,dt \right)\,ds\\
 &=\int_0^1S(s)\cdot\left(1-\theta(s)-b(s)\right)\,ds\ge0.
\end{align*}
% We extend this bound to unbounded rates by truncation.
In this case, we take $V_t=0$.
We assume henceforth that $c>0$.
For a fixed rate $q$, define
\begin{equation*}
 S_t(s)=\exp\left(-\int_t^s q(u)du\right),\qquad
 A_t(s)=\int_s^1S_t(u)du.
\end{equation*}
At $t=0$, we have $S_0=S$ and $A_0=A$ as defined in Section~\ref{sec:iid-reduction}.
As discussed in Section~\ref{sec:iid-comparison}, we discretize $\left[0,1\right]$ into $N_t$ equal-length intervals and define $t_k=k/N_t$.
In defining $V_t$, we minimize over bounded nonnegative $q$ that are nondecreasing separately on each time interval $\left[t_k,t_{k+1}\right)$, without requiring monotonicity across different intervals.
We can therefore minimize independently over $q$ on consecutive intervals.
For each such $q$, $0\le y\le1$, and $x\ge0$, define
\begin{align*}
 V_t\!\left(y,x,q\right)
 &=\left(c\cdot x+y\cdot\left(1-\theta(t)\right)\right)\cdot A_t(t)
   -y\cdot\int_t^1b(s)\cdot\left(S_t(s)-q(s)\cdot A_t(s)\right)\,ds\\
 &\quad+\left(1-y+t\cdot y+y\cdot A_t(t)\right)\cdot
        \int_t^1 b(s)\cdot e^{-q(s)}\,ds.
\end{align*}
Let $V_t\!\left(y,x\right)=\inf_q\left\{V_t\!\left(y,x,q\right)\right\}$,
which is well-defined for the whole domain $0\le y\le1$ and $x\ge0$, instead of just $0\leq x\leq y\leq 1$.
In the following, we establish a few properties of $V_t\!\left(y,x\right)$. 
\begin{enumerate}
 \item We have $L\ge\theta(0)+V_0\!\left(1,0\right)$, since for every admissible $q$, we have $V_0\!\left(1,0,q\right)=L-\theta(0)$.

 \item We have $-y\cdot\theta(t)\le V_t\!\left(y,x\right)\le0$.
 Since $\theta(t)\le\theta(0)\le1$, we obtain the lower bound by discarding all terms except $-y\cdot\int_t^1b(s)\cdot S_t(s)\,ds$ and using $S_t\le1$.
 For the upper bound, we take the constant rate $q=M$.
 We then have $S_t(s)-M\cdot A_t(s)=S_t\left(1\right)$, and every term in $V_t\!\left(y,x,q\right)$ tends to zero as $M\to\infty$.
 In particular, we have $V_t\!\left(0,x\right)=0$ and $V_1\left(y,x\right)=0$.

 \item The function $V_t\!\left(y,x\right)$ is nondecreasing in $x$ and nonincreasing in $y$.
 For each $q$, $V_t\!\left(y,x,q\right)$ is a constant plus a linear expression in $x,y$, with constant term $\int_t^1b(s)\cdot e^{-q(s)}\,ds\ge0$, and $x$ coefficient $c\cdot A_t(t)\ge0$.
 Since $V_t\!\left(y,x\right)\leq 0$, we can discard every $q$ for which the $y$ coefficient of $V_t\!\left(y,x,q\right)$ is nonnegative.
 For the retained rates, we have a nonnegative $x$ coefficient and a negative $y$ coefficient.
 Taking the infimum yields the property.

 \item For every fixed $t$, $V_t\!\left(y,x\right)$ is concave and continuous in $\left(y,x\right)$.
 We obtain concavity by taking the infimum of these functions $V_t\!\left(y,x,q\right)$.
 Since $V_t\!\left(y,x,q\right)$ is continuous for every fixed $q$, it suffices to upper and lower bound the coefficients of $y$ and $x$ by some constant, which can be easily verified by the definition of $V_t\!\left(y,x\right)$.

 \item We have $V_t\!\left(y,x\right)=0$ whenever $x\ge y$.
 Since $A_t(t)\ge A_t(s)$ and
$q(s)+e^{-q(s)}\ge1$, we have
\begin{equation*}
 S_t(s)-q(s)\cdot A_t(s)
 -A_t(t)\cdot e^{-q(s)}
 \le S_t(s)-A_t(s).
\end{equation*}
Substituting this into the definition of
$V_t\!\left(y,x,q\right)$, we obtain
\begin{align*}
 V_t\!\left(y,x,q\right)
 &\ge\int_t^1 S_t(s)\cdot
 \left(c\cdot x+y\cdot
 \left(1-\theta(s)-b(s)\right)\right)\,ds\\
 &\quad+\left(1-y+t\cdot y\right)\cdot
 \int_t^1 b(s)\cdot e^{-q(s)}\,ds
 \ge0,
\end{align*}
where the last inequality uses $x\ge y$,
$c\ge b(s)+\theta(s)-1$, and $0\le y\le1$.
Together with $V_t\!\left(y,x\right)\le0$, this proves
$V_t\!\left(y,x\right)=0$.
\end{enumerate}

Fix consecutive partition points $t=t_k$ and $t+\Delta=t_{k+1}$, where $\Delta=1/N_t$, and a state $0\le x<y\le1$.
Write $b>0$ for the constant drop rate on $\left[t,t+\Delta\right]$ and define
\begin{equation*}
 p=c\cdot x+y\cdot\left(1-\theta(t)-b\right),
 \qquad u=b\cdot\left(1-y+t\cdot y\right).
\end{equation*}
For each bounded nonnegative rate $q$ that is nondecreasing on this interval, define
\begin{equation*}
 I_q=\int_t^{t+\Delta}S_t(s)\,ds,
 \qquad E_q=\int_t^{t+\Delta}e^{-q(s)}\,ds.
\end{equation*}
For finite $r\ge0$ and the constant rate $q(s)=r$, we write $I_r=I_q$; thus $I_r=\int_0^\Delta e^{-r\cdot s}\,ds$.
We have $I_0=\Delta$ and $I_r=\left(1-e^{-\Delta\cdot r}\right)/r$ for $r>0$.
Define
\begin{equation*}
 R_q=p\cdot I_q-y\cdot b\cdot\int_t^{t+\Delta}S_t(s)\cdot\log S_t(s)\,ds
       +u\cdot E_q+y\cdot b\cdot I_q\cdot E_q.
\end{equation*}

We first relate $V_t$ to $V_{t+\Delta}$, keeping $R_q$ explicit.
We will then bound $R_q$ using only the interval containing the mean of $q$.
To relate consecutive times, fix partition points $0=r_0<r_1<\cdots<r_{N_q}=\infty$.
We cover $\left[0,\infty\right)$ by the intervals $\left[r_{j-1},r_j\right]$, $j\in[N_q]$, and define
\begin{align*}
 \phi_j&=1-y+y\cdot\left(t+I_{r_{j-1}}
             +e^{-\Delta\cdot r_{j-1}}\cdot\left(1-t-\Delta\right)\right),\\
 y_j&=\frac{y\cdot e^{-\Delta\cdot r_{j-1}}}{\phi_j},\\
 x_j&=\min\!\left\{y_j,\frac{e^{-\Delta\cdot r_{j-1}}\cdot
       \left(c\cdot x+y\cdot b\cdot\Delta\cdot
           \left(r_{j-1}+e^{-r_{j-1}}-1\right)\right)}{c\cdot\phi_j}\right\}.
\end{align*}
Note that all these quantities are independent of $q$.
Moreover, we have $\phi_j>0$ since $I_{r_{j-1}}\ge\Delta\cdot e^{-\Delta\cdot r_{j-1}}$, which implies $\phi_j\ge1-y+y\cdot e^{-\Delta\cdot r_{j-1}}>0$.
Furthermore, we have $0<y_j\le1$.
We also have $x_j\ge0$ because $r_{j-1}+e^{-r_{j-1}}-1\ge0$.
Therefore, we have $0\le x_j\le y_j \leq 1$.

\begin{lemma}\label{lem:iid-state-step}
For every rate $q$ allowed in the definition of $V_t$ with
\begin{equation*}
 r_{j-1}\le\frac1\Delta\cdot \int_t^{t+\Delta}q(s)\,ds\le r_j,
\end{equation*}
we have
\begin{equation*}
 V_t\!\left(y,x,q\right)\ge R_q+
       \phi_j\cdot V_{t+\Delta}\left(y_j,x_j\right).
\end{equation*}
\end{lemma}
\begin{proof}
Fix $j\in[N_q]$ and a rate $q$ whose mean $r=\frac1\Delta\cdot \int_t^{t+\Delta}q(s)\,ds$ lies in $\left[r_{j-1},r_j\right]$.
We have $S_t\left(t+\Delta\right)=e^{-\Delta\cdot r}$.
By monotonicity of $q$ on this time interval, we have $S_t\left(t+s\right)\ge e^{-r\cdot s}$, so $I_q\ge I_r$.
By Jensen's inequality, we have $E_q\ge\Delta\cdot e^{-r}$.
We split the defining integrals of $V_t\!\left(y,x,q\right)$ at $t+\Delta$ and use
$S_{t+\Delta}(s)=S_t(s)/S_t\left(t+\Delta\right)$.
After subtracting $R_q$, we obtain
\begin{align*}
 &V_t\!\left(y,x,q\right)-R_q
  -\frac{y\cdot e^{-\Delta\cdot r}}{y_j}\cdot
       V_{t+\Delta}\left(y_j,x_j,q\right)\\
 ={}&y\cdot e^{-\Delta\cdot r}\cdot
 \underbrace{\left(\frac{c\cdot x}{y}
       +b\cdot\left(\Delta\cdot r+E_q-\Delta\right)
       -\frac{c\cdot x_j}{y_j}\right)}_{\text{coefficient $1$}}
       \cdot A_{t+\Delta}\left(t+\Delta\right)\\
 &+\underbrace{\left(1-y+t\cdot y+y\cdot I_q
      -y\cdot e^{-\Delta\cdot r}\cdot    \left(\frac1{y_j}-1+t+\Delta\right)\right)}_{\text{coefficient $2$}}
       \cdot\int_{t+\Delta}^1b(s)\cdot e^{-q(s)}\,ds.
\end{align*}
We show that the two coefficients are nonnegative.
First, since $r+e^{-r}-1$ is nondecreasing for $r\ge0$, the definition of $x_j$ gives
\begin{equation*}
 \frac{c\cdot x}{y}+b\cdot\left(\Delta\cdot r+E_q-\Delta\right)
 \ge\frac{c\cdot x_j}{y_j}.
\end{equation*}
Second, we have $e^{\Delta\cdot r}\cdot I_r=\int_0^\Delta e^{r\cdot s}\,ds$.
Therefore $e^{\Delta\cdot r}\cdot\left(\left(1-y\right)/y+t+I_r\right)$ is nondecreasing in $r$.
Using $I_q\ge I_r$ and $r\ge r_{j-1}$, we obtain
\begin{equation*}
 1-y+t\cdot y+y\cdot I_q
 \ge y\cdot e^{-\Delta\cdot r}\cdot
       \left(\frac1{y_j}-1+t+\Delta\right).
\end{equation*}
Hence both coefficients are nonnegative.
Since other factors are also nonnegative, we obtain
\begin{align*}
 V_t\!\left(y,x,q\right)
 &\ge R_q+\frac{y\cdot e^{-\Delta\cdot r}}{y_j}\cdot
                   V_{t+\Delta}\left(y_j,x_j,q\right)\\
 &\ge R_q+\frac{y\cdot e^{-\Delta\cdot r}}{y_j}\cdot
                   V_{t+\Delta}\left(y_j,x_j\right)
 \ge R_q+\phi_j\cdot V_{t+\Delta}\left(y_j,x_j\right).
\end{align*}
For the last inequality, we use $r\ge r_{j-1}$, $V_{t+\Delta}\le0$, and
$\phi_j=y\cdot e^{-\Delta\cdot r_{j-1}}/y_j$.
\end{proof}

It remains to lower bound $R_q$.
Since $0\le I_q\le\Delta$ and every term in $R_q$ other than $p\cdot I_q$ is nonnegative, we have $R_q\ge\min\!\left\{p\cdot\Delta,0\right\}$.
For a finite upper bound on the mean of $q$, the following lemma gives two stronger estimates.

\begin{lemma}\label{lem:iid-running-finite}
Suppose that $\frac1\Delta\cdot \int_t^{t+\Delta}q(s)\,ds\le r<\infty$.
Let
\begin{equation*}
 \alpha=\left(y\cdot b\cdot e^{-\Delta\cdot r}
       +\left(-p\right)^+\cdot\frac{I_r}{\Delta}\right)\cdot\Delta^2,
 \qquad \beta=\Delta\cdot\left(u+y\cdot b\cdot I_r\right).
\end{equation*}
We have $\alpha,\beta>0$ and
\begin{equation*}
 R_q\ge\min\!\left\{p\cdot\Delta,p\cdot I_r\right\}
 +\max\!\left\{\beta-\frac{\beta^2}{4\cdot\alpha},
 \frac\alpha4\cdot\left(3+2\cdot\log\left(\frac\beta\alpha\right)\right)\right\}.
\end{equation*}
Moreover, when $\beta\cdot e^{-r}>\alpha/2$, we also have
\begin{equation*}
 R_q\ge\min\!\left\{p\cdot\Delta,p\cdot I_r\right\}
 +\beta\cdot e^{-r}+\frac{\alpha\cdot r}{2}
 -\frac{\alpha^2}{24\cdot\left(\beta\cdot e^{-r}-\alpha/2\right)}.
\end{equation*}
\end{lemma}
\begin{proof}
Since $q$ is nondecreasing with average at most
$r$, we have $-\log S_t\left(t+s\right)\le s\cdot r$.
Using the concavity of $1-e^{-z}$, we obtain
\begin{align*}
 I_r\le I_q\le\Delta, \qquad
 S_t\left(t+s\right)\ge e^{-\Delta\cdot r},\qquad
 1-S_t\left(t+s\right)\ge-\frac{I_r}{\Delta}\cdot\log S_t\left(t+s\right).
\end{align*}
Since $I_r\le I_q\le\Delta$, we have
\begin{equation*}
 p\cdot I_q\ge\min\!\left\{p\cdot\Delta,p\cdot I_r\right\}
 +\left(-p\right)^+\cdot
   \int_0^\Delta\left(1-S_t\left(t+s\right)\right)\,ds.
\end{equation*}
For $p<0$, we have equality; for $p\ge0$, we use $I_q\ge I_r$.
Then we have
\begin{align*}
R_q & \ge
\min\!\left\{p\cdot\Delta,p\cdot I_r\right\}
-\left(y\cdot b\cdot e^{-\Delta\cdot r}
+\left(-p\right)^+\cdot\frac{I_r}{\Delta}\right)
\cdot \int_t^{t+\Delta}\log S_t(s)\,ds
+\left(u+y\cdot b\cdot I_r\right)\cdot E_q \\
& \ge\min\!\left\{p\cdot\Delta,p\cdot I_r\right\}
    +\int_0^1\left(\alpha\cdot\left(1-s\right)\cdot q\left(t+\Delta\cdot s\right)
             +\beta\cdot e^{-q\left(t+\Delta\cdot s\right)}\right)\,ds.
\end{align*}
where the first inequality follows from $1-S_t(s)\ge
-\frac{I_r}{\Delta}\cdot\log S_t(s)$, $S_t(s)\ge e^{-\Delta\cdot r}$ and $I_q\ge I_r$, and in the second inequality we use
\begin{equation*}
 -\int_t^{t+\Delta}\log S_t(s)\,ds
 =\int_0^\Delta\left(\Delta-s\right)\cdot q\left(t+s\right)\,ds
 =\Delta^2\cdot\int_0^1\left(1-s\right)\cdot q\left(t+\Delta\cdot s\right)\,ds.
\end{equation*}
We next bound the integral containing $\alpha$ and $\beta$.
For every $z>0$ and $w\ge0$, we have
\begin{equation*}
 \beta\cdot e^{-w}+z\cdot w
 =z\cdot\left(\frac\beta z\cdot e^{-w}+w\right)
 \ge z\cdot\left(1+\log\left(\frac\beta z\right)\right).
\end{equation*}
For any $\ell\ge0$, we apply this with $z=\ell+\alpha\cdot\left(1-s\right)$ and $w=q\left(t+\Delta\cdot s\right)$.
We subtract $\ell\cdot q\left(t+\Delta\cdot s\right)$ and integrate, using $\int_0^1q\left(t+\Delta\cdot s\right)\,ds\le r$, to obtain
\begin{equation*}
 R_q\ge\min\!\left\{p\cdot\Delta,p\cdot I_r\right\}
 +\frac1\alpha\cdot\int_\ell^{\ell+\alpha}
 z\cdot\left(1+\log\left(\frac\beta z\right)\right)dz-\ell\cdot r.
\end{equation*}
Taking $\ell=0$ and evaluating the integral gives the logarithmic bound in the first assertion.

For the other bound, we also use $e^{-w}\ge1-w$.
When $z\ge\beta$, we obtain
\begin{equation*}
 \beta\cdot e^{-w}+z\cdot w
 \ge\beta+\left(z-\beta\right)\cdot w\ge\beta.
\end{equation*}
We apply this bound when $z\ge\beta$ and the first pointwise bound when $z<\beta$, with $z=\alpha\cdot\left(1-s\right)$ and $w=q\left(t+\Delta\cdot s\right)$.
After integration and a change of variables, we bound the deficit of the integral from $\beta$ by
\begin{align*}
 \frac1\alpha\cdot\int_0^{\min\!\left\{\alpha,\beta\right\}}
 \left(\beta-z\cdot\left(1+\log\left(\frac\beta z\right)\right)\right)dz
 \le\frac1\alpha\cdot\int_0^\beta
 \left(\beta-z\cdot\left(1+\log\left(\frac\beta z\right)\right)\right)dz
 =\frac{\beta^2}{4\cdot\alpha}.
\end{align*}
We extend the integration range because the integrand is nonnegative on $\left[0,\beta\right]$.
Thus we obtain $\beta-\beta^2/\left(4\cdot\alpha\right)$ as another lower bound for the integral, proving the first assertion.

For the second assertion, take $\ell=\beta\cdot e^{-r}-\alpha/2>0$ in the shifted integral bound above.
We apply Taylor's inequality to
$z\mapsto z\cdot\left(1+\log\left(\beta/z\right)\right)$
on $\left[\ell,\ell+\alpha\right]$.
Its second derivative is $-1/z\ge-1/\ell$, and at the midpoint $\beta\cdot e^{-r}$ its first derivative is $r$.
Thus,
\begin{equation*}
 z\cdot\left(1+\log\left(\frac\beta z\right)\right)
 \ge\beta\cdot e^{-r}\cdot\left(1+r\right)
 +r\cdot\left(z-\beta\cdot e^{-r}\right)
 -\frac{\left(z-\beta\cdot e^{-r}\right)^2}{2\cdot\ell}.
\end{equation*}
We substitute $z=\ell+\alpha\cdot\left(1-s\right)$ and integrate over $s\in\left[0,1\right]$.
We obtain zero from the linear term and $\alpha^2/12$ from the squared term.
Consequently,
\begin{align*}
 \frac1\alpha\cdot\int_\ell^{\ell+\alpha}
 z\cdot\left(1+\log\left(\frac\beta z\right)\right)dz-\ell\cdot r
 &\ge\beta\cdot e^{-r}\cdot\left(1+r\right)
       -\frac{\alpha^2}{24\cdot\ell}-\ell\cdot r\\
 &=\beta\cdot e^{-r}+\frac{\alpha\cdot r}{2}
       -\frac{\alpha^2}{24\cdot\left(\beta\cdot e^{-r}-\alpha/2\right)}.
\end{align*}
Substituting into the shifted integral bound proves the second assertion.
\end{proof}

For each finite interval $\left[r_{j-1},r_j\right]$, we take $a_j$ to be the largest of $\min\!\left\{p\cdot\Delta,0\right\}$ and the applicable lower bounds in Lemma~\ref{lem:iid-running-finite} with $r=r_j$; for the last, unbounded interval, we use $a_{N_q}=\min\!\left\{p\cdot\Delta,0\right\}$.
We apply Lemma~\ref{lem:iid-state-step} to each rate and take the infimum over all allowed rates to obtain
\begin{equation*}
 V_t\!\left(y,x\right)\ge
 \min_{j\in[N_q]}\left\{a_j+\phi_j\cdot V_{t+\Delta}\left(y_j,x_j\right)\right\}.
\end{equation*}

\subsection{Proof of the Finite-input Transfer}
\label{app:iid-transfer}

Let $\theta:\left[0,1\right]\to\left[0,\infty\right)$ be
Lipschitz and strictly decreasing, with $\theta\left(1\right)=0$. Its initial
height and drop rate are
 $\theta\left(0\right)\le1,~ b(t)=-\theta'(t)~\text{for almost every }t.$
Let $B<\infty$ satisfy
$0\le b(t)\le B$ almost everywhere.
The internal clock at rank $i$ is
$\sigma_i=\left(i-U\right)/n$, where $U\sim\operatorname{Unif}\left(0,1\right)$
is independent of the observations.

For the comparison below, recall the following quantities, evaluated
at an arbitrary bounded nonnegative nondecreasing rate $q$:
\begin{align*}
S(t)&=\exp\!\left(-\int_0^t q(s)\,ds\right),
&
A(t)&=\int_t^1 S(s)\,ds,\\
k(t)&=S(t)-q(t)\cdot A(t)-A\left(0\right)\cdot e^{-q(t)},\\
L&=\theta\left(0\right)+\left(1-\theta\left(0\right)\right)\cdot A\left(0\right)-\int_0^1 b(t)\cdot k(t)\,dt.
\end{align*}
We write $L\!\left(b,q\right)$ only when its dependence on
the threshold's drop rate and the admissible rate needs to be
explicit.

For $t\ge s$, monotonicity gives
$S(t)\le S(s)\cdot e^{-q(s)\cdot\left(t-s\right)}$. Direct integration yields:
\begin{equation}
\label{eq:apx-qA}
 0\le q(s)\cdot A(s)
 \le S(s)\cdot\left(1-e^{-q(s)\cdot\left(1-s\right)}\right)\le S(s).
\end{equation}

\begin{lemma}[Finite-input transfer]\label{prop:iid-transfer}
Suppose that $\theta\left(0\right)+0.142\cdot B\le1$. Let $\gamma\in\mathbb R$ satisfy
$L\!\left(b,q\right)\ge\gamma$ for every bounded nonnegative nondecreasing
rate $q$. Then for every $n\ge1$ and every normalized i.i.d. input, the shifted-rank algorithm satisfies $\mathbb E\left[\mathrm{ALG}\right]\ge\gamma$.
\end{lemma}

\paragraph{Discrete notation and proof organization.}
Taking $q\equiv0$ in the hypothesis of Lemma~\ref{prop:iid-transfer} gives
$\gamma\le L\left(b,0\right)=1$.
If $n=1$, the last-rank rule gives
$\mathbb E\left[\mathrm{ALG}\right]=\mathbb E[X_1]=1\ge\gamma$.
For the auxiliary statements below, fix $n\ge2$ and put
$\delta=1/n$ and $\ell=1-\delta$.
Besides the continuous quantities above, we need the following discrete quantities.

For any nonnegative nondecreasing rate
$q:\left[0,1\right]\to\left[0,\infty\right]$, define
\begin{equation*}
P(t)=\exp\!\left(
-\delta\cdot\sum_{j=1}^{\left\lfloor t/\delta\right\rfloor}q\left(t-j\cdot\delta\right)
\right),
\qquad 0\le t\le1+\delta.
\end{equation*}
For $t\in\left[0,1\right]$, set
\begin{equation*}
A_P(t)=\int_t^1 P(s)\,ds,
\qquad
\overline P(t)=\frac1\delta\cdot\int_t^{t+\delta}P(s)\,ds.
\end{equation*}
An empty sum is zero, and $e^{-\infty}=0$.

%The quantities $A_P$ and $\overline P$ are respectively an integral and a moving average of $P$.
Define the finite selection expression by
\begin{equation*}
 T_n(t)=
 \begin{cases}
 1-\overline P(t)+\dfrac{1-e^{-\delta\cdot q(t)}}\delta\cdot A_P(t),&0\le t\le\ell,\\[3pt]
 1-e^{-\delta\cdot q(t)}\cdot\overline P(\ell),&\ell<t\le1.
 \end{cases}
\end{equation*}
The second case accounts for unconditional acceptance at the last
rank.
%At $a=\ell$, the formulas agree because $A_P(\ell)=\delta\cdot\overline P(\ell)$.

For a bounded rate $q$, introduce the auxiliary quantities
\begin{equation*}
 C(q)=1-\theta\left(0\right)+\int_0^1b(t)\cdot e^{-q(t)}\,\mathrm d t,
 \qquad T_\infty(t)=1-S(t)+q(t)\cdot A(t).
\end{equation*}
%The first measures the normalization term; for the logarithmic rate of an actual input, it is the expected surplus of $X^*$ above $\theta\left(0\right)$, as the next lemma shows. The second is the continuous selection expression at cutoff $\theta(a)$: its two terms account for selection before $a$ and selection above the cutoff after $a$.
By \eqref{eq:apx-qA}, $0\le T_\infty\le1$.
Expanding the definition of $k$ and using
$\int_0^1b(t)\,\mathrm d t=\theta\left(0\right)$ gives the identity
\begin{equation*}
 L\!\left(b,q\right)=C(q)\cdot A\left(0\right)+\int_0^1b(t)\cdot T_\infty(t)\,\mathrm d t.
\end{equation*}
%The formula for $C$ also applies to extended-valued rates,
%with $e^{-\infty}=0$; the continuous quantities $S,A,k,T_\infty,L$
%are used only for bounded rates.
Define the comparison errors by
\begin{equation*}
 D_n(t)=T_\infty(t)-T_n(t),\qquad
 \mathcal E=A_P\left(0\right)-A\left(0\right)=\int_0^1\left(P(t)-S(t)\right)\,\mathrm d t.
\end{equation*}
Here $D_n(t)$ is the difference between the continuous and finite
selection expressions and may be negative.
The quantity $\mathcal E$ is the gain in survival area from the
shifted grid.

\medskip
The proof has two levels. Lemma~\ref{lem:apx-accounting} reduces
the actual-input guarantee to a comparison of these finite and
continuous quantities. Lemma~\ref{lem:apx-comparison} supplies that
comparison from four technical lemmas. The short proof of
Lemma~\ref{prop:iid-transfer} at the end then combines these two
results and passes to finite step approximations of the actual rate.

\begin{lemma}[Finite-input payoff accounting]\label{lem:apx-accounting}
Let $X_1,\ldots,X_n$ be nonnegative i.i.d. observations with $\mathbb E\left[X^*\right]=1$ and $n\ge2$. Let $F$ be their common distribution function.
Use the shifted-rank algorithm and threshold $\theta$ defined above. 
Set
$r(t)=-n\cdot\log F\left(\theta(t)\right), -\log 0=\infty.$
Throughout this lemma and its proof, $P$, $A_P$, $\overline P$,
and $T_n$ are computed from $r$.
Then $r$ is nonnegative and nondecreasing.

For every $t\in\left[0,1\right]$,
\begin{equation*}
T_n(t)=\mathbb P\left(\mathrm{ALG}>\theta(t)\right).
\end{equation*}
The excess term satisfies
\begin{equation*}
C(r)=\mathbb E\left[\left(X^*-\theta\left(0\right)\right)^+\right].
\end{equation*}

Moreover,
\begin{equation}\label{eq:apx-finite-payoff}
 \mathbb E\left[\mathrm{ALG}\right]
 \ge C(r)\cdot A_P\left(0\right)+\int_0^1b(t)\cdot T_n(t)\,\mathrm d t.
\end{equation}
\end{lemma}

\begin{proof}
We have $e^{-r(t)}=F\left(\theta(t)\right)^n=\mathbb P\left(X^*\le\theta(t)\right)$.
Then $r$ is nondecreasing since $\theta$ decreases and $F$
increases.
If $R_i$ is the probability of reaching rank $i$ conditional
on $U$, since $\left\lfloor \sigma_i/\delta \right\rfloor=i-1$ and $\sigma_i-j\cdot\delta =\sigma_{i-j}$, independence gives
\begin{equation*}
 R_i=\prod_{k<i}F\left(\theta\left(\sigma_k\right)\right)=P\left(\sigma_i\right).
\end{equation*}
Thus $P\left(\sigma_i\right)$ is the probability of reaching actual rank $i$.
Values of $P$ after time $1$ are an auxiliary extension used for
the comparison, not survival after the algorithm's forced acceptance.

For a cutoff $\theta(t)$, fix $U$, and let $j$ be the first rank whose clock
satisfies $\sigma_j\ge t$, capped at $n$. Any acceptance before $j$
exceeds the cutoff. At every reached rank $i\ge j$, a value above
the cutoff is selected, including at the last rank. Therefore,
\begin{equation*}
 \mathbb P\left(\mathrm{ALG}>\theta(t)\mid U\right)
 =1-R_j+\left(1-F\left(\theta(t)\right)\right)\cdot\sum_{i=j}^nR_i.
\end{equation*}

For $t\le\ell$, the clock at rank $j$ is uniform on $\left[t, t+\delta \right)$
as $U$ varies. Writing $\mathbb E_U$ for this phase average gives
\begin{equation*}
 \mathbb E_U[R_j]=\overline P(t),\qquad
 \mathbb E_U\left[\sum_{i=j}^nR_i\right]=\frac1\delta\cdot A_P(t).
\end{equation*}
For $t>\ell$, the cap gives $j=n$ and
$\mathbb E_U[R_n]=\overline P(\ell)$.
Using $F\left(\theta(t)\right)=e^{-\delta\cdot r(t)}$ in the conditional identity therefore
gives the two cases in the definition of $T_n$:
\begin{equation*}
 \mathbb P\left(\mathrm{ALG}>\theta(t)\right)=T_n(t).
\end{equation*}
The factor $F\left(\theta(t)\right)$ includes the mass at the cutoff, consistently
with the strict acceptance rule.

For the actual input, the tail
formula for expectation and the substitution $x=\theta(t)$ give
\begin{align*}
 C(r)&=1-\int_0^{\theta\left(0\right)}\left(1-F(x)^n\right)\,\mathrm d x\\
 &=\mathbb E\left[\left(X^*-\theta\left(0\right)\right)^+\right]
 \le n\cdot\mathbb E\left[\left(X_1-\theta\left(0\right)\right)^+\right].
\end{align*}

Since every threshold is at most $\theta\left(0\right)$, any value above $\theta\left(0\right)$ at a
reached rank is selected. Independence of $X_i$ from the event of
reaching rank $i$, conditional on $U$, yields
\begin{align*}
 \mathbb E\left[\left(\mathrm{ALG}-\theta\left(0\right)\right)^+\right]
 &=\mathbb E\left[\left(X_1-\theta\left(0\right)\right)^+\right]\cdot\mathbb E_U\left[\sum_{i=1}^nR_i\right]\\
 &=n\cdot\mathbb E\left[\left(X_1-\theta\left(0\right)\right)^+\right]\cdot A_P\left(0\right)\\
 &\ge C(r)\cdot A_P\left(0\right).
\end{align*}

On the other hand,
\begin{align*}
 \mathbb E\left[\min\left\{\mathrm{ALG},\theta\left(0\right)\right\}\right]
 &=\int_0^{\theta\left(0\right)}\mathbb P\left(\mathrm{ALG}>x\right)\,\mathrm d x\\
 &=\int_0^1b(t)\cdot T_n(t)\,\mathrm d t.
\end{align*}

Adding the two parts proves the payoff bound in the statement.
\end{proof}

\paragraph{The comparison target.}
For a bounded rate $q$, the difference between the finite-input
expression and $L\left(b,q\right)$ is
\begin{equation*}
C(q)\cdot A_P\left(0\right)+\int_0^1b(t)\cdot T_n(t)\,dt-L\left(b,q\right)
=
C(q)\cdot\mathcal E-\int_0^1b(t)\cdot D_n(t)\,dt.
\end{equation*}

\begin{sloppypar}
By monotonicity,
\begin{equation*}
\delta\cdot q\left(t-j\cdot\delta\right)
\le
\int_{t-j\cdot\delta}^{t-\left(j-1\right)\cdot\delta}q(s)\,ds
\end{equation*}
for every summand in the definition of $P(t)$.
Summing these inequalities gives $P(t)\ge S(t)$ on $\left[0,1\right]$,
and hence $\mathcal E\ge0$.
Since $C(q)\ge1-\theta\left(0\right)$ and $0\le b\le B$, the sufficient target is
\end{sloppypar}
\begin{equation*}
 \int_0^1D_n(t)^+\,\mathrm d t\le0.142\cdot\mathcal E.
\end{equation*}
This bound makes the above difference at least
$\left(1-\theta\left(0\right)-0.142\cdot B\right)\cdot\mathcal E\ge0$.

\medskip
The auxiliary lemmas address this target in four steps.
Lemma~\ref{lem:apx-last-rank} removes rate increases on the last
interval. Lemma~\ref{lem:apx-local} then controls all cutoff losses
by a terminal survival gap. Lemma~\ref{lem:apx-area} bounds the
available survival-area gain below in terms of a terminal stopping
weight, and Lemma~\ref{lem:apx-scalar} connects that weight to the
survival gap. Lemma~\ref{lem:apx-comparison} combines the estimates.

\paragraph{Shift identities used by the auxiliary lemmas.}
For a bounded nonnegative nondecreasing rate $q$, the shifted sums
give two useful identities, for $t\in\left[0,1\right]$:
\begin{equation}
\label{eq:apx-shift-mean}
 P\left(t+\delta \right)=e^{-\delta\cdot q(t)}\cdot P(t),\qquad
 \frac1\delta\cdot\int_t^{t+\delta }-\log P(s)\,\mathrm d s
       =\int_0^t q(s)\,\mathrm d s.
\end{equation}
For the second identity, substitute $u=s-j\cdot\delta $ in each summand of
$-\log P(s)$; the resulting intervals partition $\left[0,t\right]$ apart from
endpoints. In particular, Jensen's inequality gives
$\overline P(t)\ge S(t)$.

Whenever a step rate is used below, it has finitely many jumps,
uses its right-continuous representative, and takes its left-limit
value at $1$. Changes at these finitely many endpoints do not alter
the integrals or phase averages. Cutoff comparisons are asserted
almost everywhere when appropriate.

\begin{lemma}[Last-rank reduction]\label{lem:apx-last-rank}
Let $q$ be a finite nonnegative nondecreasing step rate, and put
$q_e(t)=q\left(\min\left\{t,\ell\right\}\right)$. Use a subscript $e$ for the quantities
computed from $q_e$. For example,
$D_{n,e}=T_{\infty,e}-T_{n,e}$ and
$\mathcal E_e=A_{P,e}\left(0\right)-A_e\left(0\right)$. Then,
\begin{equation*}
 P=P_e\text{ on }\left[0,1\right],\qquad
 \mathcal E\ge\mathcal E_e,\qquad
 D_n(t)\le D_{n,e}(t)\text{ for almost every }t\in\left[0,1\right].
\end{equation*}
In particular, the target $\int D_n^+\le0.142\cdot\mathcal E$ for $q_e$
implies the same target for $q$.
\end{lemma}

\begin{proof}
For $t\le1$, every rate argument in $P(t)$ is at most $\ell$.
Hence $P=P_e$ on $\left[0,1\right]$, while $S\le S_e$, which gives
$\mathcal E\ge\mathcal E_e$.
We next show that $D_n(t)\le D_{n,e}(t)$ almost everywhere.
For $t\le\ell$, the finite expression $T_n$ is unchanged, whereas
$S(t)=S_e(t)$ and $A(t)\le A_e(t)$, so the claim follows.

For $t\ge\ell$, put
$p_e=\overline P_e\left(1\right)$, $s_e=S_e\left(1\right)$, and
$\eta_t=q(t)-q(\ell)\ge0$. The rate $q_e$ is constant on the
last interval.
Integrating the shift identity in \eqref{eq:apx-shift-mean} over $\left[\ell,1\right]$ gives
\begin{equation*}
p_e=\overline P_e\left(1\right)
=\frac1\delta\cdot\int_\ell^1 P_e\left(t+\delta\right)\,dt
=e^{-\delta\cdot q(\ell)}\cdot\overline P_e(\ell).
\end{equation*}
Hence the definition of $T_n$ gives
 $T_n(t)=1-p_e\cdot e^{-\delta\cdot\eta_t}.$

On the other hand, \eqref{eq:apx-qA} implies
\begin{align*}
 T_\infty(t)
 &\le1-S(t)\cdot e^{-q(t)\cdot\left(1-t\right)}\\
 &=1-\exp\left(-\int_0^1\min\left\{q(s),q(t)\right\}\,\mathrm d s\right)\\
 &\le1-s_e\cdot e^{-\delta\cdot\eta_t}.
\end{align*}
For the last inequality, the truncated integrand agrees with $q_e$
before $\ell$ and exceeds it by at most $\eta_t$ on an interval
of length $\delta $. Jensen's inequality in \eqref{eq:apx-shift-mean}
gives $p_e\ge s_e$. Also,
$T_{\infty,e}(t)=1-s_e$ and $T_{n,e}(t)=1-p_e$ on $\left[\ell,1\right]$.
Consequently,
\begin{equation*}
 D_n(t)\le e^{-\delta\cdot\eta_t}\cdot\left(p_e-s_e\right)\le p_e-s_e=D_{n,e}(t).
\end{equation*}

Taking positive parts and integrating gives $\int D_n^+\le\int D_{n,e}^+$.
Together with $\mathcal E\ge\mathcal E_e$, this proves the final assertion.
\end{proof}

\begin{lemma}[Local selection deficit]\label{lem:apx-local}
Let $q$ be a finite nonnegative nondecreasing step rate, constant
on $\left[\ell,1\right]$. Put $\Lambda=q(\ell)>0$,
$p=\overline P\left(1\right)$, and $s=S\left(1\right)$. Thus $p$ is shifted-grid survival
averaged over $\left[1,1+\delta \right]$, and $s$ is continuous survival at time $1$.
Then $p\ge s$ and
\begin{align*}
 D_n(t)&\le\frac{q(t)}{\Lambda}\cdot\left(p-s\right)
       \quad\text{for almost every }t\in\left[0,1\right],
       \\
 \int_0^1D_n(t)^+\,\mathrm d t
 &\le\frac{p-s}{\Lambda}\cdot\int_0^1q(t)\,\mathrm d t.
\end{align*}
\end{lemma}

\begin{proof}
Jensen's inequality in \eqref{eq:apx-shift-mean} at $t=1$ gives
$p\ge s$. We prove the pointwise bound by showing that the loss per unit
of rate, $D_n(t)/q(t)$, is nondecreasing as the cutoff is lowered.
For $t\le\ell$ and $q(t)>0$, its exact expression is
\begin{equation*}
 \frac{D_n(t)}{q(t)}
 =\frac{\overline P(t)-S(t)}{q(t)}
      +A(t)-\frac{1-e^{-\delta\cdot q(t)}}{\delta\cdot q(t)}\cdot A_P(t).
\end{equation*}
First, \eqref{eq:apx-shift-mean} and differentiation of integrals give
\begin{equation*}
 \overline P'=-\frac{1-e^{-\delta\cdot q}}{\delta }\cdot P,\qquad S'=-q\cdot S,\qquad
 A'=-S,\qquad A_P'=-P
\end{equation*}
almost everywhere. All four functions are absolutely continuous.
Thus $D_n/q$ is constant on every interval where
$q$ is constant and positive, even if $P$ itself jumps there.

At a jump time $t\le\ell$, hold the four continuous quantities
in the above expression for $D_n(t)/q(t)$ fixed and vary the rate through $u>0$,
between its old and new values. The derivative of the right side is
\begin{equation*}
 \frac{\frac{1-\left(1+\delta\cdot u\right)\cdot e^{-\delta\cdot u}}\delta\cdot A_P(t)
                    -\overline P(t)+S(t)}{u^2}.
\end{equation*}
It is therefore enough to establish the jump-gap bound
\begin{equation*}
 0\le\overline P(t)-S(t)
 \le\frac{1-\left(1+\delta\cdot u\right)\cdot e^{-\delta\cdot u}}\delta\cdot A_P(t).
\end{equation*}

This bound follows from the mean and range of the shifted sum.
We first record the elementary exponential-moment estimate
\begin{equation*}
 \mathbb E\left[e^Z\right]\cdot e^{-\mathbb E[Z]}\le\frac{e^x}{1+x}
\end{equation*}
whenever the essential range of $Z$ has length at most $x\ge0$.
Here essential range ignores null events. The ratio on the left
is unchanged by translating $Z$, so assume $0\le Z\le x$.
The function $y\mapsto e^y/\left(1+y\right)$ is nondecreasing for $y\ge0$,
hence
\begin{equation*}
 \mathbb E\left[e^Z\right]\le\frac{e^x}{1+x}\cdot\left(1+\mathbb E[Z]\right)
                 \le\frac{e^x}{1+x}\cdot e^{\mathbb E[Z]}.
\end{equation*}
This proves the estimate, including $x=0$.

Let $H(t)=-\log P(t)$ and let $V$ be uniform on $\left[t,t+\delta \right)$.
By \eqref{eq:apx-shift-mean},
$\mathbb E\left[H(V)\right]=\int_0^t q(s)\,\mathrm d s$.
The function $H$ is nondecreasing. If $t>0$,
\begin{equation*}
 H\left(\left(t+\delta \right)-\right)=H\left(t-\right)+\delta\cdot q\left(t-\right),\qquad H\left(t+\right)\ge H\left(t-\right).
\end{equation*}
Hence its essential range on this interval has length at most
$\delta\cdot q\left(t-\right)\le \delta\cdot u$. At $t=0$, $H(V)=0$  almost surely. Jensen's
inequality and the exponential-moment bound, with $Z=-H(V)$, give
\begin{align*}
 0\le\overline P(t)-S(t)
 \le\left(1-\left(1+\delta\cdot u\right)\cdot e^{-\delta\cdot u}\right)\cdot\overline P(t)
 \le\frac{1-\left(1+\delta\cdot u\right)\cdot e^{-\delta\cdot u}}\delta\cdot A_P(t),
\end{align*}
because $A_P(t)\ge \delta\cdot\overline P(t)$ for $t+\delta \le1$.
The single endpoint $t+\delta $ can include the new rate.
Using the essential range is what makes the bound valid for every intermediate $u$.
This proves the required jump-gap bound, and consequently
$D_n/q$ cannot decrease at a positive-rate jump either.

Finally, constancy of $q=\Lambda$ on $\left[\ell,1\right]$ gives
$T_\infty(t)=1-s$ there, while
$p=e^{-\delta\cdot\Lambda}\cdot\overline P(\ell)$ gives $T_n(t)=1-p$.
Thus the terminal value of $D_n/q$ is $\left(p-s\right)/\Lambda$.
The monotonicity just proved yields the pointwise bound in the statement
wherever $q(t)>0$. If $q(t)=0$, all earlier rates are zero and
both selection expressions vanish. Since $p\ge s$, taking
positive parts and integrating proves the integral bound.
\end{proof}

\begin{lemma}[Survival-area gain]\label{lem:apx-area}
Let $q$ be a finite nonnegative nondecreasing step rate, constant
on $\left[\ell,1\right]$. Put $\Lambda=q(\ell)>0$, $z=\delta\cdot\Lambda$, and
$p=\overline P\left(1\right)$. Then
\begin{equation*}
 \mathcal E\ge\frac{\delta\cdot\left(1-p\right)}{1+\sqrt{1+e^{-z}}}.
\end{equation*}
\end{lemma}

\begin{proof}
Jensen's inequality in \eqref{eq:apx-shift-mean} gives
$S(t)\le\overline P(t)$. Define
\begin{equation*}
 w(u)=1-P\left(1+\delta\cdot u\right),\qquad 0\le u\le1,
\end{equation*}
so that $\int_0^1w(u)\,\mathrm d u=1-p$.
We first express the available area in terms of this terminal
weight. After interchanging the integrals, the interior contributions
cancel, leaving only the initial and terminal intervals. Since
$P(t)=1$ for $0\le t<\delta $, this gives
\begin{align*}
 \mathcal E
 &\ge\int_0^1\left(P(t)-\overline P(t)\right)\,\mathrm d t\\
 &=\frac \delta 2-\frac1\delta\cdot\int_1^{1+\delta }\left(1+\delta -t\right)\cdot P(t)\,\mathrm d t\\
 &=\delta\cdot\int_0^1\left(1-u\right)\cdot w(u)\,\mathrm d u.
\end{align*}
Thus it remains to compare a weighted integral of $w$ with its
unweighted integral.

The function $P$ is nonincreasing because $q$ is nonnegative and
nondecreasing. The shift identity gives
$P\left(1+\delta \right)=e^{-z}\cdot P\left(1\right)$, and the first summand of $-\log P\left(1\right)$ gives
$P\left(1\right)\le e^{-z}$. Therefore
\begin{equation*}
 w\left(0\right)\le w(u)\le1-e^{-z}\cdot P\left(1\right)
                 \le\left(1+e^{-z}\right)\cdot w\left(0\right).
\end{equation*}
Put $c=\left(1+\sqrt{1+e^{-z}}\right)^{-1}$, so that
$\left(1-c\right)^2=\left(1+e^{-z}\right)\cdot c^2$. On $\left[0,1-c\right]$ the weight $1-u-c$ is
nonnegative, and on $\left[1-c,1\right]$ it is nonpositive. Applying the lower
bound for $w$ on the first interval and the upper bound on the
second gives
\begin{align*}
 \int_0^1\left(1-u-c\right)\cdot w(u)\,\mathrm d u
 &\ge w\left(0\right)\cdot\left(
       \int_0^{1-c}\left(1-u-c\right)\,\mathrm d u
       +\left(1+e^{-z}\right)\cdot\int_{1-c}^1\left(1-u-c\right)\,\mathrm d u\right)\\
 &=\frac{w\left(0\right)}2\cdot\left(\left(1-c\right)^2-\left(1+e^{-z}\right)\cdot c^2\right)=0.
\end{align*}
Consequently,
\begin{equation*}
 \mathcal E\ge \delta\cdot c\cdot\int_0^1w(u)\,\mathrm d u
       =\frac{\delta\cdot\left(1-p\right)}{1+\sqrt{1+e^{-z}}}. \qedhere
\end{equation*}
\end{proof}

\begin{lemma}[Terminal scalar comparison]\label{lem:apx-scalar}
Let $q$ be a finite nonnegative nondecreasing step rate, constant
on $\left[\ell,1\right]$. Put $\Lambda=q(\ell)>0$, $z=\delta\cdot\Lambda$,
$p=\overline P\left(1\right)$, and $s=S\left(1\right)$. Then $s\le p<1$ and
\begin{equation*}
 \frac{\int_0^1q(t)\,\mathrm d t}{\delta\cdot\Lambda}
       \cdot\frac{p-s}{1-p}\le\min\left\{z/8,1/16\right\}.
\end{equation*}
\end{lemma}

\begin{proof}
Jensen's inequality gives $p\ge s$. For $t\in\left[1,1+\delta \right)$, the first
term of $-\log P(t)$ is $\delta\cdot q\left(t-\delta \right)=z$, so $p\le e^{-z}<1$.
Let $V$ be uniform
on $\left[1,1+\delta \right)$ and put $H(t)=-\log P(t)$. The mean identity in
\eqref{eq:apx-shift-mean} gives
$\mathbb E\left[H(V)\right]=\int_0^1q(t)\,\mathrm d t$. Monotonicity and the shift identity give
\begin{equation*}
 H\left(\left(1+\delta \right)-\right)-H\left(1\right)\le \delta\cdot\Lambda=z,
\end{equation*}
so its essential range has length at most $z$. Its minimum is $H\left(1\right)\ge z$.
Use the containing interval $\left[H\left(1\right),H\left(1\right)+z\right]$. Write its lower
endpoint and the mean offset in units of $z$ as
\begin{equation*}
 m=\frac{H\left(1\right)}z\ge1,\qquad
 \xi=\frac1z\cdot\int_0^1q(t)\,\mathrm d t-m\in\left[0,1\right].
\end{equation*}
By convexity of $h\mapsto e^{-h}$ on $\left[m\cdot z,\left(m+1\right)\cdot z\right]$,
we have
\begin{equation*}
 e^{-H(V)}\le e^{-m\cdot z}\cdot\left(1-
       \frac{H(V)-m\cdot z}{z}\cdot\left(1-e^{-z}\right)\right).
\end{equation*}
Taking expectations gives
\begin{equation*}
 p\le e^{-m\cdot z}\cdot\left(1-\xi+\xi\cdot e^{-z}\right)=:p_*,\qquad
 s=e^{-\left(m+\xi\right)\cdot z}.
\end{equation*}
Since $\left(p-s\right)/\left(1-p\right)$ increases with $p$, the desired estimate
reduces to the elementary inequality
\begin{equation}
\label{eq:apx-elementary-scalar}
 \frac{\left(m+\xi\right)\cdot\left(p_*-s\right)}{1-p_*}\le\min\left\{z/8,1/16\right\},
 \qquad m\ge1,\quad \xi\in\left[0,1\right],\quad z>0.
\end{equation}

The numerator is zero for $\xi=0,1$, so take $0<\xi<1$. Set
$c=1-\xi+\xi\cdot e^{-z}$. Convexity gives $c>e^{-\xi\cdot z}$.

First, the function
\begin{equation*}
 \xi\longmapsto 1-\xi+\xi\cdot e^{-z}-e^{-\xi\cdot z}
                 -\frac12\cdot\xi\cdot\left(1-\xi\right)\cdot z^2
\end{equation*}
has second derivative $z^2\cdot\left(1-e^{-\xi\cdot z}\right)\ge0$ and vanishes at both
endpoints. It is therefore nonpositive on $\left[0,1\right]$. Also,
\begin{equation*}
 e^{m\cdot z}-c
 \ge\left(m+\xi\right)\cdot z+\frac12\cdot\left(m^2-\xi\right)\cdot z^2
 \ge\left(m+\xi\right)\cdot z,
\end{equation*}
using $e^{m\cdot z}\ge1+m\cdot z+m^2\cdot z^2/2$, $e^{-z}\le1-z+z^2/2$, and $m^2\ge \xi$.
Consequently,
\begin{equation*}
 \frac{\left(m+\xi\right)\cdot\left(p_*-s\right)}{1-p_*}
 =\frac{\left(m+\xi\right)\cdot\left(c-e^{-\xi\cdot z}\right)}{e^{m\cdot z}-c}
 \le\frac12\cdot\xi\cdot\left(1-\xi\right)\cdot z\le\frac z8.
\end{equation*}

For the bound $1/16$, hold $\xi,z$ fixed. The logarithmic derivative of
the positive expression on the left of \eqref{eq:apx-elementary-scalar} is
\begin{equation*}
 \frac1{m+\xi}-\frac z{1-p_*}<0,
\end{equation*}
because $p_*\ge s$ and $1-s<\left(m+\xi\right)\cdot z$. Thus it suffices to prove the
bound at $m=1$. Put $x=e^{-z}\in\left(0,1\right)$. We bound the exponential chord gap
directly. For $0<u<1$, let
\begin{equation*}
 G(u)=\frac{1-u+u\cdot x-x^u}{u\cdot\left(1-u\right)},
\end{equation*}
extended continuously to the endpoints. 
For $0<u<1$, direct integration gives
\begin{align*}
z^2\cdot\int_0^1\int_0^{1-v}e^{-z\cdot\left(t+u\cdot v\right)}\,dt\,dv
&=\frac{1-e^{-z\cdot u}}{u}
-\frac{e^{-z\cdot u}-e^{-z}}{1-u}\\
&=\frac{1-u+u\cdot e^{-z}-e^{-z\cdot u}}{u\cdot\left(1-u\right)}
=G(u).
\end{align*}
For each fixed $t,v$, the integrand is convex in $u$.
The integral representation therefore shows that $G$ is convex.
Its endpoint values are
\begin{equation*}
G\left(0\right)=z-1+x,
\qquad
G\left(1\right)=1-x-x\cdot z.
\end{equation*}
Hence
\begin{align*}
 G(\xi)&\le\left(1-\xi\right)\cdot\left(z-1+x\right)+\xi\cdot\left(1-x-x\cdot z\right)\\
     &=\left[1-\xi\cdot\left(1+x\right)\right]\cdot z+\left(2\cdot\xi-1\right)\cdot\left(1-x\right).
\end{align*}
Since $m=1$ and $1-p_*=\left(1-x\right)\cdot\left(1+\xi\cdot x\right)$,
\begin{equation*}
 \frac{\left(1+\xi\right)\cdot\left(p_*-s\right)}{1-p_*}
 =\xi\cdot\left(1-\xi\right)\cdot\frac{\left(1+\xi\right)\cdot x\cdot G(\xi)}{\left(1-x\right)\cdot\left(1+\xi\cdot x\right)}.
\end{equation*}
We will bound the last fraction by $1/4$.

We use the elementary logarithm bounds
\begin{equation*}
 \frac{2\cdot\left(1-x\right)}{1+x}\le-\log x
 \le\frac{\left(1-x\right)\cdot\left(x^2+10\cdot x+1\right)}{6\cdot x\cdot\left(1+x\right)}.
\end{equation*}
Both gaps vanish at $x=1$, and their derivatives are nonpositive:
\begin{align*}
 \left(-\log x-\frac{2\cdot\left(1-x\right)}{1+x}\right)'
     &=-\frac{\left(1-x\right)^2}{x\cdot\left(1+x\right)^2},\\
 \left(\frac{\left(1-x\right)\cdot\left(x^2+10\cdot x+1\right)}{6\cdot x\cdot\left(1+x\right)}+\log x\right)'
     &=-\frac{\left(1-x\right)^4}{6\cdot x^2\cdot\left(1+x\right)^2}.
\end{align*}
This proves both bounds on $0<x\le1$. The sign of $1-\xi\cdot\left(1+x\right)$,
the coefficient of $z$ in the bound for $G(\xi)$, determines which
logarithm bound to substitute.

If $1-\xi\cdot\left(1+x\right)\le0$, use the lower logarithm bound in the estimate
for $G(\xi)$, obtaining $G(\xi)\le\left(1-x\right)^2/\left(1+x\right)$. Therefore
\begin{equation*}
 \frac{\left(1+\xi\right)\cdot x\cdot G(\xi)}{\left(1-x\right)\cdot\left(1+\xi\cdot x\right)}
 \le\frac{\left(1+\xi\right)\cdot x\cdot\left(1-x\right)}{\left(1+x\right)\cdot\left(1+\xi\cdot x\right)}
 \le\frac{2\cdot x\cdot\left(1-x\right)}{\left(1+x\right)^2}\le\frac14.
\end{equation*}
The middle inequality uses $\left(1+\xi\right)/\left(1+\xi\cdot x\right)\le2/\left(1+x\right)$, and the
last is equivalent to $\left(1-3\cdot x\right)^2\ge0$.

If $1-\xi\cdot\left(1+x\right)\ge0$, use the upper logarithm bound instead. After
collecting terms, the fraction is at most
\begin{equation*}
 \frac{\left(1+\xi\right)\cdot\left(1-x\right)\cdot\left[1+5\cdot x-\xi\cdot\left(1-x^2\right)\right]}{6\cdot\left(1+x\right)\cdot\left(1+\xi\cdot x\right)}.
\end{equation*}
To prove that this is at most $1/4$, expand the difference after
clearing the positive denominator:
\begin{align*}
 &3\cdot\left(1+x\right)\cdot\left(1+\xi\cdot x\right)-2\cdot\left(1+\xi\right)\cdot\left(1-x\right)\cdot\left[1+5\cdot x-\xi\cdot\left(1-x^2\right)\right]\\
 &\qquad=Q_{\xi}(x)+2\cdot\xi\cdot\left(1+\xi\right)\cdot x^3,
\end{align*}
where
\begin{equation*}
 Q_{\xi}(x)=\left(10+11\cdot\xi-2\cdot\xi^2\right)\cdot x^2-\left(5+7\cdot\xi+2\cdot\xi^2\right)\cdot x+1+2\cdot\xi^2.
\end{equation*}
This quadratic is positive for every real $x$: its leading
coefficient is positive, and the negative of its discriminant is
\begin{align*}
 &4\cdot\left(1+2\cdot\xi^2\right)\cdot\left(10+11\cdot\xi-2\cdot\xi^2\right)-\left(5+7\cdot\xi+2\cdot\xi^2\right)^2\\
 &\qquad=15-26\cdot\xi+3\cdot\xi^2+60\cdot\xi^3-20\cdot\xi^4\\
 &\qquad\ge5+4\cdot\xi+3\cdot\xi^2>0.
\end{align*}
The last inequality uses $\xi^4\le \xi^3$ and
$\xi^3\ge3\cdot\xi/4-1/4$, the latter following from
$\left(\xi-1/2\right)^2\cdot\left(\xi+1\right)\ge0$. The nonnegative cubic term then proves
the required bound in this case as well.

We have shown that the scalar expression is at most
$\xi\cdot\left(1-\xi\right)/4\le1/16$, completing \eqref{eq:apx-elementary-scalar}.

Since $m+\xi=\left(\delta\cdot\Lambda\right)^{-1}\cdot\int_0^1q(t)\,\mathrm d t$, this proves
the lemma's inequality.
\end{proof}

\begin{lemma}[Uniform comparison for step rates]\label{lem:apx-comparison}
For every finite nonnegative nondecreasing step rate $q$,
\begin{equation*}
 \mathcal E\ge0,\qquad
 \int_0^1D_n(t)^+\,\mathrm d t\le0.142\cdot\mathcal E.
\end{equation*}
\end{lemma}

\begin{proof}
We already observed that $P\ge S$, so $\mathcal E\ge0$.
By Lemma~\ref{lem:apx-last-rank},
it suffices to consider a rate constant on $\left[\ell,1\right]$. If
$\Lambda=q(\ell)=0$, the rate is zero and $D_n=0$. Otherwise put
$z=\delta\cdot\Lambda$, $p=\overline P\left(1\right)$, and $s=S\left(1\right)$.
Lemmas~\ref{lem:apx-local}--\ref{lem:apx-scalar} give
\begin{align*}
 \int_0^1D_n(t)^+\,\mathrm d t
 &\le\frac{p-s}{\Lambda}\cdot\int_0^1q(t)\,\mathrm d t\\
 &\le \delta\cdot\left(1-p\right)\cdot\min\left\{z/8,1/16\right\}\\
 &\le\left(1+\sqrt{1+e^{-z}}\right)
                \cdot\min\left\{z/8,1/16\right\}\cdot\mathcal E.
\end{align*}

Here $f(z)=1+\sqrt{1+e^{-z}}$ decreases, while
$ \left(z\cdot f(z)\right)'=f(z)-\frac{z\cdot e^{-z}}{2\cdot\sqrt{1+e^{-z}}}
 \ge2-\tfrac12\cdot z\cdot e^{-z}\ge\tfrac32>0.$
The two branches meet at $z=1/2$, giving the coefficient upper bound
\begin{equation*}
 \frac{1+\sqrt{1+e^{-1/2}}}{16}<\frac{71}{500}=0.142.
\end{equation*}

For an exact rational check, $e^{1/2}\ge13/8$ and
$\left(159/125\right)^2-21/13=528/203125>0$, so
$\sqrt{1+e^{-1/2}}<159/125$.
\end{proof}

\begin{proof}[Proof of Lemma~\ref{prop:iid-transfer}]
Fix $n\ge2$ and a normalized input with CDF $F$.
Use the auxiliary rate
$r(t)=-n\cdot\log F\left(\theta(t)\right)$ from Lemma~\ref{lem:apx-accounting}.
That lemma gives \eqref{eq:apx-finite-payoff}.

Approximate $r$ by bounded admissible step rates
\begin{equation*}
 q_k(t)=2^{-k}\cdot\left\lfloor2^k\cdot\min\left\{r(t),k\right\}\right\rfloor,
\end{equation*}
taking a right-continuous representative of each $q_k$. Each
representative changes the displayed function at only finitely
many points, so $q_k\to r$ almost everywhere. Use a subscript $k$
for quantities computed from $q_k$.
By Lemma~\ref{lem:apx-comparison} and $C(q_k)\ge1-\theta\left(0\right)$, the
difference of the finite-input expression from $L\!\left(b,q_k\right)$ is at
least $\left(1-\theta\left(0\right)-0.142\cdot B\right)\cdot\mathcal E_k\ge0$. The hypothesis on $L$ gives
\begin{equation*}
 C(q_k)\cdot A_{P,k}\left(0\right)+\int_0^1b(t)\cdot T_{n,k}(t)\,\mathrm d t
 \ge L\!\left(b,q_k\right)\ge\gamma.
\end{equation*}
It remains to take the limit only on the finite-input side.
The explicit exponential formulas give $C(q_k)\to C(r)$ and
$P_k\to P$ almost everywhere, including where the auxiliary rate $r$
is infinite. For $P_k$, the exceptional set is contained in
finitely many translates of the exceptional set for $q_k\to r$.
Dominated convergence gives $A_{P,k}\to A_P$ and
$\overline P_k\to\overline P$;\, then the piecewise formula gives
$T_{n,k}\to T_n$ almost everywhere. The bounds
$0\le P_k\le1$ and $0\le T_{n,k}\le1+1/\delta $ justify a final
application of dominated convergence:
\begin{equation*}
 C(r)\cdot A_P\left(0\right)+\int_0^1b(t)\cdot T_n(t)\,\mathrm d t
 =\lim_{k\to\infty}\left(
 C(q_k)\cdot A_{P,k}\left(0\right)+\int_0^1b(t)\cdot T_{n,k}(t)\,\mathrm d t\right)
 \ge\gamma.
\end{equation*}
Combining this with the actual-input payoff bound
\eqref{eq:apx-finite-payoff} gives
$\mathbb E\left[\mathrm{ALG}\right]\ge\gamma$ in the normalized input.
\end{proof}

\paragraph{The forty-slope threshold.}
For the threshold in Section~\ref{sec:iid-analysis}, the largest drop rate is $B=b_{40} < 1.81532$.
We verify the transfer condition by calculating $1-\theta\left(0\right)-0.142\cdot B > 0.000263 > 0$.
Using Lemma~\ref{prop:iid-transfer} and the certified bound on $L$, we obtain the displayed $0.7419$ guarantee for every $n$, proving Theorems~\ref{thm:iid-transfer} and~\ref{thm:iid}.

\section{Supplementary Prophet Secretary Analysis}
\label{app:elementary}

\subsection{Linear Lower Bounds}\label{app:secretary-linear-bounds}

In this subsection, we derive linear lower bounds for the two nonlinear contributions in $\Gamma_{p^*}(\tau)$:
$H_{p^*}\!\left(s,\Delta\right)$ from integrated conditioning and $\widehat S(s)$ from the adjustment term.
For fixed threshold and proof parameters, the remaining contributions are already linear in $p^*$, since $P^*$ is an integral of $p^*$.
We use the bounds below together with interval averaging in the next subsection to construct the finite factor-revealing LP.
Define
\begin{equation*}
    \mathcal L\!\left(\Delta,x,y\right) = \log\left(1-y\right)-\log\left(1-\Delta\cdot x-y\right)-\Delta\cdot x.
\end{equation*}
Recall that $H_{p^*}\!\left(s,\Delta\right) = \mathcal L\!\left(\Delta,p^*(s),y\right)$, where
\begin{equation*}
    \delta=\min\!\left\{s,\left(1-\Delta\right)/2\right\}, \qquad y=\int_{s-\delta}^s p^*(v)\,dv.
\end{equation*}
Since $p^*$ is nondecreasing, we have $0\le y\le\delta\cdot p^*(s)$.
For fixed $\Delta$ and $\delta$, we consider the following triangular domain:
\begin{equation*}
    D_\delta = \left\{\left(x,y\right):0\le x\le1,\ 0\le y\le\delta\cdot x\right\}.
\end{equation*}

\begin{lemma}\label{lem:log-plane}
Fix $\Delta>0$ and $\delta\ge0$ with $\Delta+\delta<1$.
For any slopes $\sigma,\nu\in\mathbb R$, we can compute a constant $\beta$, independent of $\left(x,y\right)$, such that $\mathcal L\!\left(\Delta,x,y\right)\ge\beta+\sigma\cdot x+\nu\cdot y$ for every $\left(x,y\right)\in D_\delta$.
\end{lemma}
\begin{proof}
For fixed slopes $\sigma,\nu$, write
\begin{equation*}
    R\left(x,y\right) = \mathcal{L}\left(\Delta,x,y\right)-\sigma\cdot x-\nu\cdot y.
\end{equation*}
We seek a constant $\beta$ such that $R\left(x,y\right)\ge \beta$ throughout $D_\delta$.
% s
Suppose first that $\delta>0$.
For every $\left(x,y\right)\in D_\delta$, we have $1-y\ge 1-\Delta\cdot x-y\ge 1-\Delta-\delta>0$.
Hence $R$ is continuous, and we can minimize it over the closed and bounded triangle $D_\delta$.
At every interior point, we have
\begin{equation*}
    \det\nabla^2 R\left(x,y\right) = -\frac{\Delta^2}{\left(1-y\right)^2\cdot\left(1-\Delta\cdot x-y\right)^2}<0.
\end{equation*}
We therefore exclude an interior local minimum and need only lower bound $R$ on the three edges of the triangle $D_\delta$.
For $0\le u\le 1$, we parameterize the three edges by
\begin{equation*}
    \left(x_0(u),y_0(u)\right)=\left(u,0\right),\qquad
    \left(x_1(u),y_1(u)\right)=\left(u,\delta\cdot u\right),\qquad
    \left(x_2(u),y_2(u)\right)=\left(1,\delta\cdot u\right),
\end{equation*}
and define $\varphi_e(u) = R\left(x_e(u),y_e(u)\right)$ for each $e\in\left\{0,1,2\right\}$.
Since
\begin{align*}
    \varphi_0''(u) &= \frac{\Delta^2}{\left(1-\Delta\cdot u\right)^2}\ge 0, \\
    \varphi_1''(u) &= \frac{\left(\Delta+\delta\right)^2}{\left(1-\left(\Delta+\delta\right)\cdot u\right)^2} - \frac{\delta^2}{\left(1-\delta\cdot u\right)^2}\ge 0, \\
    \varphi_2''(u) &= \frac{\delta^2}{\left(1-\Delta-\delta\cdot u\right)^2} - \frac{\delta^2}{\left(1-\delta\cdot u\right)^2}\ge 0,
\end{align*}
each $\varphi_e$ is convex.
For each edge $e$, we choose any point $u_e\in\left[0,1\right]$ and use the tangent there to obtain
\begin{equation*}
    \varphi_e(u) \ge \varphi_e(u_e) + \varphi_e'(u_e)\cdot\left(u-u_e\right), \qquad 0\le u\le1.
\end{equation*}
Minimize this linear lower bound over $\left[0,1\right]$ and let
\begin{align*}
    \beta_e =\min_{0\le u\le1} \left\{\varphi_e(u_e)+\varphi_e'(u_e)\cdot \left(u-u_e\right)\right\} =\varphi_e(u_e)-u_e\cdot \varphi_e'(u_e)+\min\left\{0,\varphi_e'(u_e)\right\}.
\end{align*}
We then have $\varphi_e(u)\ge \beta_e$ along each edge.
Taking $\beta=\min\left\{\beta_0,\beta_1,\beta_2\right\}$, we obtain $R\left(x,y\right)\ge\beta$ on the boundary and hence throughout $D_\delta$.

If $\delta=0$, we have only the segment $\left(u,0\right)$, for $0\le u\le1$.
We use the same construction for $\varphi_0$ and take $\beta=\beta_0$.
In both cases, adding back $\sigma\cdot x+\nu\cdot y$ gives the claimed lower bound.
\end{proof}

For $\Delta=0$, we use $\mathcal L\!\left(0,x,y\right)=0$ directly.
Recall that
\begin{equation*}
    \widehat S(s) = \left(1-p^*(s)\right)\cdot \exp\left(\int_0^s\log\left(1-p^*(v)\right)\,dv\right).
\end{equation*}
We use the conventions $\log0=-\infty$ and $\exp\left(-\infty\right)=0$.

\begin{lemma}\label{lem:direct-survival-plane}
Fix $0<c\le 1$ and $0\le\rho\le 1$.
For every nondecreasing $p^*:\left[0,1\right]\to\left[0,1\right]$ and every $0\le s\le c$, we have
\begin{equation*}
    \widehat S(s) \ge \rho^c\cdot\left(1+c-c\cdot\rho\right)\cdot \left(1-p^*(s)\right)-\rho^{1+c}\cdot P^*(s).
\end{equation*}
\end{lemma}
\begin{proof}
Write $\alpha=\rho^c\cdot\left(1+c-c\cdot\rho\right)$. First suppose that $0<p^*(s)<1$ and $\rho>0$.
Write
\begin{equation*}
    z=1-p^*(s),\qquad a=\frac{P^*(s)}{1-z}.
\end{equation*}
Since $p^*$ is nondecreasing, we have $0\le a\le s\le c$.
We bound the logarithm below by its chord on $\left[0,1-z\right]$:
\begin{equation*}
    \int_0^s\log\left(1-p^*(v)\right)\,dv \ge \frac{\log(z)}{1-z}\cdot \int_0^s p^*(v)\,dv = a\cdot \log(z).
\end{equation*}
Consequently, $\widehat S(s)\ge z^{1+a}$.

For $a>0$, we apply the weighted arithmetic--geometric mean
inequality with weights $1/\left(1+a\right)$ and $a/\left(1+a\right)$ to obtain
\begin{equation*}
    z^a+\frac{a\cdot \rho^{1+c}}{z} \ge \left(1+a\right)\cdot \rho^{\frac{a\cdot\left(1+c\right)}{1+a}}.
\end{equation*}
We obtain the same inequality directly when $a=0$.
Notice that $\left(c-a\right)/\left(1+a\right)\ge 0$, we have
\begin{align*}
    \left(1+a\right)\cdot \rho^{\frac{a\cdot\left(1+c\right)}{1+a}}
    &= \left(1+a\right)\cdot \rho^c\cdot \rho^{-\frac{c-a}{1+a}}\\
    &\ge \rho^c\cdot \left(1+a+\left(c-a\right)\cdot \left(1-\rho\right)\right)
    =\rho^c\cdot \left(1+c-c\cdot \rho\right) + a\cdot \rho^{1+c}.
\end{align*}
After multiplying by $z$ and rearranging, we obtain
\begin{equation*}
    z^{1+a}+a\cdot \rho^{1+c}\cdot \left(1-z\right) \ge \alpha\cdot z.
\end{equation*}
We now use $P^*(s)=a\cdot \left(1-z\right)$ and $\widehat S(s)\ge z^{1+a}$ to prove the desired bound.

If $\rho=0$, the RHS is zero.
If $p^*(s)=1$, the RHS is nonpositive.
If $p^*(s)=0$, monotonicity gives $P^*(s)=0$ and $\widehat S(s)=1$; we then use $\alpha\le 1$, since $\rho\mapsto\rho^c\cdot \left(1+c-c\cdot \rho\right)$ is nondecreasing on $\left[0,1\right]$ and equals $1$ at $\rho=1$.
These observations cover all boundary cases.
\end{proof}

\subsection{Discretization and the Factor-Revealing LP}\label{app:ind-cells}

We now convert the linear lower bounds from the preceding subsection into a finite LP.
As in the i.i.d. analysis, we construct bounds valid throughout each discretization interval and then remove the unknown probabilities $p^*$.
Here we optimize the drop rates together with the proof weights.

\paragraph{Discretization.}
We use $N$ equal-length intervals for $b$, with drop rates $b_1,\ldots,b_N$, and cover the stopping-time domain by
$T_\ell=[t_\ell^-,t_\ell^+]$, $\ell=0,\ldots,N_t-1$,
each lying within one constant-slope piece of $\theta$.
We choose a common partition for averaging $p^*$:
\begin{equation*}
    0=s_0<\cdots<s_{N_p}=1,\qquad
    I_i=[s_i,s_{i+1}],\qquad h_i=s_{i+1}-s_i,
\end{equation*}
and define
\begin{equation*}
    \xi_i=\frac{1}{h_i}\cdot \int_{I_i}p^*(s)\,ds,
    \qquad 0\le\xi_0\le\cdots\le\xi_{N_p-1}\le1.
\end{equation*}
We use averages of the original $p^*$, allowing it to vary within each interval.

We fix measurable functions $c_j(s)\in[s,1]$, $z_k(s,t)\in[0,s]$, and finitely many nonnegative bounded piecewise constant functions to represent the weights:
\begin{equation*}
    \lambda_j(s)=\sum_q L_{jq}\cdot \phi_{jq}(s),\qquad
    \mu_k(s,t)=\sum_q M_{kq}\cdot \psi_{kq}(s,t),\qquad
    \kappa(s)=\sum_q K_q\cdot \chi_q(s).
\end{equation*}
The functions are fixed; we optimize the nonnegative coefficients $L_{jq},M_{kq},K_q$.
We choose the averaging partition to include the breakpoints of $b,\phi_{jq},\chi_q$, and fix all parameters used to construct the linear lower bounds.
We assume
\begin{equation*}
    \int_0^1\int_s^1\frac{\psi_{kq}(s,t)}{1-t+z_k(s,t)}\,dt\,ds<\infty,
\end{equation*}
so that we can exchange integrals in the comparison contribution.
We partition the $(s,t)$ domain so that $b(s)$ and the functions used to represent the weights are constant on each region, and the same indices $j$ satisfy $t\le c_j(s)$.
We therefore enforce~\eqref{eq:secretary-capacity} by one linear constraint per region.

\paragraph{A uniform bound in the interval averages.}
For each $T_\ell$, we introduce $v_\ell$ and impose $v_\ell\le \min\{t_\ell^- + \theta(t_\ell^-), t_\ell^+ + \theta(t_\ell^+)\}$.
Since $\tau+\theta(\tau)$ is linear on $T_\ell$, we may use $v_\ell$ as a lower bound on $\tau+\theta(\tau)$ throughout this interval.

\begin{lemma}\label{lem:secretary-discretization}
For each $\ell$, we can construct $C_\ell$ and $g_{\ell,i}$,
$0\le i<N_p$, as linear expressions in the drop rates $b_1,\ldots,b_N$, the proof-weight coefficients $L_{jq},M_{kq},K_q$, and $v_\ell$, such that
\begin{equation}\label{eq:secretary-cell-affine}
    \Gamma_{p^*}(\tau)\ge C_\ell+\sum_{i=0}^{N_p-1}g_{\ell,i}\xi_i
    \qquad\text{for every }\tau\in T_\ell
\end{equation}
for every nondecreasing $p^*:[0,1]\to[0,1]$.
\end{lemma}
\begin{proof}
After applying Lemmas~\ref{lem:log-plane} and~\ref{lem:direct-survival-plane} and exchanging the order of integration, we lower-bound the resulting integrals in terms of the interval averages $\xi_i$, uniformly for $\tau\in T_\ell$.

For example, for an integrable coefficient function $k$ and fixed real constants $d_i$, we use $0\le p^*\le 1$ to obtain
\begin{equation*}
    \int_0^1 k(s)\cdot p^*(s)\,ds \ge \sum_i h_i\cdot d_i\cdot \xi_i - \sum_i\int_{I_i}(d_i-k(s))^+\,ds.
\end{equation*}
We can therefore replace these functionals by linear expressions in the actual interval averages.
We construct these bounds for each fixed function before multiplying by its nonnegative drop-rate or weight coefficient.
Finally, we use $\tau+\theta(\tau)\ge v_\ell$ and collect the constant and average coefficients as $C_\ell,g_{\ell,i}$.
All coefficient bounds are fixed before optimization, so these expressions are linear in $b_1,\ldots,b_N$, $L_{jq}, M_{kq}, K_q$, and $v_\ell$.
\end{proof}

\paragraph{Removing the unknown probabilities.}
We minimize the resulting expression over all ordered averages $0\le\xi_0\le\cdots\le\xi_{N_p-1}\le 1$.

\begin{lemma}\label{lem:ordered-cube}
For any real $C,g_0,\ldots,g_{N_p-1}$, we have
\begin{equation*}
    \min_{0\le\xi_0\le\cdots\le\xi_{N_p-1}\le1}
    \left\{C+\sum_{i=0}^{N_p-1}g_i\cdot\xi_i\right\}
    =C+\min_{0\le h\le N_p}\left\{\sum_{i=h}^{N_p-1}g_i\right\},
\end{equation*}
where the sum for $h=N_p$ is zero.
\end{lemma}

\paragraph{The factor-revealing LP.}
We optimize the candidate guarantee $\gamma$, the drop rates $b_1,\ldots,b_N$,
the coefficients $L_{jq},M_{kq},K_q$, and $v_0,\ldots,v_{N_t-1}$:
\begin{align*}
    \text{maximize}\quad&\gamma\\
    \text{subject to}\quad
    &0\le b_1\le\cdots\le b_N,\qquad
      \frac1N\sum_{a=1}^N b_a\le1,\\
    &L_{jq},M_{kq},K_q\ge0,\\
    &\sum_{j:c_j(s)\ge t}\lambda_j(s)+\sum_k\mu_k(s,t)\le b(s)
      &&\text{on each allocation region},\\
    &v_\ell\le t_\ell^-+\theta(t_\ell^-),\quad
      v_\ell\le t_\ell^++\theta(t_\ell^+)
      &&0\le\ell<N_t,\\
    &C_\ell+\sum_{i=h}^{N_p-1}g_{\ell,i}\ge\gamma
      &&0\le\ell<N_t,\quad0\le h\le N_p.
\end{align*}
We substitute the linear expressions from Lemma~\ref{lem:secretary-discretization} into the last line.
For every feasible solution, we apply Lemma~\ref{lem:ordered-cube} on each $T_\ell$ to obtain $\Gamma_{p^*}(\tau)\ge\gamma$ for every admissible $p^*$ and every $\tau\in[0,1]$.
By Lemma~\ref{lem:secretary-criterion}, we conclude that $\E[\ALG]\ge\E[\Gamma_{p^*}(\tau)]\ge\gamma$.

\paragraph{Numerical certificate.}
We choose $N=20$ threshold intervals, $N_t=1283$ stopping-time intervals, and a common partition of $N_p=5120$ averaging intervals.
A feasible solution satisfies
\begin{equation*}
    \gamma_{\mathrm{cert}}
    =\frac{49735269883477169}{72057594037927936}
    =0.69021552200728099\ldots,\qquad
    \theta(0)=\frac{8998481631}{12500000000}.
\end{equation*}
Verification consists of checking the fixed coefficient bounds and all constraints of the finite LP.

\section{Upper Bounds}

In this section, we give three upper bounds.
For the i.i.d. prophet inequality, we prove an upper bound of $\beta^*-\varepsilon$ for all online algorithms for some constant $\varepsilon>0$, separating the mean-only and the full-information settings. 
We also prove an upper bound of $0.74197$ for (possibly randomized) precommitted thresholds.
For the prophet secretary, we prove an upper bound of $0.6945$ for deterministic precommitted thresholds.

To construct the hard instances, we begin with bounded inputs whose expected maximum is less than $1$.
We then introduce rare large values to make the expected maximum equal to $1$.
In the following lemma, we describe this modification and calculate how it changes the algorithm's expected reward.

\begin{lemma}
\label{lem:rare-completion}
Consider $n$ nonnegative independent variables bounded by a common constant, and let their expected maximum be $M<1$.
Designate $1 \le r \le n$ variables.
Independently, replace each designated variable
by $H_\varepsilon = \frac{1-\left(1-\varepsilon\right)^r\cdot M}{1-\left(1-\varepsilon\right)^r}$ with probability $\varepsilon$, and leave it unchanged with probability
$1-\varepsilon$.
For sufficiently small $\varepsilon$, the completed input has expected maximum $1$.

On the original input, let $Y$ be the reward of a policy, and let $K$ be the number of designated variables observed up to and including acceptance. On failure, put $Y=0$ and $K=r$.
Extend the policy by accepting $H_\varepsilon$ immediately. Then its expected reward is
\begin{equation*}
 \E\!\left[Y\right]
   +\left(1-M\right)\cdot\frac{\E\!\left[K\right]}r+o\!\left(1\right),
 \qquad \varepsilon\to 0^+.
\end{equation*}
The error is uniform over all policies on the fixed original input.
If the original variables are i.i.d. and all are designated, the completed variables remain i.i.d.
\end{lemma}
\begin{proof}
For small $\varepsilon$, $H_\varepsilon$ exceeds every original value.
The maximum is the original maximum if no replacement occurs, and $H_\varepsilon$ otherwise. Therefore,
\begin{equation*}
 \E\!\left[X^*\right]
 =\left(1-\varepsilon\right)^r\cdot M
   +\left(1-\left(1-\varepsilon\right)^r\right)\cdot H_\varepsilon=1.
\end{equation*}

Couple the two runs using the same original values, arrival order, and private seed.
Conditional on the original run, the completed policy earns $Y$ if the first $K$ designated variables are unchanged, and $H_\varepsilon$ if one of them has been replaced.
Thus its exact expected reward equals:
\begin{equation*}
 \E\!\left[Y\cdot\left(1-\varepsilon\right)^K
       +H_\varepsilon\cdot
             \left(1-\left(1-\varepsilon\right)^K\right)\right].
\end{equation*}

Observe that $\left(1-\left(1-\varepsilon\right)^k\right)/\varepsilon\to k$ uniformly for $0\le k\le r$
and $\varepsilon \cdot H_\varepsilon \to \left(1-M\right)/r$ as $\varepsilon \to 0^+$.
Since $Y$ is bounded uniformly over policies, taking the limit yields
the claim and its uniformity.
The i.i.d. assertion is immediate.
\end{proof}

%On an independent input with expected maximum $1$, accepting a value greater than $1$ is weakly better than continuing: conditional on the observed history and remaining labels, the expected maximum of the remaining independent variables is at most $1$. Consequently, requiring immediate acceptance of the rare value does not restrict the upper bounds.

\subsection{A strict information loss for i.i.d. observations}
\label{app:iid-all-algorithms}

Recall that $\beta^*$ is the optimal ratio for the full-information i.i.d. prophet inequality.
We use its standard Hill--Kertz characterization:
\begin{equation*}
 \int_0^1\frac{dy}{D+y-y\cdot\log y}=1,
 \qquad \beta^*=\frac1{1+D}\approx 0.7454,
\end{equation*}
where $D>0$ and $0\cdot\log0=0$.

\begin{theorem}\label{thm:adaptive-gap}
There are constants $\varepsilon>0$ and $n_0$ such that, for every $n\ge n_0$, there are two nonnegative distributions $F_{0,n}$ and $F_{1,n}$, each with expected maximum $1$, for which every randomized, history-dependent algorithm $A$ satisfies:
\begin{equation*}
 \min\left\{\E_{F_{j,n}}\left[\ALG_A\right]:j\in\left\{0,1\right\}\right\}
 \le \beta^*-\varepsilon.
\end{equation*}
The pair does not depend on $A$.
The conclusion still holds if the algorithm is told both candidate distributions, but not which candidate generates the observations.
\end{theorem}

We first construct two auxiliary inputs in continuous time.
For the first input, we identify an optimal threshold rule attaining $\beta^*$.
We then show that any algorithm whose reward is sufficiently close to $\beta^*$ on the first input must fall a fixed amount below $\beta^*$ on the second.
Finally, we turn the two inputs into finite i.i.d. inputs and apply Lemma~\ref{lem:rare-completion} to introduce rare large values so that $\E[X^*]=1$.

\paragraph{A base instance and its regret identity.}
For a Poisson process with intensity $dt\cdot\nu\left(dx\right)$ on $\left[0,1\right]\times\left(0,\bet\right]$, assume that $\int x\cdot\nu\left(dx\right)<\infty$.
Write $M_\nu$ for its largest mark, putting $M_\nu=0$ if there are no marks, and let $h_\nu=1-\E_\nu\left[M_\nu\right]$.
For a policy selecting value $X_A$ at time $T_A$, put $\left(X_A,T_A\right)=\left(0,1\right)$ on failure and define the accounting reward
\begin{equation*}
 R_\nu(A)=\E_\nu\left[X_A+h_\nu\cdot T_A\right].
\end{equation*}
The extra term is the limiting contribution of the rare large values in Lemma~\ref{lem:rare-completion}; we apply that lemma after constructing finite inputs.

Define
\begin{equation*}
 d(y)=D+y-y\cdot\log y,\qquad
 C=\frac{\bet}{\int_0^1d(y)^{-2}\,dy},\qquad
 x(y)=C\cdot\int_0^y d(z)^{-2}\,dz.
\end{equation*}
Thus $x(0)=0$, $x\left(1\right)=\bet$, and $x'(y)=C/d(y)^2>0$.
Define the base intensity $\nu_0$ by sending the measure $dy/y$ through $x$: for every Borel set $S$ of positive mark values,
\begin{equation*}
 \nu_0(S)=\int_0^1\ind\left(x(y)\in S\right)\,\frac{dy}{y}.
\end{equation*}
Its first moment is finite because $x(y)\le C\cdot y/D^2$.
Moreover,
\begin{equation*}
 \PP_0\left(M_0\le x(y)\right)
 =\exp\left(-\int_y^1\frac{dz}{z}\right)=y,
 \qquad 0<y\le1.
\end{equation*}
Consequently, integration by parts, $D\cdot\bet=1-\bet$, and $\left(y/d(y)\right)'=\left(D+y\right)/d(y)^2$ give
\begin{align*}
 h_0
 &=1-\int_0^1x(y)\,dy
 =1-\bet+C\cdot\int_0^1\frac{y}{d(y)^2}\,dy\\
 &=C\cdot\int_0^1\frac{D+y}{d(y)^2}\,dy
 =C\cdot\left[\frac{y}{d(y)}\right]_0^1=C\cdot\bet.
\end{align*}

Define the continuous, decreasing threshold $\theta_0$ by
\begin{equation*}
 t(y)=\int_y^1\frac{dz}{d(z)},\qquad
 \theta_0\left(t(y)\right)=x(y).
\end{equation*}
Its endpoints are $\theta_0(0)=\bet$ and $\theta_0\left(1\right)=0$, and $-\theta_0'\left(t(y)\right)=C/d(y)$.
To check its Bellman equation, differentiate
$h_0+\int_y^1\left(x(z)-x(y)\right)\,dz/z$
with respect to $y$.
Its derivative is $C\cdot\log y/d(y)^2$, the derivative of $C/d(y)$, and both expressions equal $h_0$ at $y=1$.
Therefore,
\begin{equation*}
 -\theta_0'(t)
 =h_0+\int\left(x-\theta_0(t)\right)_+\cdot\nu_0\left(dx\right).
\end{equation*}

Integrate this identity up to $T_A$ and take expectations.
The Poisson compensation formula replaces the integral of the eligible excess by the sum of these excesses over marks observed while the policy is active: the indicator of being active immediately before an arrival is predictable.
This gives
\begin{equation*}
 \bet=\E_0\left[
 \theta_0\left(T_A\right)+h_0\cdot T_A
 +\sum_{\text{observed while active }\left(t,x\right)}
       \left(x-\theta_0(t)\right)_+
 \right].
\end{equation*}
All terms are integrable since the total mark value has finite expectation.
Subtracting $R_0(A)$, and combining the selected mark's excess with its terminal threshold, yields the regret identity
\begin{equation}\label{eq:regret}
 \bet-R_0(A)=\E_0\left[
 \sum_{\text{rejected while active }\left(t,x\right)}
       \left(x-\theta_0(t)\right)_+
 +\left(\theta_0\left(T_A\right)-X_A\right)_+
 \right].
\end{equation}
The policy $A_0$ accepting the first mark strictly above $\theta_0$ makes both terms zero.
Hence $R_0(A)\le\bet$ for every policy, with equality for $A_0$.

\paragraph{A second instance and a common gap.}
Let $B=\left[x\left(1/2\right),x\left(3/4\right)\right]$ and increase its intensity by a factor $1+m$:
\begin{equation*}
 \nu_1\left(dx\right)=\left(1+m\cdot\ind\left(x\in B\right)\right)\cdot\nu_0\left(dx\right).
\end{equation*}
We choose a finite positive integer $m$ below.
Under a coupling, instance $1$ is instance $0$ together with an independent process of intensity $dt\cdot m\cdot\nu_0|_B\left(dx\right)$.
Put $c=x\left(3/4\right)$ and choose $t_c$ so that $\theta_0\left(t_c\right)=c$.
All added marks are ineligible before $t_c$.
For every $0<\eta<1-t_c$, the probability that no eligible added mark arrives in $\left(t_c,t_c+\eta\right)$ is
\begin{equation*}
 \exp\left(
 -m\cdot\int_{t_c}^{t_c+\eta}
 \nu_0\left(B\cap\left(\theta_0(t),\infty\right)\right)\,dt
 \right)\longrightarrow0.
\end{equation*}
Thus, conditional on still being active at $t_c$, the rule $A_0$ stops at a time tending to $t_c$, with a value tending to $c$.
Indeed, original eligible marks have finite intensity near $t_c$, and an added value accepted before $t_c+\eta$ lies between $\theta_0\left(t_c+\eta\right)$ and $c$.

Let $\left(X^c,T^c\right)$ follow $A_0$ up to $t_c$ and select $c$ at $t_c$ if still active.
The same compensated identity as above gives
$\E_0\left[X^c+h_0\cdot T^c\right]=\bet$:
every rejected mark has zero excess, and the forced terminal value equals the threshold.
Also, the largest added mark tends to $c$, so
\begin{equation*}
 h_1\longrightarrow h_0-\E_0\left[\left(c-M_0\right)_+\right].
\end{equation*}
Since values and times are bounded, these convergences also hold in expectation. Consequently,
\begin{equation*}
 \lim_{m\to\infty}R_1\left(A_0\right)
 =\bet-\E_0\left[\left(c-M_0\right)_+\right]
             \cdot\E_0\left[T^c\right]<\bet.
\end{equation*}
The inequality is strict because $\PP_0\left(M_0<c\right)=3/4$ and $T^c>0$ almost surely.
Fix a finite $m$ for which $R_1\left(A_0\right)<\bet$.

We now prove that there is a constant $\delta>0$ such that every common policy satisfies
\begin{equation}\label{eq:poissongap}
 \min\left\{R_0(A),R_1(A)\right\}\le\bet-\delta.
\end{equation}
Otherwise, choose policies $A_k$ with both rewards at least $\bet-2^{-k}$, and run them on the same arrivals and a common private seed independent of the arrivals.
Let $G_k$ be the nonnegative expression inside the expectation in~\eqref{eq:regret}.
Then
\begin{equation*}
 \E_0\left[\sum_{k\ge1}G_k\right]
 \le\sum_{k\ge1}2^{-k}<\infty,
\end{equation*}
so $G_k\to0$ almost surely.

Fix such a sample path and $s<1$.
Every mark before $s$ is separated from its threshold by a common positive amount: the marks above $\theta_0(s)/2$ form a finite set with no threshold ties, and every other mark is at least $\theta_0(s)/2$ below its threshold.
A first disagreement with $A_0$ before $s$ either rejects a mark above the threshold or selects a mark below it.
Its positive distance appears in $G_k$, so eventually no such disagreement is possible.
If both policies are still active at $s$, their stopping times differ by at most $1-s$, and the difference between their selected values is at most the total mark value after $s$.
This total tends to zero as $s\uparrow1$, since the total mark value is finite almost surely and no mark arrives exactly at time $1$.
Applying this argument to a countable sequence $s\uparrow1$ gives
\begin{equation*}
 \left(X_{A_k},T_{A_k}\right)
 \longrightarrow\left(X_{A_0},T_{A_0}\right)
 \qquad \PP_0\text{-almost surely}.
\end{equation*}

If $N_B$ is the number of marks in $B$, the likelihood ratio is
\begin{equation*}
 \frac{d\PP_1}{d\PP_0}
 =\left(3/2\right)^{-m}\cdot\left(1+m\right)^{N_B}.
\end{equation*}
The private seed has the same independent law in both instances, so this formula also applies to the joint observation and seed.
Absolute continuity transfers the almost-sure convergence to $\PP_1$.
Since $0\le X_A+h_1\cdot T_A\le\bet+1$, bounded convergence gives
$R_1\left(A_k\right)\to R_1\left(A_0\right)<\bet$,
contradicting the choice of $A_k$.
This proves~\eqref{eq:poissongap}.

\paragraph{Transfer to finite i.i.d. observations.}
Choose $a>0$ and retain only marks of value at least $a$.
Write $\nu_j^a$ for their intensity, put $q_j=\nu_j\left(\left[a,\bet\right]\right)$, and let $Q=\max\!\left\{q_0,q_1\right\}<\infty$.
Divide $\left[0,1\right]$ into $n$ equal cells.
Let $G_{j,n}$ be the distribution of the first retained mark in one cell, with value $0$ if the cell is empty.
The cell values are independent and identically distributed.
Let $M_{j,n}$ be their maximum and $h_{j,n}=1-\E\!\left[M_{j,n}\right]$.

For any decision rule $A$ on $n$ observations that never selects $0$ before the end, simulate it in the Poisson model by feeding it the first retained mark in each cell when that mark arrives, or $0$ at the cell endpoint if the cell is empty.
Ignore later marks in that cell and generate any independent label or arrival-time information in the private seed.
This is the same policy $\widetilde A$ under both candidates.
If $A$ selects value $X_A$ at index $K_A$, with $\left(X_A,K_A\right)=\left(0,n\right)$ on failure, the simulation satisfies exactly
\begin{equation*}
 X_{\widetilde A}=X_A,
 \qquad
 \left|T_{\widetilde A}-K_A/n\right|\le1/n.
\end{equation*}

Unless a cell contains two retained marks, keeping only its first mark preserves the retained maximum.
A Poisson variable of mean $q_j/n$ is at least $2$ with probability at most $q_j^2/\left(2\cdot n^2\right)$.
Thus the probability of any such cell is at most $Q^2/\left(2\cdot n\right)$.
Removing marks below $a$ changes the maximum by at most $a$, and all other marks are at most $\bet$; hence
\begin{equation*}
 h_{j,n}\le h_j+a+\frac{\bet\cdot Q^2}{2\cdot n}.
\end{equation*}
Since $0<h_j\le1$, it follows that
\begin{equation*}
 \E_{G_{j,n}}\left[X_A+h_{j,n}\cdot K_A/n\right]
 \le R_j\left(\widetilde A\right)
      +a+\frac{1+\bet\cdot Q^2/2}{n}.
\end{equation*}

Apply Lemma~\ref{lem:rare-completion} to $n$ independent observations from each $G_{j,n}$, using replacement probability $\zeta$.
The resulting distributions $F_{j,n,\zeta}$ have expected maximum $1$.
We may assume that a policy accepts a rare value immediately and never selects $0$ before the end: a rare value exceeds $1$, whereas the conditional expected maximum of the remaining observations is at most $1$ under either candidate.
On histories with no rare value, its decisions form a common rule $A$ on the cell observations.
Uniformly over these rules, the lemma gives
\begin{equation*}
 \E_{F_{j,n,\zeta}}\left[\ALG_A\right]
 =\E_{G_{j,n}}\left[X_A+h_{j,n}\cdot K_A/n\right]+o\!\left(1\right),
 \qquad \text{when } \zeta\to 0^+.
\end{equation*}
Combining this with~\eqref{eq:poissongap}, choose $a<\delta/4$, then $n_0$ so that the error proportional to $1/n$ is at most $\delta/4$ for all $n\ge n_0$.
For each such $n$, choose $\zeta_n$ so that the completion error is at most $\delta/4$ for both candidates and all policies.
The pair $F_{j,n}=F_{j,n,\zeta_n}$ is independent of the algorithm and satisfies
\begin{equation*}
 \min\left\{\E_{F_{j,n}}\left[\ALG_A\right]:j\in\left\{0,1\right\}\right\}
 \le\bet-\delta/4.
\end{equation*}
This proves Theorem~\ref{thm:adaptive-gap} with $\varepsilon=\delta/4$.

\subsection{Precommitted thresholds in the i.i.d. setting}\label{app:hardness}\label{sec:iid-dual}

We prove an upper bound for the class of continuous and strictly decreasing thresholds with $\theta\left(1\right) = 0$, including
thresholds chosen at random before the observations.
First, we reduce random precommitment to one deterministic threshold.
Then, a finite dual witness gives a bad row for that threshold.
Finally, we realize the row by a normalized  distribution.

\paragraph{Bounding the initial height.}
The following reduction also applies to the prophet secretary upper bound in the next subsection.

\begin{lemma}[Reducing the initial height]\label{lem:height-reduction}
Suppose that $\theta(t)$ is continuous and strictly decreasing with $\theta\left(1\right)=0$. The input is normalized such that $\E\!\left[X^*\right]=1$.
Define $\hat{\theta}(t) = \min\left\{\theta(t),1\right\}$
and
$ \theta_\delta(t)
 =\min\!\left\{\theta(t),1-\delta \cdot t\right\}$ for every $0<\delta<1$.
Then the following results hold.

\text{(i)} $\hat{\theta}(t)$ satisfies:
\begin{equation*}
    \E\!\left[\ALG_{\theta}\right] \le \E\!\left[\ALG_{\hat{\theta}}\right].
   \end{equation*}

\text{(ii)}
$\theta_\delta(t)$
is continuous and strictly decreasing, has initial height at most $1$, and ends at $0$.

\text{(iii)}
$\theta_\delta(t)$ satisfies:
\begin{equation*}
 \E\!\left[\ALG_\theta\right]
 \le\E\!\left[\ALG_{\theta_\delta}\right]+\delta.
\end{equation*}
\end{lemma}

\begin{proof}
Couple the algorithms that use different threshold functions on the same values and arrival times.
When their first disagreement occurs, it must be that the lower-threshold algorithm accepts the current value and stops, while the original algorithm rejects and continues.

Conditional on the observed history, the remaining values retain their independent laws, so the expected future reward of any continuation policy is at most:
\begin{equation*}
 \E\!\left[\max_{i\ \mathrm{remaining}}\left\{X_i\right\}\right]
 \le\E\!\left[X^*\right]=1.
\end{equation*}

Suppose that the algorithms of $\theta$ and the modified threshold disagree at $\left(x,t_x\right)$.
To see (i),
 it holds that $\theta(t_x) \ge x>1$, so this change cannot reduce the expected reward.
For (iii), it must be that $x > 1-\delta \cdot t_x \ge 1-\delta$, so the conditional loss is at most $\delta$.
 Averaging over the history gives the result.

To see (ii), the minimum of two continuous, strictly decreasing functions is still continuous and strictly decreasing.
We have $\theta_\delta\left(1\right)=0$ by definition.

Finally, as for a random threshold function, condition first on its independently chosen seed and then average.
\end{proof}

\paragraph{Removing random precommitment.}
The adversarial distribution must be fixed before the
algorithm’s random seed is revealed. Choosing a different hard distribution for each realized function would not prove an upper bound for randomized algorithms. The next statement removes this issue.
\begin{lemma}[A deterministic threshold function dominates random precommitment]
\label{lem:random-functions}
Every independent random mixture of continuous, strictly decreasing
thresholds of initial height at most $1$ is dominated by one deterministic
threshold of the same class against every i.i.d. input whose common value distribution is finitely supported.
\end{lemma}
\begin{proof}
%We will express the reward as a concave function of the eligibility times, so averaging these times yields the desired deterministic domination.

For a threshold curve $\theta$, define its eligibility by the following continuous function:
\begin{equation*}
 \ell_\theta(x)=\operatorname{Leb}\left\{t\in\left[0,1\right]:\theta(t)\ge x\right\},
 \qquad x\ge0.
\end{equation*}
Then $\ell_\theta(0)=1$. $\ell_\theta$ strictly decreases up to $\theta(0)$,
and becomes zero thereafter.
By definition, for a policy still active, a positive value $x$ gets accepted when its arrival time exceeds $\ell_\theta(x)$.

Let $\Theta$ be a random threshold.
Define:
\begin{equation*}
    \bar\ell(x)=\E_\Theta\left[\ell_\Theta(x)\right], \qquad
    h=\operatorname{ess\,sup}\Theta(0)\le1.
\end{equation*}

Bounded convergence gives $\bar\ell$'s continuity, with $\bar\ell(0)=1$ and $\bar\ell(h)=0$.
$\bar{\ell}$ is also strictly decreasing on $\left[0,h\right]$: whenever $0\le x<y<h$, the event
$\left\{\Theta(0)>y\right\}$ has positive probability and
$\ell_\Theta(x)>\ell_\Theta(y)$ on that event. The endpoint $y=h$
follows from $\bar\ell(x)>0$ for $x<h$.
Therefore, inverting $\bar\ell$ on $\left[0,h\right]$ gives a continuous and strictly decreasing threshold $\bar\theta(t)$ with $\bar\theta(0)=h$ and $\bar\theta\left(1\right)=0$.

Fix a distribution with positive support
$v_1>\cdots>v_m>0$
and corresponding positive probabilities
$\{p_i\}_{i \in [m]}$.
The additional mass is put at zero.
For a fixed threshold curve $\theta$, let
$\ell_i=\ell_\theta(v_i)$ for $i \in [m]$ and $\ell_{m+1}=1$.

Define the following terms:
\begin{align*}
 P_j&=\sum_{i\le j}p_i,
 \qquad M_j=\frac{\sum_{i\le j}p_i\cdot v_i}{P_j}, \\
 D_j&=\left(1-\sum_{i<j}p_i\cdot\left(\ell_j-\ell_i\right)\right)^n~ \text{for}~1\le j \le m+1.
\end{align*}
$D_j$ is the probability that no eligible value has arrived
by time $\ell_j$.
During $\left(\ell_j,\ell_{j+1}\right)$, only $v_1,\ldots,v_j$ are eligible, and their conditional mean is $M_j$.
Then,
\begin{align*}
 \E\!\left[\ALG_\theta\right]
 =\sum_{j=1}^m M_j\cdot \left(D_j-D_{j+1}\right)
 =M_1-\sum_{j=2}^m\left(M_{j-1}-M_j\right)\cdot D_j-M_m\cdot D_{m+1}.
\end{align*}
Here we use $D_1=1$, and the second equality is a telescoping sum.

We now apply this identity to the random threshold $\Theta$.
By definition, each $D_j$ is the $n$-th power of a nonnegative expression that is a constant plus a linear combination of the eligibility times, and $z\mapsto z^n$ is convex on $\left[0,1\right]$. We have:
\begin{align*}
    \E_{\Theta}[D_j] 
    &= \E_\Theta \left[\left(1-\sum_{i<j}p_i\cdot\left(\ell_j-\ell_i\right)\right)^n\right] \\
    &\ge \left(\E_\Theta \left[1-\sum_{i<j}p_i\cdot\left(\ell_j-\ell_i\right)\right]\right)^n \\
    &= \left(  1-\sum_{i<j}p_i\cdot\left(\bar{\ell}_j-\bar{\ell}_i\right) \right)^n = D_j\left(\bar{\theta}\right),
\end{align*}
where the second line uses Jensen's inequality.
Since $M_{j-1} \ge M_j$, taking expectations over the randomness of $\Theta$ and substituting the above inequalities, we get $\E_\Theta\left[\E\!\left[\ALG_\Theta \,| \,\Theta\right]\right]
 \le \E\!\left[\ALG_{\bar\theta}\right]$.
The same $\bar\theta$ works for every finitely supported distribution.
\end{proof}

\paragraph{A finite dual witness.}
Recall that in the max–min program, we fix a drop rate and control
every increasing rate. For the upper bound, we reverse the roles: fix finitely many increasing
rates and average their payoffs. Since the reward is a constant plus a linear expression in the threshold decrease,
we control every threshold through a coefficientwise bound on the average. We then realize a
suitable function by an actual normalized distribution.

Partition $\left[0,1\right]$ into $N$ equal cells.
Define a \textit{row} to be a nonnegative, nondecreasing $N$-dimensional  vector $\mathbf{q}=\left(q_0,\ldots,q_{N-1}\right)$.
Identify $\mathbf{q}$ with the following step function:
\begin{align*}
    q(t)=q_j
~\text{for }t\in\left[j/N,\left(j+1\right)/N\right), ~ 0\le j<N; \qquad q\left(1\right)=q_{N-1}.
\end{align*}

Recall the following functions:
\begin{align*}
 S_q(t)&=\exp\!\left(-\int_0^t q(s)\,ds\right),\\
 A_q(t)&=\int_t^1 S_q(s)\,ds,\\
 k_q(t)&=S_q(t)-q(t)\cdot A_q(t)-A_q(0)\cdot e^{-q(t)}.
\end{align*}

Write $S_j=S_q\left(j/N\right)$ and $a_j=\int_{j/N}^{\left(j+1\right)/N}S_q(t)\,dt$.
Then,
\begin{equation*}
 S_0=1,\qquad S_{j+1}=S_j\cdot e^{-q_j/N},\qquad
 a_j=\begin{cases}\left(S_j-S_{j+1}\right)/q_j,&q_j>0,\\S_j/N,&q_j=0.
 \end{cases}
\end{equation*}

On cell $j$, $\left(S_q-q_j\cdot A_q\right)'=-q_j\cdot S_q+q_j\cdot S_q=0$. Thus $k_q$ is
constant there, with value
\begin{equation}\label{eq:row}
 A_q(0)=\sum_j a_j,
 \qquad
 k_{q,j}=S_{j+1}-q_j\cdot \sum_{\ell>j}a_\ell-A_q(0)\cdot e^{-q_j}.
\end{equation}

For $z_j=\theta\left(j/N\right)$ and $w_j=z_j-z_{j+1}$, define the row payoff:
\begin{equation*}
 F_q(\theta)=\left(1-z_0\right)\cdot A_q(0)+\sum_jw_j\cdot \left(1-k_{q,j}\right).
\end{equation*}

We note that this row payoff is an auxiliary functional.
Its role will be justified in Lemma~\ref{lem:finite-hard}, which realizes a selected
row by a normalized input and bounds the actual
expected reward in terms of $F_q(\theta)$.

Since $z_N=0$, $\theta$ is decreasing and $\theta(0)\le 1$, the coefficients $1-z_0,w_0,\ldots,w_{N-1}$ are nonnegative and sum
to one.

\begin{lemma}[The finite dual test]\label{prop:dual}
Suppose that there are finitely many rows $\mathbf{q}$ with weights $\lambda_q\ge0$ summing to $1$,
and
\begin{equation*}
 \sum_q\lambda_q \cdot A_q(0)\le U,
 \qquad
 \sum_q\lambda_q\cdot\left(1-k_{q,j}\right)\le U\quad\left(0\le j<N\right).
\end{equation*}
Then every threshold of initial height at most $1$ has
$F_q(\theta)\le U$ for at least one row.
\end{lemma}
\begin{proof}
The weighted average of the row payoffs is at most:
\begin{equation*}
 \left(1-z_0\right)\cdot U+\sum_jw_j\cdot U
 =U\cdot \left(1-z_0+\sum_jw_j\right)=U.
\end{equation*}
There must exist some row at most this weighted average $U$.
\end{proof}

Our certificate uses $N = 1024$ and $1014$ rows with rational weights. The displayed inequalities with $U = 0.741967$ can be verified by solving the finite dual LP and evaluating its coefficients with outward rounding.

The dual test gives a row payoff at most $U$.
To turn it into an upper bound on an algorithm, we need to realize the row by base values at the
threshold’s cell boundaries and one rare large value. The next construction gives the finite-$n$
upper bound directly.

\begin{lemma}
\label{lem:finite-hard}
Fix a finite witness satisfying Lemma~\ref{prop:dual}, with row set
$\mathcal Q$ and bound $U$. Let $n\ge 2$ be an integer satisfying
$n \ge 1+\max_{q\in\mathcal Q}\left\{\max_{0\le j<N}\left\{q_j\right\}\right\}.$
For every continuous, strictly decreasing threshold curve
$\theta$ with $\theta(0)\le1$ and $\theta\left(1\right)=0$, there exists a
distribution $F$ supported on at most $N+2$ nonnegative values
such that, for i.i.d. variables $X_1,\ldots,X_n \sim F$,
\begin{equation*}
    \mathbb E\!\left[X^*\right]=1,
\qquad
\mathbb E\left[\mathrm{ALG}_\theta\right]
\le e^{1/n}\cdot \left(U+\frac{2}{n}\right).
\end{equation*}
The threshold curve $\theta$ may depend on $n$.
\end{lemma}
\begin{proof}
Choose a row $\mathbf{q}$ with $F_q(\theta)\le U$.
Put $q_{-1}=0$.
Give each base variable mass $\left(q_j-q_{j-1}\right)/n$ at $z_j$,
$0\le j<N$, and the remaining mass at $0$. These masses are
nonnegative and sum to $q_{N-1}/n<1$. On cell $j$, the eligible-value
probability is $q_j/n$.
Equality times happen with probability zero.

Let $M$ be the base expected maximum and let $d=1-M$. Integrating its
tail over the value intervals $\left(z_{j+1},z_j\right)$ gives:
\begin{align*}
     M&=\sum_jw_j\cdot \left[1-\left(1-q_j/n\right)^n\right]<1, \\
 d&=1-z_0+\sum_jw_j\cdot \left(1-q_j/n\right)^n
 \le1-z_0+\sum_jw_j\cdot e^{-q_j}.
\end{align*}

Let $Y$ be the base selected value and $K$ the number of arrivals
observed up to and including acceptance, with $K=n$ on failure.
Recall that by Lemma~\ref{lem:rare-completion}, the reward after rare-value completion converges to $\mathbb E[Y]+d\cdot\mathbb E[K]/n$.
We bound this quantity on the base input.
Define:
\begin{equation*}
 y(t)=\int_0^tq(s)\,ds,
 \qquad C_j=\sum_{i\le j}\left(q_i-q_{i-1}\right)\cdot z_i.
\end{equation*}

Conditional on one variable arriving at $t$, each other variable has
produced an eligible arrival before $t$ with probability $y(t)/n$.
Thus the current variable is observed with probability
$\left(1-y(t)/n\right)^{n-1}$. Summing the selected-value contributions over the
$n$ variables gives
\begin{equation*}
 \E[Y]=\sum_j\int_{j/N}^{\left(j+1\right)/N}
       \left(1-y(t)/n\right)^{n-1}\cdot C_j\,dt.
\end{equation*}

The same conditioning also gives
\begin{equation*}
 \frac{\E[K]}n=\int_0^1\left(1-y(t)/n\right)^{n-1}\,dt.
\end{equation*}

Combining the two identities yields
\begin{equation*}
 \E[Y]+d\cdot\frac{\E[K]}n
 =\sum_j\int_{j/N}^{\left(j+1\right)/N}\left(1-y(t)/n\right)^{n-1}\cdot\left(C_j+d\right)\,dt.
\end{equation*}

Recall that $e^{-y(t)}=S_q(t)$, and hence
$\int_{j/N}^{\left(j+1\right)/N}e^{-y(t)}\,dt=a_j$ and
$\int_0^1e^{-y(t)}\,dt=A_q(0)$.
For $0\le y<n$, we have 
$ \left(1-y/n\right)^{n-1}\le e^{1/n-y}.$
Applying this bound and the bound on $d$ gives
\begin{equation*}
 \E[Y]+d\cdot\frac{\E[K]}n
 \le e^{1/n}\cdot \left\{\sum_ja_j\cdot C_j
       +A_q(0)\cdot \left(1-z_0+\sum_jw_j\cdot e^{-q_j}\right)\right\}.
\end{equation*}

To identify the expression in braces, use $z_i=\sum_{k\ge i}w_k$:
\begin{align*}
 \sum_ja_j\cdot C_j
 &=\sum_kw_k\cdot \left(\sum_{j\le k}a_j\cdot q_j
                         +q_k\cdot \sum_{j>k}a_j\right)\\
 &=\sum_kw_k\cdot \left(1-S_{k+1}+q_k\cdot \sum_{j>k}a_j\right)
 =z_0-\sum_kw_k\cdot \left(S_{k+1}-q_k\cdot \sum_{j>k}a_j\right).
\end{align*}
The second equality uses $a_j\cdot q_j=S_j-S_{j+1}$, including when $q_j=0$.

Substituting into the preceding bound and \eqref{eq:row} gives
\begin{equation*}
 \E[Y]+d\cdot\frac{\E[K]}n\le e^{1/n}\cdot F_q(\theta)\le e^{1/n}\cdot U.
\end{equation*}

Finally, designate all $n$ base variables by
Lemma~\ref{lem:rare-completion}. This preserves identical distributions,
normalizes the expected maximum to $1$, and adds one rare value above
$1$, which the threshold accepts immediately. The new expected reward
tends to $\E[Y]+d\cdot\E[K]/n$.

For fixed $n$ and $\theta$, choose the replacement probability
$\varepsilon>0$ sufficiently small that
\begin{equation*}
\left| \E\!\left[\ALG_\theta\right]-
\left(\mathbb E[Y]+d\cdot\frac{\mathbb E[K]}n\right)\right|
\le \frac{2\cdot e^{1/n}}n.
\end{equation*}
Then $\E\!\left[\ALG_\theta\right]\le e^{1/n}\cdot \left(U+2/n\right)$.
The support has the $N$ possible
base positive values, zero, and the rare value, which proves the claim.
\end{proof}

\begin{lemma}[The scope of the finite dual upper bound]
\label{prop:dual-scope}
Suppose the witness satisfies Lemma~\ref{prop:dual}. For every
continuous and strictly decreasing precommitted threshold, including an
independent random choice of such a threshold, for every sufficiently large
$n$, and every $0<\delta<1$, there is a normalized and finitely supported
distribution satisfying
\begin{equation*}
 \E[\ALG]\le e^{1/n}\cdot \left(U+2/n\right)+\delta.
\end{equation*}
The input is fixed before the threshold's random seed is chosen.

Consequently, the best asymptotic competitive ratio in this class is
at most $U$.
\end{lemma}
\begin{proof}
Apply Lemma~\ref{lem:height-reduction} to the random threshold $\Theta$,
and use Lemma~\ref{lem:random-functions} to obtain a deterministic threshold
$\bar\theta$ dominating the capped version $\Theta_\delta$.
Note that $\bar\theta$ depends only on the law
of $\Theta_\delta$, and is constructed independently of the
input distribution and the realized threshold seed.
We now choose the normalized input supplied by
Lemma~\ref{lem:finite-hard} for $\bar\theta$.
On this chosen distribution,
\begin{equation*}
 \E\!\left[\ALG_\Theta\right]
 \le\E\!\left[\ALG_{\Theta_\delta}\right]+\delta
 \le\E\!\left[\ALG_{\bar\theta}\right]+\delta
 \le e^{1/n}\cdot \left(U+2/n\right)+\delta.
\end{equation*}
Letting $n\to\infty$ and $\delta\to 0^+$ proves the asymptotic claim.
\end{proof}

\subsection{Precommitted thresholds in prophet secretary}
\label{sec:upper-four}

We now consider a deterministic, continuous and strictly decreasing threshold
$\theta$ with $\theta\left(1\right)=0$.
The threshold may depend on $n$, but is fixed
before the observations and their independent uniform arrival times are sampled.
We will use four inputs determined by three heights of $\theta$. Each
input gives a reward bound consisting of a constant and a weighted sum of these heights. A positive weighted
average cancels all three heights and leaves a constant.

\begin{theorem}\label{thm:upper-four}
Let $U=0.694451$ and $\lambda=0.386571$.
For every $n\ge3$, every deterministic threshold in the above class, and
every $\zeta>0$, there is a nonnegative independent input with
$\E\!\left[X^*\right]=1$ such that
\begin{equation*}
 \E\!\left[\ALG\right]\le U+\frac{\lambda}{n-1}+\zeta.
\end{equation*}
Consequently, no such algorithm guarantees a ratio of $0.694451$
uniformly over $n$.
\end{theorem}

\begin{proof}
By Lemma~\ref{lem:height-reduction}, it suffices to prove the claim
under $\theta(0)\le1$, with an arbitrarily small additive error. The
error from reducing the initial height can then be included in $\zeta$.

\paragraph{The input family and its normalization.}
Let $\theta$ be the threshold function.
Fix times $0\le u<s\le1$ and choose values $0\le v\le a\le1$ such that
$v= \theta(s)$, $a = \theta(u)$.
When $u=0$, we also allow $a\ge\theta(0)$.
Then $v$ is eligible after time $s$ and $a$ is eligible after time $u$.

%Denote values $0\le v\le a\le1$ whose eligibility times are $u$ and $s$ with $0\le u<s\le1$, so that $a$ is eligible after time $u$ and $v$ is eligible after time $s$ almost surely.When $u=0$, any $a\ge\theta(0)$ has eligibility time zero.
The input consists of copies of $v$
and one independent variable equal to $a$ with probability $p$, where $0\le p\le1$,
and to $0$ otherwise.
The endpoint $s=1$, $v=0$ requires no copies.

\bigskip

First, consider the ideal limiting case in which the floor value $v$ is
supplied at time $s$ if no eligible $a$ has arrived.
Then $a$ gets accepted when it is present and its arrival falls in $\left(u,s\right)$.
This auxiliary comparison will be approximated by finitely
many independent floor copies, with an explicit error bound below.
Let $Y_\infty$ and $T_\infty$ be the selected value and acceptance time, with
$\left(Y_\infty,T_\infty\right) = \left(0,1\right)$ on failure.
Set $d=s-u$. Let $M$ be the expected maximum of this input with at least one
floor copy.
We have:
\begin{align*}
 M&=p\cdot a+\left(1-p\right)\cdot v = v+p\cdot\left(a-v\right),\\
 \E\!\left[Y_\infty\right]&= p\cdot d\cdot a+\left(1-p\cdot d\right)\cdot v
     =v+p\cdot\left(a-v\right)\cdot d,\\
 \E\!\left[T_\infty\right]
   &=p\cdot\int_u^s t\,dt+\left(1-p\cdot d\right)\cdot s
     =s-\frac{p\cdot d^2}{2}.
\end{align*}

The same expression holds when $s=1$, $v=0$ and the copies
are omitted. In that endpoint case, an unaccepted positive value gives
$Y_\infty=0$ and $T_\infty=1$ as required.

\medskip

Next, we use Lemma~\ref{lem:rare-completion} to normalize an actual input from this family with $M<1$.
Append one zero variable and replace it by $H_\varepsilon=M+\left(1-M\right)/\varepsilon$ with probability $\varepsilon$, which then makes the expected maximum exactly $1$.
The appended variable has an independent uniform arrival time.
Since $\theta(t)>0$ for $t<1$, its zero value is ignored.

Let $Y$ and $T$ denote the selected value and acceptance time from this actual input,
with $\left(Y,T\right)=\left(0,1\right)$ on failure.
Then the probability of arriving before acceptance
is $\E[T]$. Therefore, the normalized reward tends to:
\begin{equation*}
 \E[Y]+\left(1-M\right)\cdot\E[T]
 \qquad \text{when } \varepsilon\to 0^+.
\end{equation*}

Write $h=s-p\cdot d^2/2$. For the limiting comparison, the normalized
accounting reward is:
\begin{align*}
 R\left(u,s,a,v,p\right)
 &=v+p\cdot\left(a-v\right)\cdot d
   +\left(1-v-p\cdot\left(a-v\right)\right)\cdot h\\
 &=h+p\cdot\left(d-h\right)\cdot a
   +\left(1-p\cdot d-\left(1-p\right)\cdot h\right)\cdot v.
\end{align*}
Note that this expression is a constant plus a weighted sum of the values $v$ and $a$.

\paragraph{Four choices and the cancellation.}
Set
\begin{equation*}
 t_1=0.0206,\qquad
 t_2=0.159,\qquad
 t_3=0.4981.
\end{equation*}
Strict decrease gives $1>\theta\left(t_1\right) > \theta\left(t_2\right)>
\theta\left(t_3\right)>0$. Use the following four members
of the input family:
\begin{equation*}
 \begin{array}{c|ccccc}
  j&u&s&\text{positive value}&\text{floor value}&p\\\hline
  0&0&t_1&1&\theta\left(t_1\right)&1/2\\
  1&t_1&t_2&\theta\left(t_1\right)&\theta\left(t_2\right)&1\\
  2&t_2&t_3&\theta\left(t_2\right)&\theta\left(t_3\right)&1\\
  3&t_3&1&\theta\left(t_3\right)&0&1
 \end{array}
\end{equation*}

These are the stated eligibility times. In row zero, value one is
eligible at every positive time because $\theta(0)\le1$. The
expected maxima before normalization are $\left(1+\theta\left(t_1\right)\right)/2$,
$\theta\left(t_1\right)$, $\theta\left(t_2\right)$, and $\theta\left(t_3\right)$,
respectively, and are all less than $1$.

Substituting the four parameter choices into the definition of $R$ gives
the following upper bounds $R_j$ on their limiting accounting rewards:
\begin{align*}
 R_0&=0.020547+0.979453\cdot\theta\left(t_1\right),\\
 R_1&=0.149423-0.011023\cdot\theta\left(t_1\right)
       +0.8616\cdot\theta\left(t_2\right),\\
 R_2&=0.440606-0.101506\cdot\theta\left(t_2\right)
       +0.6609\cdot\theta\left(t_3\right),\\
 R_3&=0.874049-0.372149\cdot\theta\left(t_3\right).
\end{align*}
In each row, the constant coefficient is rounded upward while its sum
with the next coefficient is unchanged. Since the corresponding threshold
height lies in $\left[0,1\right]$, this rounding can only increase the reward.

The dependence forms a chain: $R_0$ has a positive coefficient of
$\theta\left(t_1\right)$; $R_1$ has a negative coefficient of $\theta\left(t_1\right)$ and a positive coefficient
of $\theta\left(t_2\right)$; $R_2$ has a negative coefficient of $\theta\left(t_2\right)$ and a positive
coefficient of $\theta\left(t_3\right)$; $R_3$ has a negative coefficient of $\theta\left(t_3\right)$.

Set $w_3=1$ and choose the remaining weights backward:
\begin{equation*}
 w_2=\frac{0.372149}{0.6609},\qquad
 w_1=w_2\cdot\frac{0.101506}{0.8616},\qquad
 w_0=w_1\cdot\frac{0.011023}{0.979453}.
\end{equation*}

All weights are positive. The first equation makes the coefficient of $\theta\left(t_3\right)$ in
$w_2\cdot R_2+w_3\cdot R_3$ zero. The second then cancels $\theta\left(t_2\right)$, and the third
cancels $\theta\left(t_1\right)$. Let $W=\sum_{j=0}^3 w_j$. Direct substitution gives
\begin{equation}\label{eq:secretary-weight-cancellation}
 \frac1W\cdot\sum_{j=0}^3 w_j\cdot R_j<U,
 \qquad
 \frac{w_0+w_1+w_2}{W}<\lambda.
\end{equation}
Since the weights are positive, the first inequality shows that some limiting
reward is strictly less than $U$, whatever the three threshold heights are.

\paragraph{Realization with $n$ variables.}
For rows zero, one, and two, use $m=n-2\ge1$ floor copies, the one
positive-value variable, and the appended rare-value variable.

Let $Z_m$
be the first floor arrival after $s$, taking $Z_m=1$ if every copy
arrives before $s$, and let $\Delta=Z_m-s$. For $s\le t<1$, the
event $Z_m>t$ means that no floor copy arrives in $\left(s,t\right]$.
We have:
\begin{equation*}
 \E\!\left[\Delta\right]
 =\int_s^1\Pr\left(Z_m>t\right)\,dt
 =\int_s^1\left(1-t+s\right)^m\,dt
 =\frac{1-s^{m+1}}{m+1}
 \le\frac1{m+1}.
\end{equation*}

Couple this finite input with the limiting comparison using the same
positive-value variable and its arrival time. If the positive value
arrives in $\left(u,s\right)$, both rules accept it at the same time. Otherwise, the finite rule either accepts a floor copy by $s+\Delta$, accepts the positive value in $\left(s,s+\Delta\right)$, or fails when no floor
copy is eligible. In all cases, with $T_m=1$ on failure, $T_m\le T_\infty+\Delta$.
Its selected value exceeds $Y_\infty$ only if the positive value
arrives in $\left(s,s+\Delta\right)$, and the excess is then at most
$a-v$, with $a,v$ denoting the values in the relevant row.
The arrival time of that variable is independent of $\Delta$. Consequently,
\begin{equation*}
 \E\!\left[Y_m\right]\le\E\!\left[Y_\infty\right]+p\cdot\left(a-v\right)\cdot\E\!\left[\Delta\right],\qquad
 \E\!\left[T_m\right]\le\E\!\left[T_\infty\right]+\E\!\left[\Delta\right].
\end{equation*}

Since $M=v+p\cdot\left(a-v\right)$, we have
\begin{align*}
 \E\!\left[Y_m\right]+\left(1-M\right)\cdot\E\!\left[T_m\right]
 &\le R_j+\left(p\cdot\left(a-v\right)+1-M\right)\cdot\E\!\left[\Delta\right]\\
 &=R_j+\left(1-v\right)\cdot\E\!\left[\Delta\right]
 \le R_j+\frac1{n-1}.
\end{align*}

For row three, no floor copies are needed, so its accounting reward is
at most $R_3$. Append zero values to bring its size to $n$.

Finally, choose the rare-value probabilities sufficiently small. For
any $\eta>0$, the normalized finite rewards $R_j^{\left(n\right)}$ then satisfy $R_j^{\left(n\right)}\le R_j+\frac{\mathbf 1_{\left\{j<3\right\}}}{n-1}+\eta$.
Combining this bound with~\eqref{eq:secretary-weight-cancellation}, we get
\begin{equation*}
 \min_j\left\{R_j^{\left(n\right)}\right\}
 \le\frac1W\cdot\sum_{j=0}^3 w_j\cdot R_j^{\left(n\right)}
 <U+\frac\lambda{n-1}+\eta.
\end{equation*}

If an initial-height reduction was used, choose its error and $\eta$ so that their sum is
less than $\zeta$. This proves the claimed finite-$n$ bound.
Moreover, the first weighted sum in~\eqref{eq:secretary-weight-cancellation}
is a constant strictly less than $U$. Taking $n$ sufficiently large and then
$\zeta$ sufficiently small therefore gives an input with reward strictly
less than $U$, proving the final claim.
\end{proof}

\end{document}